\documentclass[prx,superscriptaddress,twocolumn,aps,nobibnotes,nofootinbib,amssymb,floatfix,a4paper,longbibliography]{revtex4-2}

\usepackage{adjustbox,amsfonts,amsmath,amssymb,amsthm,array,bbm,bm,booktabs,braket,cancel,centernot,enumerate,enumitem,fancyhdr,graphicx,hyperref,mathdots,mathtools,mdwlist,multirow,orcidlink,pgfplots,soul,subfigure,thmtools,tikz,tikz-qtree}
\usepackage[normalem]{ulem}
\usetikzlibrary{chains,quantikz2}
\hypersetup{colorlinks=true,linkcolor=blue,citecolor=blue,urlcolor=blue}
\pgfplotsset{compat=1.18} 
\allowdisplaybreaks

\newlength{\qrow}
\newlength{\qcol}
\usepackage{etoolbox}

\providecommand{\bibfield}[2]{#2}
\renewcommand{\bibfield}[2]{%
  \ifstrequal{#1}{title}{\textit{#2}}{#2}%
}

\theoremstyle{plain}
\newtheorem{thm}{Theorem}
\newtheorem{lemma}[thm]{Lemma}

\newcommand{\ketbra}[2]   {\left|#1\middle\rangle\!\middle\langle#2\right|}
\newcommand{\ot}{\otimes}
\newcommand\C{ {\mathcal C} }
\renewcommand\S{ {\mathcal S} }
\newcommand\N{ {\mathcal N} }
\newcommand\F{ {\mathcal F} }
\newcommand{\mx}{\mathrm{max}}
\newcommand{\eq}{\mathrm{eq}}

\renewcommand{\v}[1]{\ensuremath{\boldsymbol #1}}
\newcommand{\re}{\mathrm{Re}}
\newcommand{\ma}{\mathrm{Ma}}
\newcommand{\add}{\mathrm{ADD}}
\newcommand{\taus}{\bar{\tau}^\star}
\newcommand\dbe{ {\mathrm{DBE}} }

\newenvironment{aligns}{\subequations \align} {\endalign \endsubequations}

\begin{document}


\title{An end-to-end quantum algorithm for nonlinear fluid dynamics \\with bounded quantum advantage}

\author{David Jennings\,\orcidlink{0000-0003-1201-3725}}
\affiliation{PsiQuantum, 700 Hansen Way, Palo Alto, CA 94304, USA}

\author{Kamil Korzekwa$^*$\orcidlink{0000-0002-0683-5469}}
\affiliation{PsiQuantum, 700 Hansen Way, Palo Alto, CA 94304, USA}

\renewcommand{\thefootnote}{\fnsymbol{footnote}} 
\footnotetext[1]{Lead author email: kkorzekwa@psiquantum.com (authors are listed alphabetically within each affiliation).}
\renewcommand{\thefootnote}{\arabic{footnote}}

\author{Matteo Lostaglio\orcidlink{0000-0003-1942-2417}}
\affiliation{PsiQuantum, 700 Hansen Way, Palo Alto, CA 94304, USA}

\author{Richard Ashworth\orcidlink{0000-0003-4446-9559}}
\affiliation{Airbus Operations Ltd, Pegasus House Aerospace Avenue, Filton, Bristol, UK}

\author{Emanuele Marsili\orcidlink{0000-0003-0341-3241}}
\affiliation{Airbus Operations Ltd, Pegasus House Aerospace Avenue, Filton, Bristol, UK}

\author{Stephen Rolston\orcidlink{0009-0000-1524-986X}}
\affiliation{Airbus Operations Ltd, Pegasus House Aerospace Avenue, Filton, Bristol, UK}

\begin{abstract}
    Computational fluid dynamics (CFD) is a cornerstone of classical scientific computing, and there is growing interest in whether quantum computers can accelerate such simulations. The existing proposals for fault-tolerant quantum algorithms for CFD have almost exclusively been based on the Carleman embedding method, used to encode nonlinearities on a quantum computer. Here, we first identify several key challenges that must be addressed in order to establish quantum advantage using this method, including Carleman convergence, time-stepping requirements, condition-number scaling, and data extraction. Building on this analysis, we develop a novel algorithm for the incompressible lattice Boltzmann equation that addresses them, and provide its detailed analysis, including all potential sources of algorithmic complexity, as well as gate count estimates. We find that for an end-to-end problem, a modest quantum advantage may be preserved for selected observables in the high-error-tolerance regime. We establish a lower bound for the Reynolds number scaling of our quantum algorithm in dimension $D$ at Kolmogorov microscale resolution with $O(\re^{\frac{3}{4}(1+\frac{D}{2})} \times q_M)$, where $q_M$ is a multiplicative overhead for data extraction with $q_M = O(\re^{\frac{3}{8}})$ for the drag force. This upper bounds the scaling improvement over classical algorithms by $O(\re^{\frac{3D}{8}})$. However, our numerical investigations suggest a lower speedup, with a scaling estimate of $O(\re^{1.936} \times q_M)$ for $D=2$. Our results give robust evidence that small, but nontrivial, quantum advantages can be achieved in the context of CFD, and motivate the need for additional rigorous end-to-end quantum algorithm development.
\end{abstract}

\maketitle


\section{Introduction}

The development of classical computers, and associated numerical methods, revolutionized engineering by making it possible to numerically solve formerly intractable differential equations, particularly in fluid dynamics. While early methods focused on special cases such as two-dimensional flows, modern high-performance computers allow one to solve highly-detailed, three-dimensional Navier-Stokes equations~\cite{stefanin2024aircraft,patel2024assessing}. Nevertheless, accurately simulating highly turbulent flows with direct numerical simulation (DNS) -- relevant for many industrial applications, but of particular relevance in aircraft design -- is often computationally infeasible~\cite{spalart2000strategies}. Leading practitioners hence resort to less accurate turbulence models such as Reynolds-averaged Navier-Stokes or large eddy simulations~\cite{Pope_2000}. More recently, parallelization approaches using GPUs, as well as machine learning models, have been developed~\cite{appa2021performance}. Nonetheless, access to large-scale DNS solutions with the fine resolution required across vast lengths to model all scales in their complexity is not expected to be possible anytime soon.

A new paradigm of quantum computing offers a way to break such classical computational barriers. On a mathematical level, it enables efficient implementation of unitary operations on exponentially large Hilbert spaces. While the restriction to unitary operations might seem limiting at first sight, block-encoding techniques can be used to encode general matrices inside blocks of the unitary ones~\cite{lin2022lecture}. When these block-encodings are combined with quantum algorithmic techniques, such as linear combination of unitaries or quantum signal processing~\cite{gilyen2019quantum}, they allow one to perform linear algebra operations on quantum computers. It then becomes possible to solve large linear systems exponentially faster than using classical computers, subject to important caveats on the condition number, as well as data loading and readout~\cite{aaronson2015read-f44}. Correspondingly, discretized linear differential equations can be solved via quantum linear solvers (QLSs)~\cite{harrow2009quantum-5fb,costa2022optimal,jennings2023efficient, dalzell2024shortcut, morales2024quantum} as a single linear system~\cite{berry2014high, montanaro2016quantum, berry2017quantum, costa2019quantum, linden2022quantum, berry2024quantum, ameri2023quantum, jennings2024cost}, 
where the solution is stored coherently in a quantum state and needs to be post-processed in a read-out step to extract classical information about the solution. It is expected that quantum computers will offer exponential advantage for the simulation of quantum systems~\cite{feynman2018simulating, lloyd1996universal, berry2007efficient, aspuru2005simulated, mcardle2020quantum}. Moreover, it has been shown that there exist classical linear dynamical systems for which the same speedup persists~\cite{babbush2023exponential}. 

However, industrially relevant computational fluid dynamics (CFD) problems -- which account for a significant portion of modern high-performance computing workloads -- are inherently nonlinear. As computations on quantum computers are fundamentally linear, leveraging them for CFD simulations requires reformulating the governing equations using a linear representation. One prominent method is Carleman linearization, where the finite-dimensional nonlinear dynamics is transformed to infinite-dimensional linear dynamics, and then truncated to a very high, but finite-dimensional, linear system~\cite{carleman1932application, forets2017explicit}. This method has been recently proposed and studied as a means of developing quantum algorithms for solving nonlinear differential equations~\cite{liu2021efficient, an2022efficient,krovi2023improved, costa2025further, wu2025quantum, jennings2025quantum}. In the CFD context, this path has been explored by the authors of Refs.~\cite{itani2022analysis,li2025potential,penuel2024feasibility,turro2025practical,sanavio2024lattice,itani2024quantum,sanavio2025carleman}, who investigated quantum algorithms based on Carleman linearization to simulate the lattice Boltzmann equation (LBE)~\cite{kruger2016lattice}.\footnote{For other quantum approaches to CFD problems see Refs.~\cite{succi2023quantum,tennie2025quantum,meng2025toward,bharadwaj2023hybrid} or Sec.~I of Ref.~\cite{itani2024quantum}.} The LBE is a kinetic model of fluid dynamics that describes fluids mesoscopically and recovers the Navier–Stokes equations when characteristic length scales are much larger than the molecular mean free path.

In this work, we first analyze these recent QLS-based approaches employing Carleman embedding, identifying several algorithmic challenges that must be carefully addressed in order to establish quantum advantage. We then develop a new quantum algorithm specifically designed to address these challenges. Our approach is based on a shifted version of the incompressible lattice Boltzmann equation~\cite{he1997lattice,guo2000lattice} that is linearized using a discrete Carleman embedding method~\cite{berkolaiko1998analysis,pruekprasert2020moment,pruekprasert2024moment} and solved using a quantum linear solver~\cite{dalzell2024shortcut}. Through a careful analysis of query and gate complexities, supported by extensive numerical investigations (complemented by Ref.~\cite{jennings2026simulating}), we show that our algorithm can achieve a modest quantum advantage. However, we also highlight the challenges that algorithms based on Carleman embedding face when applied to industrially relevant problems, as the convergence of the method may be limited to low–Reynolds-number CFD simulations.

The paper is organized as follows. In Sec.~\ref{sec:summary}, we first provide a broad overview of our main results, highlighting our main contributions and crucial obstacles that we overcome. Then, in Sec.~\ref{sec:background}, we present the necessary background material on the lattice Boltzmann equation and recent quantum algorithms for solving it. In Sec.~\ref{sec:bottlenecks}, we identify several aspects of these approaches that require careful treatment when assessing their computational complexity and potential quantum advantage. Next, in Sec.~\ref{sec:problem_formulation}, we give the details of our approach to the LBE problem that is designed to address the identified challenges, formulating a linear system of equations whose solution describes the LBE dynamics. With the LBE problem specified, in Sec.~\ref{sec:algorithm} we present the details of our quantum algorithm for solving it and extracting the information about the drag force. Then, in Sec.~\ref{sec:performance}, we analyze the performance of the algorithm, analyzing its accuracy and estimating its query and gate complexity costs. In Sec.~\ref{sec:conclusions}, we discuss the potential for quantum advantage and present conclusions. Finally, some of the technical details of our work can be found in appendices, in particular Appendix~\ref{app:notation} contains the summary of the notation used in the paper. 


\section{Summary of results}
\label{sec:summary}

The results presented in this work can be grouped into the following four categories:
\begin{enumerate}
    \item \textbf{Identification of complexity challenges.} We identify key complexity challenges that must be addressed in order to establish quantum advantage for simulating the LBE.
    \item \textbf{Classical problem reformulation.} We develop a reformulation of the standard LBE and its linear embedding that addresses these challenges. 
    \item \textbf{Quantum algorithm implementation.} We design an end-to-end quantum algorithm for simulating LBE: from the data encoding of the problem, through the explicit construction of novel block-encodings allowing one to prepare a coherent output describing the steady-state of the fluid system, and finally to the data extraction subroutine enabling one to get information about the drag force through measurements of the coherent output.
    \item \textbf{Quantum algorithm performance.} We provide detailed resource estimates for running our algorithm and an analysis of the Carleman embedding error, indicating the potential for a modest polynomial quantum advantage and identifying the parameter region of applicability.
\end{enumerate}

Our first contribution is the critical assessment of the recent QLS-based quantum algorithmic proposals employing Carleman embedding to simulate the dynamics described by the lattice Boltzmann equation~\cite{li2025potential,penuel2024feasibility,turro2025practical}. These proposals are all based on the same broad quantum algorithmic pipeline, see Fig.~\ref{fig:scheme}, where the crucial step involves mapping the nonlinear fluid equations to a much larger set of linear ones via Carleman linearization, which is also studied in Refs.~\cite{sanavio2024lattice,itani2024quantum,sanavio2025carleman}\footnote{We focus on the QLS-based route, represented by Refs.~\cite{li2025potential,penuel2024feasibility,turro2025practical}, because it is the setting in which end-to-end complexity cost has been analyzed. By contrast, Refs.~\cite{sanavio2024lattice,itani2024quantum,sanavio2025carleman} focus on Carleman-embedding for a single time-step of the LBE evolution, without addressing its integration into a full time-evolution algorithm.}. We identify the following algorithmic challenges arising at different stages of this pipeline, and analyze their impact on the overall computational complexity and the achievable quantum advantage:

\begin{figure}[t]
    \centering
    \includegraphics[width=\linewidth]{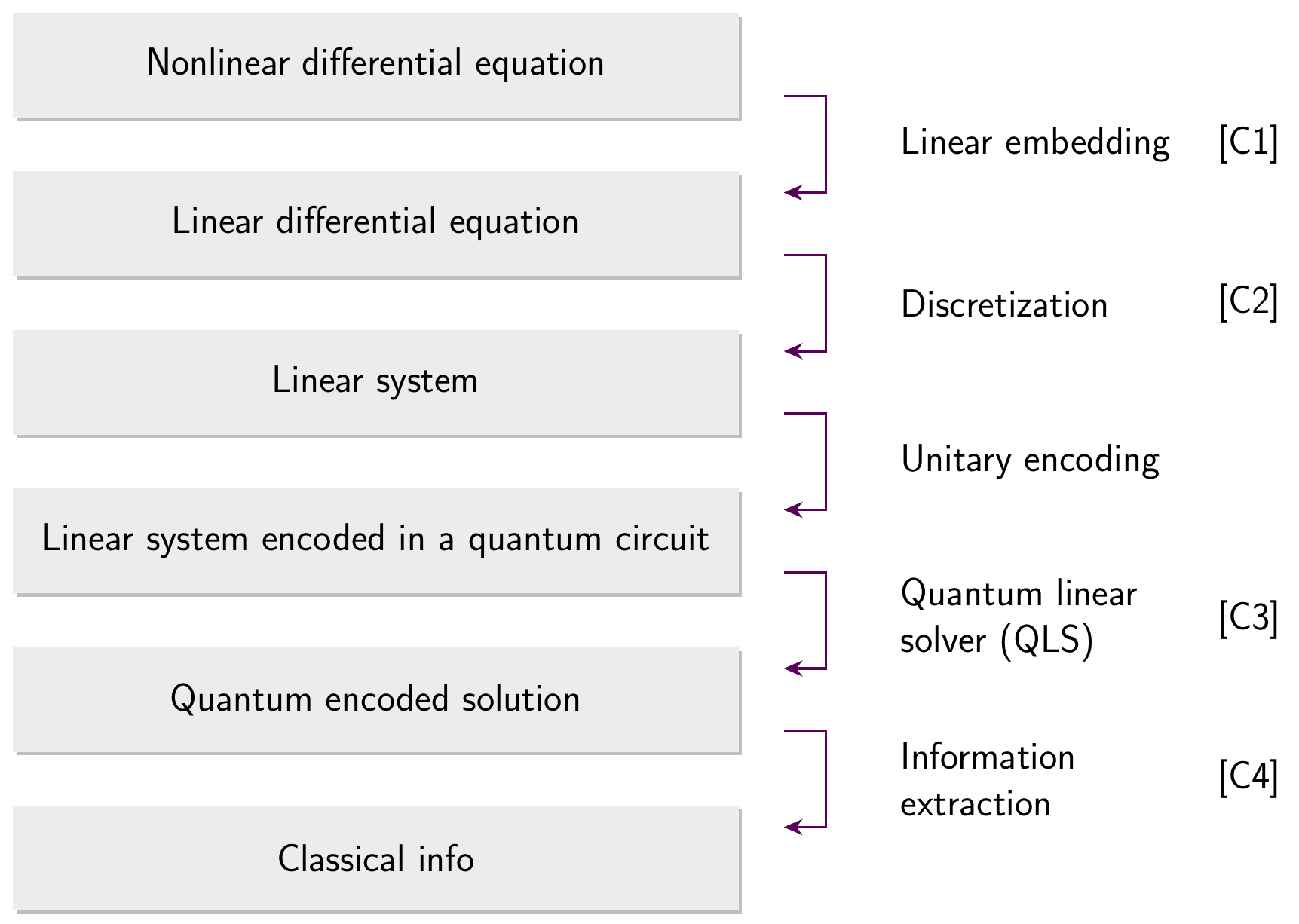}
    \caption{\textbf{Pipeline and challenges.} The QLS-based algorithmic pipeline for Carleman-linearized dynamics commonly used in the literature, together with the challenges C1-C4 that we identify in Sec.~\ref{sec:bottlenecks}.}
    \label{fig:scheme}
\end{figure}

\begin{enumerate}
    \item[C1.] \emph{Applicability of the Carleman approach.} The Carleman linear embedding introduces errors, depending on the core Carleman truncation parameter $N_C$. Prior studies either do not properly incorporate the full LBE dynamics~\cite{li2025potential}, do not look into the question of the linear embedding errors~\cite{penuel2024feasibility}, or find no convergence~\cite{turro2025practical}. The available evidence on the performance of the method is therefore limited~\cite{sanavio2024lattice}.

    \item[C2.] \emph{Small time steps.} The lattice Boltzmann method typically uses an implicit second-order accurate time discretization~\cite{kruger2016lattice, guo2013lattice}. This is advantageous as, by treating the stiff collision term implicitly, the time-step size $\Delta t$ scaling with the inverse Reynolds number is alleviated. However, prior proposals based on the discrete Boltzmann equation (DBE)~\cite{li2025potential, penuel2024feasibility,turro2025practical}, employ an explicit first-order accurate discretization of the collision integral (which restricts $\Delta t <  2\tau$~\cite{ubertini2006unstructured, ubertini2010three}), and an explicit forward-time central-space scheme of the streaming term rather than integrating it exactly along characteristics. This formulation therefore forces the time-step to scale with the inverse Reynolds number and also requires satisfying the CFL condition due to the non-exact treatment of streaming.
    
    \item[C3.] \emph{Condition number estimates.} Optimal quantum linear solvers (QLSs) scale linearly with the condition number $\kappa_A$ of the linear system $A$ being solved. The scaling of $\kappa_A$ with the number of spatial discretization points $N$ can potentially entirely negate quantum advantage. Prior works either do not estimate~$\kappa_A$~\cite{penuel2024feasibility}, present arguments for the behavior of~$\kappa_A$ that contain errors~\cite{li2025potential}, or suggest using the suboptimal HHL algorithm~\cite{turro2025practical}.

    \item[C4.] \emph{Inefficient data extraction.} We identify a previously overlooked inefficiency (related to the norm of the LBE state vector) in the data extraction schemes proposed in Refs.~\cite{li2025potential, penuel2024feasibility} that must be addressed in order to retain the potential for quantum advantage.
\end{enumerate}

Our second contribution is the classical problem reformulation allowing us to overcome the identified challenges C2 and C4. We first employ a \emph{shifted incompressible lattice Boltzmann equation} which, for the price of restricting to incompressible flows, allows us to represent the fluid using a state vector whose norm scales with the magnitude of the lattice velocity. Combining this with a tailored choice of simulation parameters that forces the lattice velocity to be small enough, we control the norm of the state vector to be constant, i.e., it does not scale with the system size. This way we mitigate the problem of inefficient data extraction identified in challenge C4. We also employ a \emph{discrete Carleman embedding}~\cite{berkolaiko1998analysis,pruekprasert2020moment,pruekprasert2024moment}, previously investigated in a restricted manner in the LBE context~\cite{sanavio2024lattice}, which allows us to approximate a nonlinear discrete-time evolution by embedding it into a discrete-time linear evolution of an enlarged system. By leaving time discretization on the classical side and directly embedding the second-order accurate LBE (instead of embedding the underlying time-continuous equation), we avoid the issue with small time steps. As a result of the classical problem reformulation, we end up with a linear system $A\v{Y} = \v{b}$, whose solution carries information about the fluid system during the LBE evolution.

Our third contribution is a detailed design of a quantum algorithm for simulating LBE through solving this linear system. This includes the description of an efficient encoding of the Carleman embedded LBE data into quantum registers, and a construction of the unitary block-encoding of the Carleman linear system matrix $A$ together with unitary state preparations required by a quantum linear solver to produce the coherent encoding of the final steady-state of the fluid. In doing so, we develop novel block-encodings for discrete Carleman matrices based on higher-order singular value decomposition for rectangular matrices describing quadratic nonlinearities in the LBE collision step. We then present the details concerning the query complexity $q_Q$ of running the quantum linear solver algorithm for our problem, which scales roughly as $q_Q=\tilde{O}(\alpha_A \kappa_A / \|A\|)$, where $\alpha_A$ is the block-encoding prefactor of~$A$ and $\kappa_A$ is the condition number of~$A$. Finally, we explain how to perform measurements to extract classical information about the drag force. This extraction procedure brings additional multiplicative factor overhead $q_M$, scaling as $O(\re^{3/8})$, to the complexity of our quantum algorithm.

Finally, our fourth and central contribution, is a careful analysis of the performance of the introduced algorithm, which overcomes challenge C3 and addresses C1. Broadly speaking, we answer two main questions: how well our algorithm reproduces the lattice Boltzmann dynamics and at what complexity cost. First, we numerically investigate the applicability of the Carleman approach for the lattice Boltzmann equation, thus addressing the challenge C1. We perform extensive numerical simulations for the D1Q3 model with Carleman truncation orders $N_C\in\{1,\dots,5\}$ and Reynolds numbers $\re\in[10,1000]$; and for D2Q9 model with $N_C\in\{1,\dots,3\}$ and $\re\in [10,250]$. We assume periodic boundary conditions, no walls, and initial states given by one-dimensional sinusoidal and colliding states, as well as two-dimensional Taylor-Green vortices and Gaussian vortex dipoles. We find that for low enough Reynolds numbers, the Carleman truncation error $\epsilon_C$ indeed converges, in the sense that increasing $N_C$ decreases $\epsilon_C$. However, we also observe the existence of a threshold Reynolds number $\re_T$, such that the Carleman procedure does not converge for $\re > \re_T$, i.e., increasing truncation order $N_C$ increases the error $\epsilon_C$ instead of decreasing it. For one-dimensional systems, this threshold is relatively low, $\re_T\sim 10^2$; for two-dimensional systems it is higher by an order of magnitude or more, and we expect it to be even higher for three-dimensional systems.

To address the challenge C3, we perform a careful assessment of the complexity cost of running our algorithm, which can be written as a product of three terms:
\begin{equation}
    \mathrm{Complexity}_Q = q_Q \times  g \times q_M.
\end{equation}
In the above, $q_Q$ is the query complexity of running the QLS, $g$ is the number of elementary $T$-gates needed to implement a single query, and $q_M$ is the overhead related to repeating the measurements in order to extract the classical information from coherently encoded final state. This should be compared with the runtime of classical algorithms for $D$-dimensional fluid systems that scales as
\begin{equation}
    \mathrm{Complexity}_C = O(NT^\star) = O(\re^{\frac{3}{4}(D+1)}),
\end{equation}
where we recall that $N$ is the total number of spatial discretization points and $T^\star$ denotes the number of discrete time steps of the simulated evolution. We estimate the $T$-gate cost $g$ of a single query, and prove that for a fixed $N_C$ this just brings a constant overhead over the query cost $q_Q$. We also prove that the information extraction cost $q_M$ scales as $O(\re^{3/8})$ for the drag force, due to the boundary nature of this observable, but expect $O(1)$ scaling for certain bulk observables. Thus, the dominant component of the quantum complexity cost comes from the query complexity $q_Q$, which itself is governed primarily by the condition number $\kappa_A$. We then derive a lower bound for the condition number $\kappa_A$, proving that at best one can hope for $\kappa_A=O(\sqrt{N}T^\star)$, which shows that one can expect at most a quadratic improvement in spatial degrees of freedom. We also perform extensive numerical simulations, using the Lanczos method for Carleman matrices with $N_C\in\{1,\dots,4\}$ and of dimension up to $\sim 10^8$, to estimate the condition number for one- and two-dimensional fluid systems.\footnote{While sizes corresponding to Carleman systems for high-impact uses-cases would be much larger, and beyond the reach of classical hardware for the foreseeable future, these are -- to our knowledge -- the largest numerical simulations ever done to establish the scaling of Carleman-based algorithms for fluids (and in fact, for Carleman-based algorithms more generally).} We find that the condition number $\kappa_A$, so also the query complexity $q_Q$, scales as 
\begin{equation}
    q_Q = O(\re^\chi),
\end{equation}
with $\chi$ growing with both the problem dimension $D$ and the Carleman truncation order $N_C$. In particular, for $D=2$ and $N_C=2$, we estimate $\chi$ to be $1.936$, which results in a lower than classical scaling of $\re^{2.25}$ for extracting bulk observables (and slightly higher scaling of $\re^{2.311}$ for extracting the drag force). We thus conclude that while modest quantum advantage could be expected for two-dimensional fluid systems (and, by extrapolation, for three-dimensional ones), it will be limited to low $N_C$, and so to low Reynolds numbers or high error outputs.


\section{Background}
\label{sec:background}

Throughout the paper, we set the base of $\log$ to 2, and when we use the natural logarithm, we denote it by $\ln$. We will also use the Dirac ket notation, e.g., $\ket{\v{Y}}$ will denote the normalized quantum state encoding the corresponding vector $\v{Y}$. 


\subsection{Lattice Boltzmann equation}
\label{sec:LBE}

 Within LBE, the fluid is modeled at time $t$ and position $\v{r}$ in $D$-dimensional space by a vector-valued function $\v{f}(\v{r},t)$, whose components $f_m(\v{r},t)$, with $m\in\{1,\dots,Q\}$, are proportional to the probability density of finding a fluid particle with a discrete velocity $\v{e}_m$. The local mass density $\rho$ and momentum density $\rho\v{u}$ of the fluid are then given by:
\begin{equation}
    \rho(\v{r},t)=\sum_{m=1}^Q f_m(\v{r},t),\quad
    \rho\v{u}(\v{r},t)=\sum_{m=1}^Q f_m(\v{r},t)\v{e}_m.
\end{equation}
LBE is derived from the \emph{discrete Boltzmann equation} (DBE) describing continuous dynamics of $\v{f}$ via~\cite{succi2001lattice, kruger2016lattice}
\begin{align}
    \label{eq:DBE}
    \frac{\partial f_m(\v{r},t)}{\partial t} +\v{e}_m\cdot \nabla f_m(\v{r},t) = \Omega_m(\v{r},t),
\end{align}
where $\v{e}_m\cdot \nabla f_m$ is an advection term that describes \emph{streaming}, i.e., the flow of the fluid in space, while $\Omega$ describes \emph{collisions} that are local in space. The simplest choice of $\Omega$ that we will adopt is the Bhatnagar-Gross-Krook (BGK) collision operator given by~\cite{bhatnagar1954model}
\begin{equation}
    \label{eq:BGK}
    \Omega_m(\v{r},t) = -\frac{1}{ \tau} \left[f_m(\v{r},t)-f_m^\eq(\v{r},t)\right],
\end{equation}
where $\tau$ is the relaxation time determining the speed of equilibration, and $\v{f}^\eq$ is the Taylor expansion of the local Maxwell equilibrium up to second order terms: 
\begin{align}
    \label{eq:eq}
    f_m^\eq(\v{r},t)&=  \rho(\v{r},t) \left(1 + \frac{\v{e}_m\cdot \v{u}(\v{r},t)}{c_s^2} + \frac{(\v{e}_m\cdot \v{u}(\v{r},t))^2}{2c_s^4}\right.\nonumber\\
    &\quad\qquad\qquad\left.-\frac{|\v{u}(\v{r},t)|^2}{2c_s^2}\right)w_m + O(\ma^3).
\end{align}
In the above, $c_s$ denotes the speed of sound, $\ma:=|\v{u}|/c_s$ is the Mach number, and $\v{w}$ is a vector of Maxwell weightings for the different discrete velocities satisfying $w_m\geq 0$ and \mbox{$\sum_m w_m=1$}. These weightings are chosen to ensure the lowest-order velocity moments are reproduced and conservation laws are respected. Since the error in the above second-order expansion of the equilibrium density is $O(\ma^3)$, one assumes a small Mach number, typically of order $\ma\sim 10^{-2}$. 

\begin{table*}[t!]
\centering
\begin{tabular}{|c|lll|}
\hline
\multirow{2}{*}{D1Q3} & $m=1$:&$\v{e}_m^\star=0$,& $w_m=2/3$, \\
&$m\in\{2,3\}$:&$\v{e}_m^\star\in\{\pm 1\}$,&$w_m=1/6$.\\\hline
\multirow{3}{*}{D2Q9} & $m=1$:& $\v{e}_m^\star=(0,0)$,&$w_m=4/9$, \\
& $m\in\{2,\dots,5\}$: &$\v{e}_m^\star\in\{(\pm 1,0), (0,\pm 1)\}$,&$w_m=1/9$,\\
& $m\in\{6,\dots,9\}$:& $\v{e}_m^\star\in\{(\pm 1, \pm 1)\}$,& $w_m=1/36$.\\ \hline
\multirow{4}{*}{D3Q27} & $m=1$:& $\v{e}_m^\star=(0,0,0)$,& $w_m=8/27$, \\
& $m\in\{2,\dots,7\}$: & $\v{e}_m^\star\in\{(\pm 1,0,0), (0,\pm 1,0),(0,0,\pm 1)\}$,& $w_m=2/27$,\\
& $m\in\{8,\dots,19\}$:& $\v{e}_m^\star\in\{(\pm 1, \pm 1,0),(\pm 1, 0,\pm 1),(0,\pm 1, \pm 1)\}$,& $w_m=1/54$,\\
&$m\in\{20,\dots,27\}$:& $\v{e}_m^\star\in\{(\pm 1, \pm 1,\pm 1)\}$,& $w_m=1/216$.\\\hline
\end{tabular}
\caption{Discrete velocities in lattice units for standard LBE models~\cite{kruger2016lattice}. The discrete weightings $w_m$ are chosen so as to reproduce low-order velocity moments of the Maxwell-Boltzmann distribution. }
\label{tab:velocities}
\end{table*}

In order to address the DBE with numerical methods, one needs to discretize space and time. In \emph{lattice Boltzmann methods}, spatial discretization is achieved by choosing a uniform and regular lattice with spacing $\Delta x$ and a discrete velocity set $\{\v{e}_m\}$ whose components are integer multiples of $\Delta x  / \Delta t$, where $\Delta t$ is the duration of one discrete time step. Now, using the method of characteristics, one can integrate the DBE over one time step to obtain
\begin{equation}
    \label{eq:DBE_integrated}
    f_m(\v{r}+\v{e}_m\Delta t, t+\Delta t)\!=\!f_m(\v{r},t) +\!\int_0^{\Delta t} \!\!\!\!\Omega_m(\v{r}+\v{e}_m s,t+s) ds.
\end{equation}
Using the first-order (rectangular) discretization to approximate the integral on the right hand side, 
\begin{equation}
    \label{eq:discretization1}
    \int_0^{\Delta t} \Omega_m(\v{r}+\v{e}_m s,t+s) ds = \Delta t \Omega_m(\v{r},t)  + O(\Delta t^2),
\end{equation}
one arrives at the original \emph{lattice Boltzmann equation} (LBE):
\begin{equation}    
    \label{eq:LBE_1st}
    f_m(\v{r}+\v{e}_m\Delta t, t+\Delta t) = f_m(\v{r},t) -\frac{\Delta t}{\tau}  \left[f_m(\v{r},t)\!-\!f^\eq_m(\v{r},t)\right]\!.
\end{equation}
The above, however, is only first-order accurate in the step size $\Delta t$, and so requires very small $\Delta t< 2\tau$ to accurately simulate the fluid dynamics. This in turn translates into a large number of time steps and usually an inefficient simulation (we will discuss this in more detail in Sec.~\ref{sec:problem_timestep}).

To overcome this first-order accuracy issue, one uses the implicit second-order (trapezoidal) discretization to approximate the considered integral with
\begin{align}
    \label{eq:discretization2}
    \!\!\int_0^{\Delta t} \!\!\!\Omega_m(\v{r}+\v{e}_m s,t+s) ds =& \frac{\Delta t}{2}\left[\Omega_m(\v{r}+\v{e}_m \Delta t,t+\Delta t)\right.\nonumber\\
    &\left.+ \Omega_m(\v{r},t)\right] + O(\Delta t^3).
\end{align}
One then follows the original idea of Refs.~\cite{he1998discrete,he1998novel}, and introduces a new set of variables
\begin{equation}
    \label{eq:fbar}
    \bar{f}_m(\v{r},t) = f_m(\v{r},t)  -\frac{\Delta t}{2}\Omega_m(\v{r},t),
\end{equation}
as well as the shifted relaxation time
\begin{equation}
\label{eq:shiftedtau}
    \bar{\tau} = \tau+\frac{\Delta t}{2}.
\end{equation}
Using these together with Eq.~\eqref{eq:discretization2} in Eq.~\eqref{eq:DBE_integrated}, one arrives at the commonly used lattice Boltzmann equation
\begin{equation}
    \label{eq:LBE_2nd}
    \bar{f}_m(\v{r}+\v{e}_m\Delta t,t+\Delta t) = \bar{f}_m(\v{r},t) - \frac{\Delta t}{\bar{\tau}}[\bar{f}_m(\v{r},t)-f^\eq_m(\v{r},t)],
\end{equation}
which is second-order accurate and does not restrict one to $\Delta t < 2\tau$. Although in this approach one follows the transformed variables $\bar{f}_m$ rather than the original ones $f_m$, the macroscopic variables $\rho$ and $\v{u}$ stay unchanged, as
\begin{equation}
    \sum_{m=1}^Q \bar{f}_m = \sum_{m=1}^Q f_m,\qquad \sum_{m=1}^Q \bar{f}_m\v{e}_m = \sum_{m=1}^Q f_m \v{e}_m.  
\end{equation}

Next, ``lattice units'' are used to represent all physical parameters by dimensionless numbers. Denoting physical quantities in lattice units by a $\star$ superscript, the time steps are chosen to be of unit size, $\Delta t^\star=1$. Moreover, for the spatial discretization to also be of unit size, $\Delta x^\star=1$, the velocities $\v{e}_m^\star$ are chosen to have unit entries. The typical choice of discrete velocities for one-, two- and three-dimensional problems is given by the D1Q3, D2Q9 and D3Q27 models, where the velocities $\v{e}_m^\star$ and their corresponding weights $w_m$ are given in Table~\ref{tab:velocities}. 

Using lattice units, the first-order and second-order accurate lattice Boltzmann equations, Eqs.~\eqref{eq:LBE_1st} and \eqref{eq:LBE_2nd}, read
\begin{aligns}
    \label{eq:LBE_lat}
    f_m(\v{r}^\star+\v{e}^\star_m,t^\star+1) =& - \frac{1}{\tau^\star}[f_m(\v{r}^\star,t^\star)-f^\eq_m(\v{r}^\star,t^\star)]\nonumber\\
    &+f_m(\v{r}^\star,t^\star), \\
    \bar{f}_m(\v{r}^\star+\v{e}^\star_m,t^\star+1) =&  - \frac{1}{\bar{\tau}^\star}[\bar{f}_m(\v{r}^\star,t^\star)-f^\eq_m(\v{r}^\star,t^\star)]\nonumber\\
    &+\bar{f}_m(\v{r}^\star,t^\star),
\end{aligns}
where $t^\star$ is an integer between 0 (initial time) and $T^\star$ (final time), whereas each $r_i^\star$ is an integer between 1 (one edge of the simulation region) and $N_i$ (the other edge). The conversion to physical units is given by:
\begin{subequations}
\begin{align}
    &\!\!\! t=t^\star \Delta t,\quad T = T^\star \Delta t,\quad \tau= \tau^\star \Delta t = \frac{2\bar{\tau}^\star-1}{2}\Delta t, \label{eq:LBE1}\\
    &\!\!\! \v{r}=\v{r}^\star\Delta x,\quad L_i = N_i \Delta x,\\
   &\!\!\! \v{e}_m=\v{e}_m^\star \frac{\Delta x}{\Delta t},\quad \v{u}=\v{u}^\star \frac{\Delta x}{\Delta t},\quad c_s = c_s^\star \frac{\Delta x}{\Delta t}
   \label{eq:LBE2}
\end{align}
\end{subequations}
where $T$ is the total evolution time, $L_i$ is the spatial dimension in direction $i$, $\Delta x$ is the physical distance between neighboring lattice points, and $\v{u}^\star$ is the lattice velocity. 

The LBE is then usually divided into separate collision and streaming steps in the following way:
\begin{aligns}   
    &\!\!\!\!\bar{\v{f}}^C(\v{r}^\star,t^\star) \!=\! \bar{\v{f}}(\v{r}^\star,t^\star)\!-\!\frac{1}{\bar{\tau}^\star}[\bar{\v{f}}(\v{r}^\star,t^\star)\!-\!\v{f}^\eq(\v{r}^\star,t^\star)],\!\!\label{eq:LBE_col} \\
    &\!\!\!\!\bar{f}_m(\v{r}^\star+\v{e}_m^\star,t^\star+1) = \bar{f}_m^C(\v{r}^\star,t^\star),\label{eq:LBE_str} 
\end{aligns}
with analogous equations for the unbarred variables corresponding to the first-order accurate method. 

Finally, the connection between the LBE and the macroscopic equations of fluid mechanics is provided by the Chapman-Enskog analysis~\cite{kruger2016lattice}. It is based on a perturbation expansion of $\bar{f}_m = \bar{f}_m^{\eq} + \epsilon \bar{f}^{(1)}_m + \epsilon^2 \bar{f}^{(2)}_m + O(\epsilon^3)$ around the equilibrium distribution with the expansion parameter $\epsilon \propto \mathrm{Kn}$, the \emph{Knudsen number} (the ratio of the mean free path to the characteristic length-scale of the problem). Then, by taking moments of this expansion and keeping only the two lowest order terms, one recovers the Navier-Stokes equations in the $\mathrm{Kn}\ll 1$ regime, with the kinematic shear viscosity $\nu$ given by the LBE relaxation time $\tau$ as
\begin{equation}
\label{eq:linktoNS}
    \nu = c_s^2 \tau = c_s^2\left(\bar{\tau}-\frac{\Delta t}{2}\right).
\end{equation}
We can rewrite the right hand side of the above using the lattice units and physical discretizations,
\begin{equation}
\label{eq:linktoNS_lattice}
    \nu = {c_s^\star}^2\left(\bar{\tau}^\star-\frac{1}{2}\right)\frac{\Delta x^2}{\Delta t}.
\end{equation}
Employing the fact that the lattice speed of sound for D1Q3, D2Q9 and D3Q27 models is given by $1/\sqrt{3}$, one obtains a relation between a physical value of kinematic viscosity $\nu$ and the simulation parameters:
\begin{equation}
    \label{eq:chapman}
    \Delta t = \frac{\Delta x^2}{3 \nu}\left(\bar{\tau}^\star-\frac{1}{2}\right).
\end{equation}
To model the same physical scenario, characterized by viscosity $\nu$, one then has the freedom to choose two out of three model parameters $\bar{\tau}^\star$, $\Delta x$, and $\Delta t$, so that they satisfy the above relation. Although different choices do not affect the modeled physics, they influence the accuracy, stability and efficiency of both the classical and quantum LBE algorithms.


\subsection{Quantum algorithms for fluid dynamics}
\label{sec:recent_approaches}

Let us now briefly summarize the most recent QLS-based quantum approaches employing Carleman embedding to simulate the lattice Boltzmann equation presented in Refs.~\cite{li2025potential,penuel2024feasibility}. These works focus on the LBE indirectly, by developing a quantum algorithm for solving the underlying discrete Boltzmann equation. They start from the partial differential equation given by the DBE, Eq.~\eqref{eq:DBE}, and perform spatial discretization to transform it into an ordinary differential equation (ODE). Since the obtained ODE is nonlinear, Carleman linear embedding is then employed, and a linear ODE describing the fluid dynamics in a large linear space is derived. Finally, the resulting linear ODE is solved using known quantum algorithms for solving linear ODEs~\cite{berry2017quantum,krovi2023improved,jennings2024cost}, and the output quantum state, coherently encoding the solution, can be then measured to extract information about the fluid system.  

In order to transform the DBE into an ordinary differential equation, the authors of Refs.~\cite{li2025potential,penuel2024feasibility} consider a spatial discretization of the simulation region, assumed to be a 3-dimensional rectangular cuboid of size $L_x\times L_y\times L_z$, into a grid of $N = N_x N_y N_z$ points with nearest neighbors separated by $\Delta x$. As a result, instead of $Q$-dimensional vectors $\v{f}(\v{r},t)$ at each point in space and time, they are dealing with $d$-dimensional vectors $\v{f}(t)$ at each point in time with $d=NQ$, where
\begin{equation}   
    \label{eq:discretization}
    f_{m,\v{r}}(t) : = f_m(\v{r}, t),\qquad \frac{r_i}{\Delta x}\in\{1,\dots,N_i\},
\end{equation}
and $\Delta x:=L_i/N_i$ with $i\in\{x,y,z\}$. With such a discretization, and assuming weak incompressibility, the DBE in Eq.~\eqref{eq:DBE} can be represented by an ordinary nonlinear differential equation for a vector $\v{f}(t)$ as 
\begin{align}
    \label{eq:DBEmat}
    \frac{d\v{f}(t)}{dt} = (G+F_1)\v{f}(t) + F_2 \v{f}^{\otimes 2}(t)+F_3 \v{f}^{\otimes 3}(t).
\end{align}
The details of this derivation can be found in Appendix~\ref{app:DBE_as_ODE}, together with the explicit forms of matrices $G$ and $F_k$.

Next, the authors of Refs.~\cite{li2025potential,penuel2024feasibility} replace the obtained nonlinear ODE from Eq.~\eqref{eq:DBEmat} with a linear ODE in a larger space using Carleman embedding. To explain how this can be achieved, let us introduce an infinite vector ${\v{y}}^\infty (t)$ of size $d+d^2+d^3+\dots$,
\begin{align}
    {\v{y}}^\infty(t):=[\v{y}_1(t),\v{y}_2(t),\dots]\label{eq:y_infty}.
\end{align}
Then, formally the non-linear DBE evolution for $\v{f}(t)$ from Eq.~\eqref{eq:DBEmat} is equivalent to the following infinite linear ODE for ${\v{y}}^\infty$:
\begin{align}
    \label{eq:carl_evol}
    \frac{d{\v{y}}^\infty(t)}{dt} = \C^\infty_{\dbe} \v{y}^\infty(t),\qquad \v{y}_k(0)={\v{f}}^{\otimes k}(0),
\end{align}
where
\begin{aligns}
    \C^\infty_{\dbe} &= 
        \begin{pmatrix}
            C^1_1 & C^1_2 & C^1_3 & 0     & \dots\\
            0     & C^2_2 & C^2_3 & C^2_4 & \dots\\
            0     & 0     & C^3_3 & C^3_4 & \dots\\
            \vdots&\vdots & \vdots& \vdots& \ddots
        \end{pmatrix},\\
     C^k_{j} &:=
        \sum\limits_{i=0}^{k-1} I^{\ot i}\ot\left(F_{j-k+1}+\delta_{j,k} G\right) \otimes  I^{\otimes k-1-i},
\end{aligns}
is the infinite Carleman matrix with 3 non-zero blocks in each row, and $\delta_{j,k}$ is the Kronecker delta. Truncating this infinite set of equations at level $N_C$, so that one deals with the finite Carleman vector $\v{y}:=[\v{y}_1,\v{y}_2,\dots,\v{y}_{N_C}]$, gives
\begin{align}
    \label{eq:carl_evol_trunc}
    \frac{d{\v{y}}(t)}{dt} = {\C_{\dbe}} {\v{y}}(t),\qquad \v{y}_k(0)={\v{f}}^{\otimes k}(0),
\end{align}
with $\C_{\dbe}$ being the Carleman matrix given by $\C^\infty_{\dbe}$ restricted to the first $N_C$ blocks. Such truncation introduces an error $\epsilon_C$, so that 
\begin{align}
\label{eq:Carlemanconvergence}
    \v{y}_1(t) \approx_{\epsilon_C} \v{f}(t)
\end{align} 
only under certain conditions and for large enough $N_C$. Identifying conditions under which this occurs is a crucial, and often involved, question one needs to answer towards an efficient quantum algorithm.

Under the assumption that the problem from Eq.~\eqref{eq:carl_evol_trunc} approximates well enough the discrete Boltzmann system in the sense of Eq.~\eqref{eq:Carlemanconvergence}, one can attempt to run a quantum linear solver (QLS) algorithm to output its solution. More precisely, modulo time-discretization errors, the solver outputs a normalized Carleman state at the final time $T$:
\begin{equation}
\label{eq:carleman_final}
    \ket{\v{y}(T)} = \frac{1}{\|\v{y}(T)\|}  [\v{y}_1(T),\dots.\v{y}_{N_C}(T)],
\end{equation}
or a \emph{history state} encoding a superposition of the solution at different times:
\begin{align}
    \label{eq:carleman_history}
    \ket{\v{y}_H} = \sum_{m=1}^{\left\lceil T/\Delta t \right\rceil } \frac{\| \v{y}(m \Delta t)\|}{\sqrt{\N}} \ket{\v{y}((m \Delta t))} \otimes \ket{m \Delta t},
\end{align}
with the normalization factor:
\begin{equation}
     \N = \sum_{m=1}^{\left\lceil T/\Delta t \right\rceil } \| \v{y}(m \Delta t)\|^2.
\end{equation}
The cost of obtaining these solution states depends crucially on the condition number of the linear system of equations obtained after time-discretization, see Refs.~\cite{berry2017quantum,
krovi2023improved, jennings2024cost, an2024fast} and references therein, but also on norm properties. This cost is typically defined in terms of number of times we need to apply two core unitaries, one that encodes $\C_{\dbe}$ in one of its blocks, and another that encodes the preparation of a quantum state based on $\v{f}(0)$. Establishing constructions for these unitaries, presenting evidence for a favorable scaling of the condition number and favorable norm properties, are necessary conditions for any claim of quantum advantage. Finally, non-trivial post-processing is required to extract data about the fluid system from the output of the QLS, i.e., from the quantum state $\ket{\v{y}(T)}$ or $\ket{\v{y}_H}$. 

The works in Refs.~\cite{li2025potential,penuel2024feasibility} are important, as they push our understanding forward by analyzing various components of this pipeline. However, as we shall see in the next section, several algorithmic challenges remain to be addressed before the potential for quantum advantage can be fully assessed.


\section{Identification of complexity challenges}
\label{sec:bottlenecks}


\subsection{Applicability of the Carleman approach}
\label{sec:problem_carleman}

The Carleman linear embedding approach to fluid simulation involves extending the nonlinear fluid system into a large, but linear system of equations. The core parameter here is $N_C$, the Carleman truncation order. The resulting ODE system has dimension $O(d^{N_C})$, and the two core questions are:
\begin{enumerate}
    \item Does the Carleman truncation error $\epsilon_C$ converge? This means that $\epsilon_C$ (in the sense of Eq.~\eqref{eq:Carlemanconvergence}) should decrease with increasing $N_C$. Certain sufficient conditions are known for special classes of systems (dissipative~\cite{forets2017explicit, liu2021efficient,an2022efficient,krovi2023improved,costa2025further}, nonresonant~\cite{wu2025quantum}, and conservative~\cite{jennings2025quantum}), but these are not directly applicable to our setting, which is why we resort to numerical evidence.
    \item Assuming the Carleman truncation error $\epsilon_C$ converges, what is a suitable choice of $N_C$ in practice?
\end{enumerate}

Limited numerical analysis is currently available for the Carleman truncation error~$\epsilon_C$ for the discrete Boltzmann equation:
\begin{itemize}
    \item The authors of Ref.~\cite{li2025potential} presented a preliminary numerical study of $\epsilon_C$ for a collisions-only D1Q3 model, with $\tau^\star =1$, a few hundred time steps, $N=1$ (since no streaming is included, a single node is taken), suggesting that $N_C=3$ suffices to get errors below $10^{-14}$. The issue with this approach is that, in such a purely collisional setting, one can actually prove analytically that the Carleman error for truncation order $N_C\geq 3$ is exactly zero, see Appendix~\ref{app:vanishing_carleman} for details. Nontrivial $\epsilon_C$ for $N_C \geq 3$ can only be generated through the interplay of both collisions and streaming, and hence cannot be assessed via numerics on collisions only. The reported figure $\epsilon_C \approx 10^{-14}$ figure seems then to be due to numerical scheme precision rather than the actual Carleman error. 
    \item An $N_C=3$ truncation was chosen in the study in Ref.~\cite{penuel2024feasibility}, where the authors used the result of Ref.~\cite{forets2017explicit} to obtain an extremely narrow time window for guaranteed Carleman convergence, i.e., they proved the Carleman convergence for simulation times much below a single time-step. 
    \item The authors of Ref.~\cite{sanavio2024lattice} instead analyzed a Kolmogorov flow in the D2Q9 model with \mbox{$N_x = N_y = 32$} and around $100$ time steps, setting a truncation $N_C=2$, which entirely foregoes the cubic term in the nonlinearity. Root mean square errors of $10^{-4}-10^{-3}$ were observed in the low Reynolds number regime of $\re\in[10,100]$.
    \item A similar setup and a lid-driven cavity flow were analyzed in Ref.~\cite{turro2025practical} for $N_C\in\{1,2\}$, $N_x=N_y\in\{16,\dots,64\}$, $\re \in[25,50]$. In the first case, errors on the order of $10^{-2}$ (and decaying to zero in a long-time regime) were observed, where the discrepancies compared with Ref.~\cite{sanavio2024lattice} may stem from differences in the boundary conditions and initializations. For the lid driven cavity, the errors of around $10^{-3}$ were found (saturating for long times). Importantly, in the settings considered, it was found that pure linearization, $N_C=1$, outperforms the $N_C=2$ Carleman truncation.  
\end{itemize}

It is therefore problematic to draw from these results any conclusions concerning the ability of the Carleman procedure to capture genuinely nonlinear effects for two reasons. First, the pure linearization errors are already fairly small, and so the analyzed flows must exhibit very weak nonlinearities. Second, the increase of the Carleman error when going from $N_C=1$ to $N_C=2$ suggests that one may already be operating beyond the convergence radius of the Carleman procedure. Thirdly, there is no obvious trend to be gathered from these few scattered studies.

We thus conclude that the question of applicability of the Carleman embedding method to the lattice Boltzmann equation is still mostly unanswered. Since a linear embedding of the nonlinear dynamics is pivotal to any quantum algorithm, in this work we shall consider this problem closely. 


\subsection{Scaling of the condition number}
\label{sec:problem_condition}

Let us now assume that we are in a scenario where the problem of Carleman convergence is solved, i.e., for the use-case of interest a high enough Carleman truncation order $N_C$ leads to a good approximation of the first Carleman block: $\v{y}_1(t) \approx \v{f}(t)$. The next step is to design a quantum algorithm that prepares the Carleman state $\ket{\v{y}(T)}$ (or the Carleman history state $\ket{\v{y}_H}$) describing the fluid system at the final time $T$ (or during the whole evolution between $t=0$ and $t=T$). Recent quantum approaches to the discrete Boltzmann equation rely on quantum linear solvers applied to a linear system arising from time-discretization of Eq~\eqref{eq:carl_evol_trunc}. A crucial parameter to be evaluated here is the scaling of the condition number $\kappa$ of that linear system, since quantum linear solvers have complexity scaling at best linearly with it~\cite{costa2022optimal}. Checking this scaling is crucial for any claim of potential quantum advantage, since under an unfavorable scaling the quantum algorithm may end up performing substantially worse than the best known classical algorithms.

A standard approach that has been followed in recent works~\cite{li2025potential,penuel2024feasibility} is based on discretizing time via a truncated Taylor series method applied to Eq.~\eqref{eq:carl_evol_trunc}. Several works presented condition number upper bounds for this setting under slight variations~\cite{berry2015simulating, krovi2023improved, jennings2024cost}. For the class of problems at hand, it was found that the condition number scales as
\begin{equation}
    \kappa = O\left(T,\|\C_{\dbe}\|,\max_{t \in [0,T]} \| e^{\C_{\dbe} t}\|\right),
\end{equation}
where we recall that $T$ is the total simulation time, $\C_{\dbe}$ is the truncated Carleman matrix, and $\| \cdot \|$ denotes the operator norm (i.e., the largest singular value).

Concerning the norm of $\C_{\dbe}$, in Ref.~\cite{li2025potential} [Eq.~(93)], it is claimed that $\|\C_{\dbe}\| = O(\tau^{-1})$, where $\tau$ is a  constant. However, the top-left block of the matrix $\C_{\dbe}$ includes the finite-difference discretization of the streaming operator $\v{e}_m \cdot \nabla f_m$. In turn, this means that the largest element in $\C_{\dbe}$ scales as $\Delta x^{-1} = O(N_x,N_y,N_z)$, and hence $\|\C_{\dbe}\|$ must also scale as $O(N_x,N_y,N_z)$ or worse. This negates the headline exponential advantage in $N$, independently of the question of the initial state preparation and data extraction. Also, note that in the context of the lattice Boltzmann equation, one has $\|\C\| = O((\taus)^{-1})$ (which does not scale with $N$), but this is simply because we moved to lattice units, where the advection time is $N$-dependent, so we just shifted the $N$-scaling to the number of required time steps. More precisely, in the discrete case, it is known that the number of time steps needs to grow with the system size in order for the fluid flow to develop, as one must allow at least for the interaction of the edges of the simulation region with themselves~\cite{penuel2024feasibility}. A more detailed analysis of $\|\C_{\dbe}\|$ for this case can be found in~Ref.~\cite{penuel2024feasibility}. Overall, this tells us that the best we can hope for the quantum algorithm under current approaches is a polynomial (in fact cubic) speedup, not an exponential one.\footnote{This limitation could potentially be handled via a different algorithm implementing an `interaction-picture' in the frame of streaming, but this would require further analysis.} 

Now, consider $\max_{t \in [0,T]} \| e^{\C_{\dbe} t}\|$. In Ref.~\cite{li2025potential} [Sec. 5e], it is claimed that this quantity is $O(1)$ under stability assumptions. Formally, the claim is that the norm of $e^{\C_{\dbe} t}$ is $O(1)$ when $\C_{\dbe}$ is marginally stable, i.e., when the largest real part of the eigenvalues of $\C_{\dbe}$ is nonpositive. While it is true that $\|e^{\C_{\dbe} t}\|$ can be upper bounded by a quantity independent of $t$, $\max_{t \in [0,T]} \|e^{\C_{\dbe} t}\|$ generally can depend on $N$.\footnote{$\| e^{\C_{\dbe} t}\| \leq \kappa_P$ for any $P>0$ satisfying the Lyapunov condition $P \C_{\dbe} + \C^\dag_{\dbe} P < 0$, where $\kappa_P$ is the condition number of $P$. The parameter $\kappa_P$ can bring in an $N$-dependence.} The work in Ref.~\cite{penuel2024feasibility} instead sets $\max_{t \in [0,T]} \|e^{\C_{\dbe} t}\|$ equal to $1$, as a placeholder until further analysis becomes available.

In summary, pinning down the $N$ dependence of the condition number is crucial for determining the size of a quantum advantage, if any exists. In this work, we shall look at this problem in detail.


\subsection{Small time steps}
\label{sec:problem_timestep}

Previous quantum algorithmic approaches to the LBE focused on the discrete Boltzmann equation (DBE) from Eq.~\eqref{eq:DBE}. This is a natural starting point to apply the Carleman linear embedding formalism, as well as quantum algorithms whose analysis normally assumes a time-continuous model. However, achieving a numerically stable first-order accurate time-discretization of the collision integral requires $\Delta t < 2\tau$~\cite{kruger2016lattice}. Furthermore, the explicit forward-time central-space discretization of the streaming term imposes a CFL condition on the time-step. The standard lattice Boltzmann method avoids these particular constraints on the time-step size by treating streaming exactly and stiff collision term implicitly using a second order trapezoidal rule (see Sec.~\ref{sec:LBE}). As discussed by Refs.~\cite{ubertini2006unstructured, ubertini2010three}, a second order accurate time discretization of DBE (termed DBE with memory) is possible, however it can only be achieved by introducing a memory term which introduces second derivatives in both space and time. The approach of Refs.~\cite{li2025potential, penuel2024feasibility} is therefore equivalent to dropping the memory term from the DBE with memory formulation of \cite{ubertini2010three}. This results in a first-order accurate formulation constrained by the numerical stability condition $\Delta t < 2\tau$. Therefore, the currently proposed quantum approaches could be extended to be more competitive in practice. In this work, we shall address this time-step problem.


\subsection{Inefficient data extraction}
\label{sec:problem_info}

Finally, assuming that one somehow deals with the Carleman convergence problem and the complexity of obtaining the final time Carleman state $\ket{\v{y}(T)}$ from Eq.~\eqref{eq:carleman_final}, there is still the problem of extracting data from the quantum state encoding the dynamics. We will focus on the problem of extracting classical information about the fluid at time $T$ from $\ket{\v{y}(T)}$. Note that analogous issues as discussed below appear when one wants to extract information from the history state $\ket{\v{y}_H}$.

First, note that it is not difficult to estimate that in the weakly compressible (or incompressible) limit, at all times $t$ one has 
\begin{equation}
    \sqrt{\frac{N}{Q}} \lesssim \|\v{f}(t)\| \lesssim \sqrt{N}. 
\end{equation}
As a result, the norm of the $k$-th Carleman block of $\ket{\v{y}(T)}$ scales as
\begin{equation}
\label{eq:crucialratio} \frac{\|\v{y}_k(T)\|}{\|\v{y}(T)\|} = O(\|\v{f}(T)\|^{k-N_C}) = O\left(\sqrt{N}^{k-N_C}\right).
\end{equation}
This, in turn, means that in order to extract information from the first Carleman block (which approximates $\v{f}(T)$ well enough), a quantum algorithm incurs a cost $O(\sqrt{N}^{N_C-1})$, simply due to the fact that amplitude is exponentially concentrated in the higher order $k$-blocks of the Carleman vector~\cite{costa2025further}. Hence, already for $N_C=3$, which is the smallest order capturing the cubic nonlinearities of DBE, the cost scales linearly in the number of discretization points, so the same as the classical algorithm. Increasing $N_C$ further to decrease the Carleman truncation error renders the quantum algorithm less efficient than its classical counterpart.

One could also try extracting information from some higher block $\v{y}_k$, as it approximates $k$ copies of the solution, so one could read the solution independently $k$ times from such a block. However, not only would that result in a higher truncation error, but also we would still incur an extra scaling $O(\sqrt{N}^{N_C-k})$. This extra cost kills quantum advantage over classical algorithms unless $k=N_C$ (when the dependence on $N$ could disappear), or $k=N_C-1$ (when one could hope for a quadratic improvement from classical $O(N)$ to $O(\sqrt{N})$). However, $\v{y}_{N_C}(t) = \v{f}_{\mathrm{lin}}^{\ot N_C}(t)$ with $\v{f}_{\mathrm{lin}}$ describing the solution to the \emph{linear} ODE 
\begin{align}
    \label{eq:DBElinear}
    \frac{d\v{f}_{\mathrm{lin}}(t)}{dt} = (G+F_1)\v{f}_{\mathrm{lin}}(t),\qquad \v{f}_{\mathrm{lin}}(0) = \v{f}(0),
\end{align}
which corresponds to applying naive linearization procedure (i.e., dropping the nonlinear terms) applied to the original DBE given by Eq.~\eqref{eq:DBE}, which we could have obtained in a simpler and more direct manner without invoking any Carleman procedure. Similarly, $\v{y}_{N_C-1}$ only describes the solution of the \emph{quadratic} ODE, completely ignoring cubic nonlinearities. Moreover, it is not clear whether choosing a high $N_C$ and extracting data from $\v{y}_{N_C-1}$ is any better, in terms of truncation error $\epsilon_C$, than just choosing $N_C=2$ and focusing on $\v{y}_1$. Hence, it is quite likely that even if one is in the regime of converging error, one cannot control it without making the algorithm inefficient.

Most existing literature on quantum algorithms for the lattice Boltzmann equation only discusses how to output the Carleman states $\ket{\v{y}(T)}$ or $\ket{\v{y}_H}$, without discussing the critical problem of how to extract classical information from a coherent encoding of the dynamics~\cite{itani2024quantum,sanavio2024lattice,turro2025practical}. The two works where extraction of data is discussed are Ref~\cite{li2025potential} and, in more detail, Ref.~\cite{penuel2024feasibility}. However, they appear to miss this problem:
\begin{itemize}
    \item In Ref.~\cite{li2025potential}, the crucial ratio from Eq.~\eqref{eq:crucialratio} is claimed to be $O(1/\sqrt{N_C})$ (Eq.~(102)). This can be traced back to the bound of Eq.~(85) in Ref.~\cite{li2025potential}, which implicitly assumes $\| \v{f}(t)\| \leq 1$, which does not hold. Technically, one can perform a rescaling $\v{f} \mapsto \gamma \v{f}$ to ensure the norm is smaller than $1$ as required, but this adds an extra scaling with $N$ to the condition number, removing any quantum advantage, as we shall see below. 
    \item The authors of Ref.~\cite{penuel2024feasibility}, on the other hand, try to avoid the problem by extracting information from all Carleman blocks. The issue with the argument is the assumption that the Carleman procedure outputs a state of the form $\v{y}_k \approx \v{f}^{\otimes k}$. In reality, due to the reasons discussed above,  the amplitude is concentrated in the  block $k=N_C$, which simply contains information about the linear dynamics, $\v{y}_{N_C}(t) = \v{f}_{\mathrm{lin}}^{\ot N_C}(t)$. Since tackling relevant applications requires high values of $N$ (Ref.~\cite{penuel2024feasibility} reports use-cases with $N$ between $10^{19}$ and $10^{25}$), this basically means that the output of the quantum algorithm will just reflect the results one would obtain from a naive linearization. As such, it cannot capture the nonlinearity as required. 
\end{itemize}

As anticipated, one potential way around the problem described above is to rescale the DBE according to:
\begin{aligns}
    \v{f} \mapsto \gamma \v{f}, \quad G \mapsto G, \quad F_1 \mapsto F_1, \\ 
    F_2 \mapsto F_2/\gamma, \quad F_3 \mapsto F_3/\gamma^2. 
\end{aligns}
By taking $\gamma < 1/\|\v{f}\|$, one can ensure the norm of $\v{f}$ is rescaled to be smaller than $1$, and thus avoid the extra complexity cost mentioned. However, as already mentioned in Sec.~\ref{sec:problem_condition}, the condition number $\kappa$ of the linear system describing time-discretized DBE evolution for time $T$ scales as $O(\|\C_{\dbe}\| T)$. With such rescaling, we have
\begin{equation}
    \|\C_{\dbe}\| = O(1/\gamma^2) = O(\|\v{f}\|^2) = O(N),
\end{equation}
and so $\kappa=O(NT)$. Given that the complexity of the best QLS scales linearly with $\kappa$, we end up with a quantum algorithm whose complexity scales the same way as the classical algorithms with the number of spatial discretization points $N$ and simulation time $T$. We then conclude that if the inefficient data extraction problem cannot be mitigated, no quantum advantage for the simulation of nonlinear dynamics is possible.


\section{Classical problem reformulation}
\label{sec:problem_formulation}


We now discuss a formulation of the problem that helps to overcome the obstacles described in the previous section. Instead of exploiting the assumption of weak incompressibility to motivate a truncated Taylor expansion of $1/\rho\approx 2-\rho$, enabling $\v{f}^\eq$ to be expressed as a cubic function of $\v{f}$, we will instead use the incompressible formulation of LBE developed in Ref.~\cite{he1997lattice}. At the expense of discarding weakly compressible effects, which in low Mach number flows are already negligible, this approach provides an $O(\ma^2)$ incompressible LBE formulation whose collision integral is quadratic in the state vector and thus upon application of the Carleman embedding facilitates more efficient data extraction that its cubic counterpart.

\subsection{Shifted incompressible LBE}\label{sec:shifted}

When $\ma \ll 1$, compressibility effects are negligible and it is useful to instead solve the incompressible Navier-Stokes equations which (in dimensional form) are:
\begin{subequations}
    \begin{align}
        \label{eq:Momentum}
        \frac{\partial \v{u}}{\partial t} + \left( \v{u} \cdot \nabla \right) \v{u} &= -\frac{1}{\rho_0} \nabla p + \nu \nabla^2 \v{u},\\
        \label{eq:divU}
        \nabla \cdot \v{u} & = 0.
    \end{align}
\end{subequations}
In the above, $\rho_0$ and $\nu$ are the constant reference density and kinematic viscosity, respectively, while $p(\v{r},t)$ is the pressure field. In order to recover Eqs.~\eqref{eq:Momentum}-\eqref{eq:divU} via the lattice Boltzmann equation, one needs to modify the equilibrium distribution from Eq.~\eqref{eq:eq}, so that only the density variations that do not strongly violate Eq.~\eqref{eq:divU} are retained. Expanding the density as $\rho = \rho_0 + \delta \rho$ with $|\delta \rho| \ll \rho_0$ and letting the pressure fluctuation be given by $p = c_s^2 \delta \rho$, we can substitute these expressions into Eq.~\eqref{eq:eq} to arrive at
\begin{align}
    &f_m^\eq(\v{r},t) = w_m \rho_0 + w_m \frac{p}{c_s^2} \nonumber\\
    &\quad+ w_m \rho_0 \left( \frac{\v{e}_m\cdot \v{u}(\v{r},t)}{c_s^2} + \frac{(\v{e}_m\cdot \v{u}(\v{r},t))^2}{2 c_s^4} -\frac{|\v{u}(\v{r},t)|^2}{2 c_s^2}\right) \nonumber\\ 
    &\quad+ O(\ma^3).
\end{align}
In addition to terms proportional to $\v{u}$, the term $p \sim {O}(\ma^2)$ must be retained because, in the incompressible Navier-Stokes equations, the pressure term corrects for the divergences produced by the convective nonlinearity (see, e.g., Chapter~1 of Ref.~\cite{doering1995applied}). However, as $w_m \rho_0 \sim {O}(1)$ is constant, we can also shift the density $f_m - w_m\rho_0 \;  \to \; g_m$ (this is related to the trick used in Refs.~\cite{higuera1989simulating,benzi1992lattice} allowing one to save digits lost in the cancellation $f_m - w_m\rho_0$). The shifted equilibrium density then becomes
\begin{align}
    &g_m^\eq(\v{r},t)= w_m \frac{p}{c_s^2} \nonumber\\ 
    &\quad + w_m\rho_0 \left( \frac{\v{e}_m\cdot \v{u}(\v{r},t)}{c_s^2} + \frac{(\v{e}_m\cdot \v{u}(\v{r},t))^2}{2 c_s^4} -\frac{|\v{u}(\v{r},t)|^2}{2 c_s^2}\right) \nonumber\\ 
    &\quad + O(\ma^3),
\end{align}
with the pressure and velocity given by
\begin{aligns}
    p(\v{r},t) &:= c_s^2 \sum_{m=1}^Q g_m(\v{r},t), \\
    \v{u}(\v{r},t) &:= \frac{1}{\rho_0}\sum_{m=1}^Q g_m(\v{r},t)\v{e}_m.
\end{aligns}
As discussed in Chapter~3 of Ref.~\cite{guo2013lattice}, this formulation recovers exactly the momentum equation, Eq.~\eqref{eq:Momentum}, but approximates the divergence free condition, Eq.~\eqref{eq:divU}, as 
\begin{equation}
    \nabla \cdot \v{u} = \frac{1}{c_s^2} \frac{\partial p}{\partial t}, 
\end{equation}
where the right hand side is ${O}(\ma^2)$. Although the above approach is clearly suited for steady or slowly varying flows in which hydrodynamic timescale is much slower than acoustic time-scales $T \gg L/c_s$~\cite{silva2012first}, the collision integral proposed by Ref.~\cite{guo2000lattice} (also quadratic in $g_m$) is exact and therefore preferable for time-dependent flows. However, as the latter collision integral is more complex, we have restricted this work to the simpler formulation in order to demonstrate its connection with the weakly compressible approach.

Returning to non-dimensional variables, we observe from Eqs.~\eqref{eq:LBE_col}-\eqref{eq:LBE_str} that this shifted LBE state vector evolves according to the following equations:
\begin{aligns}
    \label{eq:LBE_col_shift}
    &\!\!\!\!\v{g}^C(\v{r}^\star,t^\star)  = \v{g}(\v{r}^\star,t^\star)- \frac{1}{\bar{\tau}^\star}[\v{g}(\v{r}^\star,t^\star)-\v{g}^\eq(\v{r}^\star,t^\star)],\!\!\\
    \label{eq:LBE_str_shift}
    &\!\!\!\!g_m(\v{r}^\star+\v{e}_m^\star,t^\star+1)  = g_m^C(\v{r}^\star,t^\star),
\end{aligns}
with
\begin{align}
    \label{eq:eq_shift}
    g_m^\eq= \left[\delta \rho + \left(3\v{e}_m^\star\cdot \v{u}^\star + \frac{9}{2}(\v{e}_m^\star\cdot \v{u}^\star)^2 -\frac{3}{2}|\v{u}^\star|^2\right) \right]w_m,
\end{align}
which can be rewritten as a quadratic function of $\v{g}$:
\begin{align}
     g_m^\eq& = \left( \frac{9}{2}\left(\sum_{m_1=1}^Q E_{m,m_1}g_{m_1}\right)^2 -\frac{3}{2}\left|\sum_{m_1=1}^Q g_{m_1}\v{e}_{m_1}^\star\right|^2\right.\nonumber\\
     &\qquad \left.+\sum_{m_1=1}^Q  g_{m_1}+3 \sum_{m_1=1}^Q E_{m,m_1} g_{m_1}\right) w_m\nonumber\\
      & =   \left(\sum_{m_1=1}^Q \sum_{m_2=1}^Q \left(\frac{9}{2} E_{m,m_1} E_{m,m_2}-\frac{3}{2} E_{m_1,m_2}\right) g_{m_1} g_{m_2}\right.\nonumber\\
      &\qquad +\left.\sum_{m_1=1}^Q  g_{m_1}+\sum_{m_1=1}^Q 3 E_{m,m_1} g_{m_1} \right)w_m , \label{eq:eq_shift_explicit}
\end{align}
where we introduced a Gram matrix $E$ with matrix elements 
\begin{equation}
    E_{m,m_1}=\v{e}_m^\star \cdot \v{e}_{m_1}^\star,
\end{equation} 
and the fact that for the considered discrete velocity models from Table~\ref{tab:velocities} one has $c^\star_s=1/\sqrt{3}$. Although we restrict our attention to a simple streaming step without walls (just periodic boundary conditions), and without external driving we discuss and analyze the extension to non-trivial geometries and driven flows elsewhere~\cite{jennings2026simulating}.

Given that the $D$-dimensional position vector $\v{r}^\star$ in lattice units takes only integer values, the state of the system at each time step $t^\star$ can be represented by the following vector
\begin{equation}
    \v{g}(t^\star) = \sum_{\v{r}^\star} \sum_{m=1}^Q  g_m(\v{r}^\star,t^\star) \ket{\v{r}^\star} \ot \ket{m},
\end{equation}
where we explicitly introduced position and velocity registers, and the dimension of $\v{g}$ is
\begin{equation}
    d:= N Q,
\end{equation}
where 
\begin{equation}
    N:=\prod_{i=1}^D N_i
\end{equation}
is the total number of spatial lattice points. We can then rewrite the two steps of the shifted LBE, Eqs.~\eqref{eq:LBE_col_shift}-\eqref{eq:LBE_str_shift}, as
\begin{aligns}
    \label{eq:LBE_col_shift_matrix}
    \v{g}^C(t^\star) & = \left(  I+ F_1\right)\v{g}(t^\star)+ F_2 \v{g}(t^\star)^{\ot 2},\\
    \label{eq:LBE_str_shift_matrix}
    \v{g}(t^\star+1) & = S\v{g}^C(t^\star),
\end{aligns}
where $S$ is a $d\times d$ permutation matrix (so also a unitary matrix) encoding streaming:
\begin{equation}
    \label{eq:streaming}
    S = \sum_{\v{r}^\star}\sum_{m=1}^Q \ketbra{\v{r}^\star+\v{e}_m^\star}{\v{r}^\star} \ot \ketbra{m}{m},
\end{equation}
whereas $F_k$ are $d\times d^k$ matrices describing linear and quadratic contributions from collisions, the form of which can be read off from Eq.~\eqref{eq:eq_shift_explicit} to be:
\begin{aligns}
    \label{eq:F1}
    &F_1 =  \sum_{\v{r}^\star} \ketbra{\v{r}^\star}{\v{r}^\star} \ot \sum_{m,m_1=1}^Q  (\tilde{F}_1)_{m,m_1}\ketbra{m}{m_1},\\
    \label{eq:F2}
    &F_2 = \sum_{\v{r}^\star} \ketbra{\v{r}^\star}{\v{r}^\star \v{r}^\star} \ot  \!\!\!\!\!\!\!\sum_{m,m_1,m_2=1}^Q \!\!\!\!\!\!\!(\tilde{F}_2)_{m,m_1m_2} \ketbra{m}{m_1m_2},
\end{aligns}
where
\begin{aligns}
    \!\!\!\!(\tilde{F}_1)_{m,m_1}&=\frac{1}{\taus} (w_m + 3 w_{m} E_{m,m_1} - \delta_{m,m_1}),\\
    \!\!\!\!(\tilde{F}_2)_{m,m_1m_2} & =\frac{w_m}{\bar{\tau}^\star}  \left(\frac{9}{2}E_{m,m_1}E_{m,m_2}-\frac{3}{2}E_{m_1,m_2}\right).\!\!
\end{aligns}
Notice that the above can be rewritten as
\begin{aligns}
    \label{eq:F1_loc}
    F_1 &= I_{\v{r}^\star} \ot \tilde{F}_1, \\
    F_2 &= I_{\v{r}^\star,\v{r}^\star} \ot \tilde{F}_2, 
    \label{eq:F2_loc}
\end{aligns}
where $I_{\v{r}^\star}$ and $I_{\v{r}^\star,\v{r}^\star}$ are $N\times N$ and $N \times N^2$ identity matrices on the position registers, whereas $\tilde{F}_1$ and $\tilde{F}_2$ are $Q\times Q$ and $Q \times Q^2$ collision matrices acting on the velocity register.

As explained above, the term $\delta\rho$ in Eq.~\eqref{eq:eq_shift} scales as $O(|\v{u}^\star|^2)$, while the remaining terms in that equation have explicit dependence on $\v{u}^\star$, so that we can conclude that
\begin{equation}
    g_{m}^\eq =O(|\v{u}^\star|).
\end{equation}
Moreover, in the incompressible lattice Boltzmann formulation, the Mach number is taken as the small expansion parameter. Thus, the Chapman–Enskog expansion then gives that the non-equilibrium component of $\v{g}$ is of order $O(|\v{u}^\star|)$. Physically, in the incompressible low-Knudsen regime, collisions keep the distribution close to the local equilibrium manifold, so non-equilibrium kinetic corrections remain of the same asymptotic order as the hydrodynamic velocity perturbation. Thus, denoting the upper bound on the magnitude of lattice velocity $|\v{u}^\star|$ during the evolution by $u^\star_\mx$, we get
\begin{equation}
    \|\v{g}\| = \sqrt{\sum_{\v{r}^\star} \sum_{m=1}^Q |g_m(\v{r}^\star,t^\star)|^2} = O\left(\sqrt{N}u^\star_\mx\right).
\end{equation}
Therefore, in order to keep the norm of the solution bounded, $\|\v{g}\|<1$, one can simply enforce the following value of lattice velocity
\begin{equation}
    \label{eq:u_norm_condition}
    u^\star_\mx = \frac{u_0^\star}{\sqrt{N}} 
\end{equation}
for some constant $u_0^\star$. This way one can control the norm of $\v{g}$, without a direct negative impact on the condition number as in the rescaling approach (as we shall see in Sec.~\ref{sec:condition}), thus overcoming the main obstacle for data extraction that we discussed in Sec.~\ref{sec:problem_info}. 

Finally, although we analyze the inclusion of non-trivial geometries elsewhere~\cite{jennings2026simulating}, let us briefly explain here how to include walls in our model. One can introduce an indicator function $\mathcal{W}$ encoding the wall boundary conditions in the following way: 
\begin{equation}
    \label{eq:boundary}
    \mathcal{W}_{\v{r}^\star} = \left\{
    \begin{array}{cc}
         1&:~\v{r}^\star \mathrm{~is~a~wall~node,}   \\
         0& :~\v{r}^\star \mathrm{~is~a~fluid~node,}
    \end{array}
    \right.
\end{equation}
Then, the bounce-back effect of the walls can be accounted for by modifying the streaming matrix in the following way:
\begin{align}
     \label{eq:streaming_walls}
    S = &\sum_{\v{r}^\star}\sum_{m=1}^Q (1\!-\!\mathcal{W}_{\v{r}^\star+\v{e}_m^\star})(1\!-\!\mathcal{W}_{\v{r}^\star})\ketbra{\v{r}^\star\!+\!\v{e}_m^\star}{\v{r}^\star} \!\ot\! \ketbra{m}{m}\nonumber\\
    & + \sum_{\v{r}^\star}\sum_{m=1}^Q \mathcal{W}_{\v{r}^\star+\v{e}_m^\star}(1\!-\!\mathcal{W}_{\v{r}^\star})\ketbra{\v{r}^\star}{\v{r}^\star} \ot \ketbra{-m}{m}\nonumber\\
    & + \sum_{\v{r}^\star}\sum_{m=1}^Q \mathcal{W}_{\v{r}^\star} \ketbra{\v{r}^\star}{\v{r}^\star} \ot \ketbra{m}{m},
\end{align}
where, with a slight abuse of notation, $-m$ is a velocity index corresponding to discrete velocity $-\v{e}^\star_{m}$. In the above, the first term corresponds to unmodified streaming when both nodes considered are fluid; the second term encodes bouncing, when a particle in a fluid node hits a wall node; and the third one simply means no evolution (identity matrix) inside the wall nodes. Crucially, the streaming matrix modified according to the above recipe is still a permutation matrix, so $S$ stays unitary.


\subsection{Parameter selection}
\label{sec:parameters}

We will now explain how to select simulation parameters $N_x$ (spatial extent in lattice units), $T^\star$~(temporal duration in lattice units), and $\bar{\tau}^\star$ (relaxation time in lattice units) in order to guarantee the required scaling of the lattice velocity from Eq.~\eqref{eq:u_norm_condition} and, at the same time, to be able to simulate general dynamics of a fluid described by kinematic viscosity $\nu$. As we shall see, our main assumptions here are as follows:
\begin{itemize}
    \item The number of discretization points $N$ should scale with a power of the Reynolds number, typically $\re^{3D/4}$, in order to resolve the \emph{Kolmogorov microscale} and capture the smallest dynamical features of turbulent flows.
    \item The simulation has to be run till a steady flow develops, which happens after some constant multiple of the advection time.
\end{itemize}

First, let us denote the upper bound on the maximal physical velocity during the evolution by $u_\mx$ and introduce the ratio $\eta_u$ of the flow velocity $u$ to $u_\mx$,
\begin{equation}
    \eta_u = \frac{u}{u_\mx}.
\end{equation}
Next, assume for simplicity a cubic simulation region, so that all $L_i$ and $N_i$ are equal, and denote by $\eta_L$ the ratio of the characteristic length $L$ to the simulation region size $L_x$,
\begin{equation}
    \eta_L := \frac{L}{L_x}.
\end{equation}
Thus, the Reynolds number of the considered fluid problem can be expressed as
\begin{equation}
    \re := \frac{L u}{\nu} = \frac{L_x u_\mx}{\nu} \eta_L \eta_u.
\end{equation}

Next, using Chapman-Enskog relation, Eq.~\eqref{eq:chapman}, in the expression for the conversion between physical and lattice velocity, we get
\begin{align}
    {u}^\star_\mx &= \frac{\Delta t}{\Delta x} u_\mx = \frac{\Delta x u_\mx}{\nu } \frac{(\bar{\tau}^\star-\frac{1}{2})}{3} \nonumber\\
    &= \frac{L_x u_\mx}{\nu } \frac{(\bar{\tau}^\star-\frac{1}{2})}{3N_x} = \frac{(\bar{\tau}^\star-\frac{1}{2})}{3N_x}\frac{\re}{\eta_L\eta_u}. 
\end{align}
Thus, enforcing Eq.~\eqref{eq:u_norm_condition} means
\begin{equation}
    \bar{\tau}^\star = \frac{1}{2} +\frac{3 u_0^\star}{\re}\frac{\eta_L\eta_u}{N_x^{D/2-1}}.
\end{equation}
Next, we recall that in order to resolve the Kolmogorov microscale, one requires $\Delta x \simeq L /\re^{3/4}$. We thus choose the number of discretization points per direction to be the following function of the Reynolds number:
\begin{equation}
\label{eq:NxandRe}
    N_x = \frac{\re^{\beta}}{\eta_L}.
\end{equation}
In what follows, we will keep $\beta$ general to allow the choice of a better or worse resolution, but when we perform simulations or want to compare with the performance classical algorithms, we will always choose $\beta = 3/4$, which guarantees resolving the Kolmogorov microscale. 

With that, the expression for the relaxation constant becomes
\begin{equation}
    \bar{\tau}^\star = \frac{1}{2} + \frac{3u_0^\star \eta_L^{D/2}\eta_u}{\re^{\beta(D/2-1)+1}}.
\end{equation}
Since $\Delta x$ is fixed by the condition to resolve the Kolmogorov microscale and $\bar{\tau}^\star$ is fixed by the norm condition from Eq.~\eqref{eq:u_norm_condition}, $\Delta t$ is also fixed by the Chapman-Enskog relation, Eq.~\eqref{eq:chapman}, to be
\begin{equation}
    \Delta t = \frac{L}{u} \frac{u_0^\star \eta_L^{D/2}\eta_u}{\re^{\beta(D/2+1)}}.
\end{equation}
Finally, assuming that the steady flow will develop after $1/\eta_T$ advection times, 
\begin{equation}
    T = \frac{1}{\eta_T}\frac{L}{u},
\end{equation}
the total number of time steps to simulate the fluid under these conditions is given by
\begin{equation}
    T^\star = \frac{T}{\Delta t} =  \frac{1}{\eta_T\eta_u\eta_L^{D/2}} \frac{\re^{\beta(D/2+1)}}{u_0^\star}.
\end{equation}

In this work we will investigate the simplified scenario, where we put 
\begin{equation}
    \eta_T=\eta_u=\eta_L =1,\qquad u_0^\star = 1,
\end{equation}
i.e., we simulate for one advection time, the characteristic length of the problem is equal to its spatial extent, the maximal velocity is well approximated by the flow velocity, and the constant $u_0^\star$ controlling the norm $\|\v{g}\|$ is 1. In such a case, we get the following simulation parameters:
\begin{aligns}
    \label{eq:parameter_choice1}
    N_x = \left\lceil \re^\beta \right\rceil ,\qquad T^\star = \left\lceil \re^{\beta(D/2+1)} \right\rceil,\\
    \label{eq:parameter_choice2}
    \bar{\tau}^\star = \frac{1}{2} + \frac{3}{\re^{\beta(D/2-1)+1}},\quad u_\mx^\star = \frac{1}{\re^{\beta D/2}},
\end{aligns}
where we used the ceiling function to guarantee integer values of $N_x$ and $T^\star$. Thus, for a fixed problem dimension $D$ and spatial resolution parameter $\beta$ (with $\beta= 3/4$ resolving the Kolmogorov microscale), the above simulation parameters become a simple function of the Reynolds number~$\re$. Note that since the parameters $\eta_T$, $\eta_u$, $\eta_L$ and $u_0^\star$ are constants (i.e., independent of the system size or the Reynolds number), the particular choice made by us does not affect the overall scaling of simulation parameters with $\re$, and so does not affect the asymptotic scaling of the algorithm's complexity.


\subsection{Discrete Carleman embedding}
\label{sec:discrete_carleman}

We now have a discrete-time nonlinear problem describing the fluid dynamics via Eqs.~\eqref{eq:LBE_col_shift_matrix}-\eqref{eq:LBE_str_shift_matrix} with the parameter choice given by Eqs.~\eqref{eq:parameter_choice1}-\eqref{eq:parameter_choice2}. This provides a well-defined problem, and we now wish to consider how it can be simulated on a quantum computer. As discussed previously, this cannot be done directly, but we must construct a linear representation of the problem that is amenable for the quantum computer.
Therefore, the next step is to embed the nonlinear equations in a larger dimensional space, so that it can be well approximated by a discrete-time linear problem. To achieve this, we start by introducing infinite vectors $\v{y}^{C,\infty}(t^\star)$ and $\v{y}^\infty(t^\star)$ of size $d+d^2+d^3+\dots$,
\begin{aligns}
    \v{y}^{C,\infty}(t^\star) & := [\v{y}_1^C(t^\star),\v{y}_2^C(t^\star),\dots],\\
    \v{y}^{\infty}(t^\star) & := [\v{y}_1(t^\star),\v{y}_2(t^\star),\dots],
\end{aligns}
where
\begin{equation}
    \v{y}_k^C(t^\star) = [\v{g}^C(t^\star)]^{\otimes k},\qquad \v{y}_k(t^\star) = [\v{g}(t^\star)]^{\otimes k},
\end{equation}
which carry information about all powers of the shifted LBE evolution, respectively ending with a streaming or a collision step. As this replacing of higher powers with new variables is inspired by the Carleman linear embedding procedure~\cite{berkolaiko1998analysis,pruekprasert2020moment,pruekprasert2024moment}, we will refer to $\v{y}^{C,\infty}(t^\star)$ and $\v{y}^\infty(t^\star)$ as the \emph{infinite Carleman vectors}. From Eq.~\eqref{eq:LBE_str_shift_matrix}, we straightforwardly obtain the streaming update rule for $\v{y}^{\infty}(t^\star)$:
\begin{align}
     {\v{y}}^{\infty}(t^\star+1) &=\S^\infty {\v{y}}^{C,\infty}(t^\star),
\end{align}
where
\begin{equation}
    \label{eq:embeded_str_bound}
    \mathcal{S}^\infty = 
    \begin{pmatrix}
        S & 0 & 0 & \dots\\
        0     & S^{\otimes 2} & 0 & \dots\\
        0     & 0     & S^{\otimes 3} &\dots\\
        \vdots&\vdots & \vdots& \ddots 
    \end{pmatrix}
\end{equation}
is an \emph{infinite Carleman streaming matrix}, which is unitary. 

To get the collision update rule for $\v{y}^{C,\infty}(t^\star)$, let us start from the LBE collision equation, Eq.~\eqref{eq:LBE_col_shift_matrix}, for $k$ copies of the state vector:
\begin{align}
    [&\v{g}^C(t^\star)]^{\otimes k}\nonumber\\
    &=\left[\left( I+ F_1\right)\v{g}(t^\star) + F_2 [\v{g}(t^\star)]^{\otimes 2}\right]^{\otimes k}\nonumber\\
    &=\sum_{k_1=0}^k \left[\left( I+ F_1\right)^{\otimes k_1} \otimes F_2^{\otimes k-k_1}+\mathrm{perms.}\right] [\v{g}(t^\star)]^{\otimes 2k-k_1}\nonumber\\
    &=\sum_{l=k}^{2k} \left[\left( I+ F_1\right)^{\otimes 2k-l} \otimes F_2^{\otimes l-k}+\mathrm{perms.}\right] [\v{g}(t^\star)]^{\otimes l}\nonumber\\
    &=\sum_{l=k}^{2k} C_l^k [\v{g}(t^\star)]^{\otimes l},
\end{align}
where we have introduced $d^k\times d^l$ matrices
\begin{equation}
\label{eq:Ckl}
    C_l^k:=\left( I+F_1\right)^{\otimes 2k-l} \otimes F_2^{\otimes l-k}+\mathrm{perms.},
\end{equation}
and where $+ \ \mathrm{perms.}$ denotes a sum over distinct permutations over $k$ subsystems, with $2k-l$ subsystems of type 1, and $l-k$ subsystems of type 2. This means that
\begin{equation}
    \v{y}_k^C(t^\star) = \sum_{l=k}^{2k} C^k_l \v{y}_l (t^\star),
\end{equation}
which can be written compactly as
\begin{align}
    \label{eq:carl_evol_d}
    \v{y}^{C,\infty}(t^\star) = \C^{\infty}\v{y}^{\infty}(t^\star),
\end{align}
with the \emph{infinite Carleman collision matrix}
\begin{equation}
\C^\infty = 
\begin{pmatrix}
C^1_1 & C^1_2 & 0 & 0     &\dots&0\\
0     & C^2_2 & C^2_3 & C^2_4 &\dots&0\\
0     & 0     & C^3_3 & C^3_4 &\dots&0\\
\vdots&\vdots & \vdots& \vdots& \ddots &\vdots
\end{pmatrix}, 
\end{equation}
where in each row $k$ there are $k+1$ non-zero blocks.

The nonlinear update rule for a $d$-dimensional vector $\v{g}(t^\star)$ from Eqs.~\eqref{eq:LBE_col_shift_matrix}-\eqref{eq:LBE_str_shift_matrix},
\begin{equation}
    \v{g}(t^\star+1) = S\left[\left( I+ F_1\right)\v{g}(t^\star) + F_2 \v{g}(t^\star)^{\otimes 2}\right],
\end{equation}
can thus be formally replaced by a linear update rule for an infinite-dimensional vector $\v{y}^{\infty}(t^\star)$,
\begin{equation}
    \v{y}^{\infty}(t^\star+1) = \S^\infty \C^\infty \v{y}^{\infty}(t^\star). 
\end{equation}
This infinite system can be then truncated at level $N_C$, so that we are dealing with finite \emph{Carleman vectors}
\begin{aligns}
    \v{y}^{C}(t^\star) & := [\v{y}_1^C(t^\star),\v{y}_2^C(t^\star),\dots,\v{y}_{N_C}^C(t^\star)],\\
    \v{y}(t^\star) & := [\v{y}_1(t^\star),\v{y}_2(t^\star),\dots,\v{y}_{N_C}(t^\star)],
    \label{eq:truncatedCarlemanstate}
\end{aligns}
of dimension
\begin{equation}
    d_C:= d+d^2+\dots +d^{N_C} = \frac{d(d^{N_C}-1)}{d-1},
\end{equation}
which carry information about all powers of the shifted incompressible LBE evolution up to the $N_C$-th power. The collision and streaming steps are now captured by the \emph{Carleman collision matrix} $\C$ and \emph{Carleman streaming matrix} $\S$, which are given by $\C^\infty$ and $\S^\infty$ restricted to the first $N_C$ blocks. Due to the block-diagonal form of $\S$, coming from linearity of the original transformation, it can be implemented exactly for a finite truncation order $N_C$, i.e., no truncation error is introduced. The collision step, on the other hand, has $\min\{k+1,N_C-k+1\}$ non-zero blocks in each row $k$ and, since truncation removes some of the terms as compared to the original matrix $\C^\infty$, a truncation error is typically expected, i.e., the evolution obtained from a truncated linear system will generally differ from the original shifted LBE dynamics. We will discuss the Carleman truncation error $\epsilon_C$, and in particular its dependence on the truncation order $N_C$, in Sec.~\ref{sec:error}. The update rule for the Carleman vector is thus given by
\begin{equation}
    \label{eq:LBE_recurrence}
    \v{y}(t^\star+1) = \S\C \v{y}(t^\star),
\end{equation}
and the above recurrence relation has a clear solution given by
\begin{equation}
    \label{eq:carlemanSolution}
    \v{y}(t^\star) = (\S\C)^{t^\star} \v{y}_{\mathrm{ini}},
\end{equation}
where $\v{y}_{\mathrm{ini}}$ is the initial Carleman vector obtained from the initial condition $\v{g}(0)$:
\begin{equation}
    \label{eq:carleman_ini}
    \v{y}_{\mathrm{ini}} = [\v{g}(0),\v{g}(0)^{\ot 2},\dots,\v{g}(0)^{\ot N_C}].
\end{equation}
Thus, given the initial state $\v{g}(0)$ of the original fluid system, its state after $T^\star$ collision and streaming steps is given by
\begin{equation}
    \v{g}(T^\star) \approx_{\epsilon_C} \left[(\S\C)^{T^\star} \v{y}_{\mathrm{ini}}\right]_1
\end{equation}
where $\approx_{\epsilon_C}$ captures the introduced truncation error.


\subsection{Linear system formulation}
\label{sec:linear_system}

Since the heart of our quantum algorithm for the lattice Boltzmann equation is based on a quantum linear solver, in this section we explain how to cast the Carleman state evolved under discrete-time dynamics from Eq.~\eqref{eq:carlemanSolution} as a solution to a linear system of equations. More precisely, we first construct a linear system $A_H$, whose solution $\v{Y}_H$ carries information about the whole history of the evolution described by the shifted incompressible LBE (i.e., at all discrete time steps between 0 and $T^\star$). We then construct another system $A_F$, whose solution $\v{Y}_F$ can be used to recover the final state of the evolution (i.e., after $T^\star$ time steps). In what follows, we will use a shorthand notation of $A$ and $\v{Y}$ to denote either $A_H$ and $\v{Y}_H$, or $A_F$ and $\v{Y}_F$, when the discussed results for both systems are identical.

\paragraph*{History state.}

The shifted incompressible LBE problem with $T^\star$ streaming and collision steps can be written in the form of the following linear system of dimension $d_C(T^\star+1)$:
\begin{align}
    \label{eq:AH}
        \begin{pmatrix}
            I& 0 &  \dots&0&0\\
            -\S\C   & I&  \dots&0&0\\
            \vdots&\vdots &   \ddots&\vdots&\vdots \\
            0&0&\dots&I&0\\
            0&0&\dots&-\S\C & I
        \end{pmatrix}
        \begin{pmatrix}
            \v{y}(0)\\
            \v{y}(1)\\
            \vdots\\
            \v{y}(T^\star-1)\\
            \v{y}(T^\star)
        \end{pmatrix} = 
        \begin{pmatrix}
            \v{y}_{\mathrm{ini}}\\
            0\\
            \vdots\\
            0\\
            0
        \end{pmatrix},
\end{align}
or in a more compact form as
\begin{equation}
    \label{eq:linearHistory}
    A_H\v{Y}_H=\v{b}.    
\end{equation}
The vector $\v{Y}_H$ is the \emph{Carleman history state}, carrying information about Carleman vector at all times during the simulated evolution. It can be verified by direct calculation that the inverse of $A_H$ is given by
\begin{align}
    \label{eq:AHInv}
        A_H^{-1}=\begin{pmatrix}
            I           & 0           & \dots  & 0       & 0\\
            \S\C        & I           & \dots  & 0       & 0\\
            \vdots      &\vdots &  \ddots & \vdots  &\vdots\\
            (\S\C)^{T^\star-1} & (\S\C)^{T^\star-2}  & \dots & I & 0\\
            (\S\C)^{T^\star} & (\S\C)^{T^\star-1}  & \dots & \S\C & I
        \end{pmatrix},
\end{align}   
which is indeed in accordance with Eq.~\eqref{eq:carlemanSolution}. We thus see the solution to the set of linear equations above,
\begin{equation}
    \v{Y}_H = A_H^{-1} \v{b} = \left(\v{y}_{\mathrm{ini}},\S\C \v{y}_{\mathrm{ini}},\dots,  (\S\C)^{T^\star} \v{y}_{\mathrm{ini}} \right)^\top,
    \label{eq:YH}
\end{equation}
carries information about the whole evolution described by the shifted incompressible LBE. 

\paragraph*{Final state.}

For the recovery of just the final state, instead of the whole history state, we will consider an alternative linear system, inspired by Ref.~\cite{berry2014high}, with additional $(2^{W}-1)(T^\star+1)$ idling steps after the original $T^\star+1$ evolution steps (where $W$ is an integer to be determined). More precisely, we consider
\begin{equation}
    \label{eq:linearFinal}
    A_F\v{Y}_F=\v{b},
\end{equation}
with
\begin{align}
    \label{eq:AF}
        A_F=\left(
        \begin{array}{ccccccccccc}
            I& 0 &  \dots&0&0& 0& \dots & 0 & 0\\
            -\S\C   & I&  \dots&0&0& 0& \dots & 0 & 0\\
            \vdots&\vdots &  \ddots&\vdots&\vdots&\vdots &\dots & 0 & 0\\
            0&0&\dots&I&0& 0& \dots & 0 & 0\\
            0&0&\dots&-\S\C&I& 0& \dots & 0 & 0\\
            0&0&\dots& 0& -I & I & \dots & 0 & 0\\
            \vdots&\vdots&\dots& \vdots& \vdots & \vdots & 
            \ddots&\vdots&\vdots \\
            0& 0  & \dots & 0 & 0&0&\dots& -I&I
        \end{array}
        \right).
\end{align}
The solution of the above is
\begin{align}
    \v{Y}_F &= A_F^{-1} \v{b} = \Bigg(\v{y}_{\mathrm{ini}},\S\C \v{y}_{\mathrm{ini}}, ,\dots, (\S\C)^{T^\star} \v{y}_{\mathrm{ini}},\nonumber\\
    &\qquad\qquad\qquad\underbrace{(\S\C)^{T^\star} \v{y}_{\mathrm{ini}}, \dots, (\S\C)^{T^\star} \v{y}_{\mathrm{ini}}}_{(2^{W}-1) (T^* +1) \mathrm{ times}} \Bigg)^\top.\label{eq:YF}
\end{align}
Clearly, the first $T^\star+1$ blocks of $\v{Y}_F$ coincide with the history state $\v{Y}_H$, but the remaining $(2^{W}-1)(T^\star+1)$ blocks contain the Carleman vector corresponding to the state of the system at the final time $T^\star$. As we will explain in Sec.~\ref{sec:recovering_final}, this will help our quantum algorithm to recover the final state more efficiently than from the history state. Also, the inverse $A_F^{-1}$ has its first $T^\star+1$ rows the same as for $A_H^{-1}$ (see Eq.~\eqref{eq:AHInv}), after which there follow rows, where each row $k$ is given by the last row of $A_H^{-1}$ appended by $k$ identities. We note that the choice of $(2^{W}-1)(T^\star+1)$ idling steps was made for practical implementation reasons (it can be achieved using $W$ extra qubits) and, as we shall see, the value of $W$ can be chosen to control the probability of recovering the final state.


\section{Quantum algorithm implementation}
\label{sec:algorithm}

In the previous section, we explained how to rephrase the solution of the shifted incompressible lattice Boltzmann equation as a solution $\v{Y}$ of a set of linear equations described by either Eq.~\eqref{eq:linearHistory} or Eq.~\eqref{eq:linearFinal}. In this section, we will construct a quantum algorithm that first employs a quantum linear solver to prepare the solutions $\v{Y}$ coherently encoded in the amplitudes of a quantum state; and then post-processes them to extract classical information about the fluid.

First, in Sec.~\ref{sec:encoding}, we will describe how we encode the data in quantum registers, and then we will present the following three steps of our algorithm:

\begin{enumerate}

    \item \textbf{Input.} In order for a quantum linear solver to output the solution to $A\v{Y}=\v{b}$, the matrix $A$ and vector $\v{b}$ must be passed to it in a way that a quantum computer can deal with. We achieve this in Sec.~\ref{sec:block_encodings} by constructing a unitary block-encoding $U_A$ of $A$ and a unitary state preparation $U_{\v{b}}$ of $\v{b}$.

    \item \textbf{Processing.} In Sec.~\ref{sec:QLS}, we explain how to use the quantum linear solver, with the best theoretical guarantees currently available, to bound the query complexity to the unitary block-encoding $U_{A}$ and unitary state preparation $U_{\v{b}}$ to obtain a state $\epsilon_Q$ away from the normalized solution $\ket{\v{Y}}$.
    
    \item \textbf{Output.} In Sec.~\ref{sec:measurement}, we describe necessary post-processing steps in the form of measurements performed on $\ket{\v{Y}}$ that allow one to extract classical information about the fluid system.
\end{enumerate}    


\subsection{Data encoding}
\label{sec:encoding}

We will describe our data encoding process in a few steps, starting from the smallest data structure of a single state vector $\v{g}$ encoded as a quantum state $\ket{\v{g}}$; then we will explain how to encode a Carleman vector $\v{y}$ from Eq.~\eqref{eq:truncatedCarlemanstate} into a quantum state $\ket{\v{y}}$; and finally how to encode the full data $\v{Y}$ including the time register  (meaning $\v{Y}_H$ in Eq.~\eqref{eq:linearHistory} or $\v{Y}_F$ in Eq.~\eqref{eq:linearFinal}) into a quantum state $\ket{\v{Y}}$. When writing the encodings, we will assume $D=3$, and the smaller dimensions can be straightforwardly obtained by dropping the unnecessary registers. The total number of spatial points is $N=N_xN_yN_z$, with $N_i$ corresponding to the number in the $i$'th direction. To simplify expressions, when computing qubit counts we give expressions like $\log N_x$, which should be taken to mean $\lceil \log N_x \rceil$, namely we round up to the nearest integer. 

To encode the shifted LBE state vector into a quantum state $\ket{\v{g}}$, we use three spatial registers composed of $\log N_x$, $\log N_y$, and $\log N_z$ qubits, and three velocity registers, each composed of two qubits:
\begin{align}
    \ket{\v{g}(t^\star)} = \sum_{r_x^\star,r_y^\star,r_z^\star=1}^{N_x,N_y,N_z} 
    \sum_{\substack{v_x,v_y,v_z\in\\\{00,01,10\}}}& \frac{g_{m(v_x,v_y,v_z)}(\v{r}^\star,t^\star)}{\sqrt{\N}}\nonumber\\
    & \ket{r_x^\star,r_y^\star,r_z^\star}\ot \ket{v_x,v_y,v_z},
\end{align}
where $\N$ is the normalization, $r_i^\star$ are discrete spatial positions along axis $i$ in the binary representation, whereas $(v_x,v_y,v_z)$ encode velocity indices according to:
\begin{equation}
    v_i = \left\{\begin{array}{ll}
         00 & :\quad (\v{e}_m^\star)_i ~=0, \\
         10 & :\quad (\v{e}_m^\star)_i ~=-1, \\
         01 & : \quad (\v{e}_m^\star)_i ~=1.
    \end{array}
    \right.
\end{equation}
Note that this velocity encoding is almost optimal, as it requires $2D$ qubits, while the optimal one uses 2, 4, and 5 for $D\in\{1,2,3\}$. The total number of qubits required to encode an LBE state vector is thus $\log N+2D$.

Next, we need to encode the Carleman state $\ket{\v{y}}$. While the original Carleman vector $\v{y}$ is composed of a direct sum of vectors $\v{y}_1,\dots,\v{y}_{N_C}$, for practical reasons, it is better to use a tensor product encoding~\cite{liu2021efficient}. More precisely, we add an extra Carleman block register $\ket{n_C}$ composed of $\log N_C$ qubits and pad the lower blocks with zeros:
\begin{equation}
    \ket{\v{y}(t^\star)} = \sum_{n_C=1}^{N_C} \!\frac{\| \v{y}_{n_C}(t^\star)\|}{\|\v{y}(t^\star)\|} \ket{n_C} \ot \ket{\v{y_{n_C}}(t^\star)}\ot \ket{0}^{\ot N_C-n_C}.
\end{equation}
Above, $\ket{\v{y_{n_C}}}$ is composed of $n_C$ single LBE state vector registers (i.e., $n_C (\log N+2D)$ qubits), and a single $\ket{0}$ is composed of a single LBE state vector register, i.e., of $\log N +2D$ qubits. Since $\v{y}_{n_C}$ lives in a vector space consisting of $n_C$ copies of the space for the LBE state vector, we will denote its single state components by $\v{y}^{[l]}$ for $l\in\{1,\dots,n_C\}$.

Using the proposed tensor product encoding, the structure of the Carleman matrices changes, so that all their Carleman blocks become square matrices with a fixed dimension. More precisely, the action of the Carleman collision matrix can be extended to
\begin{align}
    \C &=  \sum_{k=1}^{N_C} \sum_{l=k}^{\min\{2k,N_C\}} \ketbra{k}{l} \ot C^k_l \ot \ket{0}^{\ot l-k} \ot I^{\ot N_C-l} \nonumber\\
    &=: \sum_{k=1}^{N_C} \sum_{l=k}^{\min\{2k,N_C\}} \ketbra{k}{l} \ot \bar{C}^k_l \ot I^{\ot N_C-l},\label{eq:collision_encoded}
\end{align}
so that every rectangular matrix $C^k_l$ becomes a square matrix $\bar{C}^k_l$ by adding in the output $l-k$ zero states, and then we extend it to a square matrix with a fixed dimension for all blocks by padding it with additional $N_C-l$ identity matrices. Similarly, the Carleman streaming matrix is extended to
\begin{equation}
    \S =  \sum_{k=1}^{N_C} \ketbra{k}{k} \ot S^{\ot k} \ot I^{\ot N_C-k}.
\end{equation}   

Finally, we need to encode states $\ket{\v{Y}_H}$ and $\ket{\v{Y}_F}$. We achieve this by adding an extra time register $t^\star$ composed of $\log(T^\star+1)$ qubits and, in the case of $\ket{\v{Y}_F}$, an extra waiting register $w$ composed of $W$ qubits:
\begin{aligns}
    \ket{\v{Y}_H} =& \sum_{t^\star=0}^{T^\star} \frac{\|\v{y}(t^\star)\|}{\| \v{Y}_H\|}  \ket{t^\star}\ot \ket{\v{y}(t^\star)}, \\
    \ket{\v{Y}_F} =& \frac{1}{\|\v{Y}_F\|}\left(\sum_{t^\star=0}^{T^\star}\! \|\v{y}(t^\star)\|\ket{0}\ot \ket{t^\star}\ot \ket{\v{y}(t^\star)} \right.\nonumber\\
    &\left. + \|\v{y}(T^\star)\|\sum_{w=1}^{2^W-1} \sum_{t^\star=0}^{T^\star} \ket{w}\ot \ket{t^\star}\ot \ket{\v{y}(T^\star)}\right).
\end{aligns}
Summarizing, the full data register is composed of the following:
\begin{itemize}
    \item A single time register $t^\star$ consisting of $\log (T^\star+1)$ qubits.
    \item A single waiting register $w$ consisting of $W$ qubits (absent in the case of history state).
    \item A single Carleman block register $n_C$ consisting of $\log N_C$ qubits.
    \item $N_C$ single state vector registers $\v{y}^{[k]}$, each composed of $D$ spatial registers $r_i^\star$ (consisting of $\log N_i$ qubits each) and $D$ velocity registers $v_i$ (consisting of $2$ qubits each).
\end{itemize}

Adding this all together, we get that the number $n_D$ of qubits in the data register is given by
\begin{equation}
    n_D = \log (T^\star+1) + W + \log N_C + N_C(\log N + 2D).
\end{equation}
Using the simulation parameter choice from Eqs.~\eqref{eq:parameter_choice1}-\eqref{eq:parameter_choice2}, with $N_x=N_y=N_z$, we can represent it as a function of the Reynolds number $\re$:
\begin{equation}
    \label{eq:data_qubits}
    n_D = \beta \left[ D\left(N_C+\frac{1}{2}\right)+1\right] \log \re + 2D N_C +\log N_C +W,
\end{equation}
and so the size of the data register scales efficiently with the Reynolds number $\re$ and the Carleman truncation order $N_C$. For example, for a $D=3$ system with Reynolds number $\re=10^8$, parameter $\beta=3/4$ to resolve the Kolmogorov microscale, $W=10$ to extract the final state with high probability, and a relatively high truncation order $N_C=10$, one needs 722 qubits.


\subsection{Unitary block-encodings and state preparations}
\label{sec:block_encodings}

The quantum algorithm involves a range of non-unitary matrices that define the problem instance, as well as state preparation data. Quantum computers are primarily concerned with quantum circuits that generate unitary operations. We must therefore encode the matrices defining the problem within a unitary context that the quantum computer can process. This is achieved through the concept of a \emph{unitary block-encoding} of a matrix $A$. In what follows, we summarize key results that are central to the algorithm construction, while all the details can be found in Appendix~\ref{app:BE}. 

The construction of block-encodings of $A_H$ and $A_F$ can be divided into two steps. In the first step, we assume access to the block-encodings of the collision matrices, $I+F_1$ and $F_2$, that are the basic building elements of the Carleman collision matrix $\C$. Then, using standard techniques of linear combination of unitaries (LCU), we can block-encode both $A_H$ and $A_F$ with a block-encoding prefactor $\alpha_A$ and using the number of ancillary qubits $n_A$, with:
\begin{aligns}
    \label{eq:prefactor_final_main}
    \alpha_{A} =& 1+ \sum_{l=0}^{\lfloor\frac{N_C}{2}\rfloor} \binom{N_C-l}{l} \alpha_{ I+F_1}^{N_C-2l} \alpha_{F_2}^{l},\\
    n_A =& \max_{l\in\{0,\dots, \lfloor N_C/2\rfloor\}}  \left[\log\binom{N_C-l}{l} + l(4 D-1)\right] \nonumber\\
    &+ N_C + 1 + \log\left(\left\lfloor \frac{N_C}{2}\right\rfloor+1\right),
    \label{eq:ancilla_qubits_main}
\end{aligns}
where $\alpha_{I+F_1}$ and $\alpha_{F_2}$ are the block-encoding prefactors for the collision matrices.

The second step then consists of constructing block-encodings for these collision matrices. Recalling that they act locally in space (see Eqs.~\eqref{eq:F1_loc}-\eqref{eq:F2_loc}), we only need to block-encode small $Q\times Q$ and $Q\times Q^2$ matrices $I+\tilde{F}_1$ and $\tilde{F}_2$. For $I+\tilde{F}_1$, we employ the block-encoding based on singular value decomposition (SVD), resulting in the block-encoding prefactor given by largest singular value element of $I+\tilde{F}_1$. By explicit calculation we then end up with
\begin{align}
    \alpha_{I+F_1}&=\frac{1}{\sqrt{2^{D+1}}\,\taus} 
  \Bigg(
    3^{D}
    + 2^{D+1}\left(\bar{\tau}^{\star 2} - \taus\right)\nonumber\\&\qquad\qquad\qquad\quad
    + \sqrt{
        9^{D}
        + 4 \cdot 6^{D}\left(\bar{\tau}^{\star 2} - \taus\right)
      }
  \Bigg)^{\frac{1}{2}}.\label{eq:alpha_F1_main}
\end{align}
As for the block-encoding of $\tilde{F}_2$, we use a method based on higher-order singular value decomposition (HOSVD). More precisely, we decompose the $Q \times Q^2$ matrix $\tilde{F}_2$ as
\begin{equation}
       \tilde{F}_2 = L_2 \Sigma_2 (R_2\ot R_2)^\dagger,
\end{equation}
where $L_2$ and $R_2$ are $Q\times Q$ unitary matrices, while $\Sigma_2$ is a $Q\times Q^2$ matrix with as small rank as possible. This way, we end up with a block-encoding of $F_2$ with the prefactor
\begin{align}
    \alpha_{F_2} &= \frac{2}{3\taus}\left(\frac{3}{\sqrt{2}}\right)^D\sqrt{D+2}.\label{eq:alpha_F2_main}
\end{align}
Using Eqs.~\eqref{eq:alpha_F1_main}~and~\eqref{eq:alpha_F2_main} in Eq.~\eqref{eq:prefactor_final_main}, we end up with the final expression for the block-encoding prefactor.

The final missing component before we can run a quantum linear solver is a unitary state preparation $U_{\v{b}}$ of $\ket{\v{b}}$. For that, we assume access to a unitary state preparation $U_{\psi_{\mathrm{ini}}}$ of the initial state $\ket{\psi_{\mathrm{ini}}}$ for the shifted LBE problem:
\begin{equation}
    \ket{\psi_{\mathrm{ini}}} \propto \sum_{\v{r}^\star} \sum_{m=1}^Q g_{m}(\v{r}^\star,0) \ket{\v{r}^\star,m},
\end{equation}
where we use a concise notation for the velocity registers. In this work, we disregard walls, inlets and outlets (which we analyze elsewhere~\cite{jennings2026simulating}), but consider non-trivial initial states given by $g_{m,\v{r}^\star}(0)$ describing e.g., Taylor-Green vortices and vortex dipoles, see Sec.~\ref{sec:error} for details. The complexity cost of constructing $U_{\psi_{\mathrm{ini}}}$ will clearly depend on the initial state $\ket{\psi_{\mathrm{ini}}}$ considered. In the worst-case scenario, this overhead may become sufficiently large to undermine the quantum advantage associated with simulating the time evolution. For example, for random or unstructured states, where each entry must be loaded independently, the complexity of state preparation may scale as $O(N)$ with the number of lattice points. It is therefore important to emphasize that any claims regarding the end-to-end complexity are contingent upon efficient initial-state preparation. In practice, however, the initial velocity fields of interest typically exhibit an explicit functional dependence on position, meaning that the corresponding initial states are highly structured, making this assumption well justified. Thus, we will assume here that the complexity cost of $U_{\psi_{\mathrm{ini}}}$ is negligible compared to the cost of unitary block-encoding $U_A$ of the linear system. 

Moreover, one can also show that the complexity cost of $U_{\v{b}}$ does not differ much from that of $U_{\psi_{\mathrm{ini}}}$, and so can be neglected. Namely, in order to prepare $\ket{\v{b}}$ using $U_{\psi_{\mathrm{ini}}}$, one can employ a result from Ref.~\cite{liu2021efficient}, where Lemma~4 shows how to construct $U_{\v{b}}$ using $O(N_C)$ queries to $U_{\psi_{\mathrm{ini}}}$ supplemented by $O(N_C)$ single-qubit gates. Hence, each query to $U_{\v{b}}$ translates to $O(N_C)$ queries to $U_{\psi_{\mathrm{ini}}}$ (in fact, exactly $N_C$, so there are no hidden large constant prefactors). 


\subsection{Quantum linear solver}
\label{sec:QLS}

Once we have encoded the data about our fluid dynamics problem in unitaries $U_A$ and $U_{\v{b}}$, we can use a quantum linear solver to prepare a quantum state encoding the solution to our problem in its amplitudes. We shall use the algorithm  from Ref.~\cite{dalzell2024shortcut}, since at the time of writing, it provides the lowest rigorously guaranteed upper bound to the number of required applications of (controlled) $U_A$, $U_{\v{b}}$ and their inverses, i.e., the best upper bound for the query complexity $q_Q$. However, in order to get the correct query complexity upper bound, we need to carefully translate our problem to the standardized linear problem studied in Ref.~\cite{dalzell2024shortcut}, which we do in Appendix~\ref{app:QLS}. This leads us to the following upper bound on~$q_Q$:
\begin{align}    
        q_Q \leq &~\frac{\alpha_A}{\|A\|} \Bigg(56.0 \kappa_A +1.05 \kappa_A \ln\Bigg(\frac{(1-\epsilon_Q^2)^{\frac{1}{2}}}{\epsilon_Q}\Bigg)\Bigg)\nonumber\\
        &~+2.78\ln\left(\frac{\alpha_A}{\|A\|}\kappa_A\right)^3+3.17.
        \label{eq:queryBound}
\end{align}

Upon successful execution, the quantum linear solver algorithm outputs a coherent encoding of the solution to $A\v{Y}=\v{b}$ as normalized quantum states. More precisely, for the history state and final state output we get:
\begin{aligns}
   \ket{\v{Y}_H}&=\frac{1}{\|\v{Y}_H\|}\sum_{t^\star=0}^{T^\star} \ket{t^\star}\ot \v{y}(t^\star)\nonumber\\
   &=\frac{1}{\|\v{Y}_H\|}\sum_{t^\star=0}^{T^\star} \ket{t^\star}\ot (\S\C)^{t^\star} \v{y}_{\mathrm{ini}},\label{eq:carl_history}\\ 
   \ket{\v{Y}_F}&=\frac{1}{\|\v{Y}_F\|}\Bigg(\sum_{t^\star=0}^{T^\star} \ket{0}\ot\ket{t^\star}\ot \v{y}(t^\star)\nonumber\\
   &\qquad+\sum_{w=1}^{2^W-1}\sum_{t^\star=0}^{T^\star} \ket{w}\ot\ket{t^\star}\ot \v{y}(T^\star)\Bigg)\nonumber\\
   &=\frac{1}{\|\v{Y}_F\|}\Bigg(\sum_{t^\star=0}^{T^\star} \ket{0}\ot\ket{t^\star}\ot (\S\C)^{t^\star} \v{y}_{\mathrm{ini}} \nonumber\\
   &\qquad+\sum_{w=1}^{2^W-1}\sum_{t^\star=0}^{T^\star} \ket{w}\ot\ket{t^\star}\ot (\S\C)^{T^\star}\v{y}_{\mathrm{ini}}\Bigg).\label{eq:carl_final}
\end{aligns}
These two outputs describe the coherent encoding of the dynamics, provided in the form of an entangled quantum state over the registers. We next turn to post-processing and information-extraction.


\subsection{Measurements and data extraction}
\label{sec:measurement}


\subsubsection{Recovering final and history LBE states}
\label{sec:recovering_final}

We start by explaining how to recover the Carleman state of the system at the final time $T^\star$ from the state $\ket{\v{Y}_F}$. To achieve this, we perform a measurement of both the waiting register $w$ and time register $t^\star$, discarding them afterwards. From Eq.~\eqref{eq:carl_final}, it is clear that if the $w$ outcome is anything but all zero state, then irrespectively of the $t^\star$ outcome, we collapse the system onto the final Carleman state:
\begin{equation}
    \ket{\v{y}(T^\star)} = \frac{1}{\|\v{y}(T^\star)\|} (\S\C)^{T^\star}\v{y}_{\mathrm{ini}}.
\end{equation}
Introducing
\begin{aligns}
    \label{eq:NH1}
    \N_H &= \sum_{t^\star=0}^{T^\star} \|(\S\C)^{t^\star}\v{y}_{\mathrm{ini}}\|^2,\\
    \N_F &=(T^\star+1)\|(\S\C)^{T^\star}\v{y}_{\mathrm{ini}}\|^2,\label{eq:NH2}
\end{aligns}
we see that this happens with probability
\begin{equation}
    \label{eq:pF}
    p(\v{y}(T^\star)) = \frac{1}{1+ \frac{1}{2^W-1}\cdot\frac{\N_H}{\N_F}}.
\end{equation}
The above approaches 1 exponentially with increasing the number $W$ of qubits in the waiting register. Moreover, under reasonable assumptions on the final state solution of the flow, the expected value of $W$, ensuring a probability of success of at least $1/2$, is constant, independent of the system size. More precisely, as long as the solution yields nonzero velocity field on a constant fraction of $N$ lattice sites (expected, e.g., for steady states of non-trivial flows), the ratio $\N_H/\N_F$ will be of order $O(1)$, and so a constant $W$ will guarantee $p>1/2$. For example, when $\N_H \approx \N_F$, a single waiting qubit gives a constant probability of success, approximately equal to $1/2$. Given the expected high success probability, in practice a repeat-till-success strategy appears to be more convenient than a strategy based on amplitude amplification, although the latter is also possible.

Next, both for the Carleman history state $\ket{\v{Y}_H}$ and the final Carleman state $\ket{\v{y}(T^\star)}$, we measure the Carleman block register $n_C$. We keep the outcome only if $n_C=1$, and discard all the state vector registers except for the first one. Thus, we collapse the system onto
\begin{aligns}
    \ket{\psi_F} & := \ket{\v{y}_1(T^\star)}   \approx \ket{\v{g}(T^\star)}, \\
    \ket{\psi_H} & :=  \sum_{t^\star=0}^{T^\star} \ket{t^\star}\ot \ket{\v{y}_1(t^\star)} \approx \sum_{t^\star=0}^{T^\star} \frac{\|\v{g}(t^\star)\|}{\N_{\v{g}}} \ket{t^\star}\ot \v{g}(t^\star),
\end{aligns}
where the approximation error comes from the Carleman truncation error $\epsilon_C$ and QLS error $\epsilon_Q$, whereas $\N_{\v{g}}$ is the normalization factor. Assuming $\|\v{y}_k\|\approx \|\v{g}\|^k$, one obtains $\ket{\psi_F}$ with probability
\begin{aligns}
    p(\psi_F)  \approx \frac{\|\v{g}(T^\star)\|^2}{\sum_{k=1}^{N_C} \|\v{g}(T^\star)\|^{2k}} = \frac{1-\|\v{g}(T^\star)\|^2}{1-\|\v{g}(T^\star)\|^{2N_C}}.
\end{aligns}
Moreover, assuming that the norm of $\v{g}(t^\star)$ does not change significantly during the evolution, so that $\|\v{g}(t^\star)\| \approx \|\v{g}(T^\star)\|$ for all $t^\star$, the probability $p(\psi_H)$ of obtaining $\ket{\psi_H}$ is the same, i.e., $p(\psi_H)\approx p(\psi_F)$. Crucially, note that differently from prior proposals, our introduction of a shifted LBE model and our choice of parameters from Eqs.~\eqref{eq:parameter_choice1}-\eqref{eq:parameter_choice2} guarantees $\|\v{g}\|<1$. Hence, we get constant success probabilities: in particular, when $\|\v{g}\|\approx 1$, the probabilities become approximately $1/N_C$.\footnote{Also note that one can improve this constant probability quadratically by using amplitude amplification.} To show that this solves the problem with challenge C4 from Sec.~\ref{sec:problem_info}, rather than just shifts it into challenge C2, we need to analyze how our choice of parameters affects the condition number $\kappa_A$, which we do in Sec.~\ref{sec:condition}.


\subsubsection{Extracting the drag force}
\label{sec:info_extraction}

We now proceed to explaining how to perform measurements on  $\ket{\psi_F}$ to estimate the drag force at the final time $T^\star$ on some boundary wall. Our analysis will closely follow the proposal presented in Ref.~\cite{penuel2024feasibility}, with some small changes due to us working with the shifted LBE vector $\v{g}$, and not the LBE vector $\bar{\v{f}}$. Using the the standard momentum exchange approach to the drag force~\cite{kruger2016lattice}, the local momentum density change $\Delta \v{p}^\star(\v{r^\star},T^\star,m)$, at time $T^\star$ and at the fluid node $\v{r^\star}$, with the wall node at $\v{r^\star}+\v{e}^\star_m$ in lattice units is given by
\begin{equation}
    \!\!\! \Delta \v{p}^\star(\v{r^\star},T^\star,m) = \left(\bar{f}_m(\v{r}^\star,T^\star)+\bar{f}_{-m}(\v{r}^\star,T^\star)\right) \v{e}_m^\star,\!
\end{equation}
where $-m$ is the velocity index corresponding to discrete velocity $-\v{e}_m^\star$. Using the definition of the shifted LBE vector, the above can be rewritten as
\begin{align}
    \Delta \v{p}^\star(\v{r^\star,T^\star,m}) =& \left( g_m(\v{r}^\star,T^\star)+g_{-m}(\v{r}^\star,T^\star)\right) \v{e}_m^\star\nonumber\\
    &+\Delta \v{p}^\star_0(m),
\end{align}
with  $\Delta \v{p}^\star_0 (m) = 2 w_m \v{e}_m^\star$
denoting the zero-velocity equilibrium value of momentum density change.

Next, let us encode the wall boundary of interest using the $\mathcal{W}$ function from Eq.~\eqref{eq:boundary}. Since in lattice units the volume of one spatial cell is 1, the total momentum exchange $\Delta \v{P}^\star$ at the boundary is just the sum over local momentum density changes:
\begin{equation}
    \Delta \v{P}^\star (T^\star) = \sum_{\v{r}^\star} \sum_m \mathcal{W}_{\v{r}^\star+\v{e}^\star_m}(1-\mathcal{W}_{\v{r}^\star}) \Delta \v{p}^\star(\v{r}^\star,T^\star,m). 
\end{equation}
Now, the drag force $\v{F}^\star(T^\star)$ can be calculated by dividing the total momentum change by the time it takes. Since in lattice units one time step has length 1, the drag force is simply equal to $\Delta \v{P}^\star(T^\star)$, which can be explicitly written as
\begin{align}
    \label{eq:drag_force}
    \v{F}^\star (T^\star)&= \v{F}_0^\star + \sum_{\v{r}^\star} \sum_m \mathcal{W}_{\v{r}^\star+\v{e}^\star_m}(1-\mathcal{W}_{\v{r}^\star})\nonumber\\
    &\qquad\qquad\left( g_m(\v{r}^\star,T^\star)+g_{-m}(\v{r}^\star,T^\star)\right) \v{e}_m^\star,
\end{align}
where
\begin{equation}
    \v{F}_0^\star = \sum_{\v{r}^\star} \sum_m \mathcal{W}_{\v{r}^\star+\v{e}^\star_m}(1-\mathcal{W}_{\v{r}^\star}) 2 w_m \v{e}_m^\star.
\end{equation}
Hence, if we can estimate $\v{F}^\star(T^\star)$, we can get the drag force in physical units by simply multiplying $\v{F}^\star(T^\star)$ by $\Delta x^2/\Delta t$ and by the reference mass (i.e., the mass of a single unit cell).

We shall assume that $\v{F}^\star_0$ can be efficiently computed classically, and we will use $\ket{\psi_F}$ to estimate the value of the second term on the right hand side of Eq.~\eqref{eq:drag_force}, i.e., to get its components:
\begin{align}
    \mathcal{F}_k &:= \sum_{\v{r}^\star} \sum_m \mathcal{W}_{\v{r}^\star+\v{e}^\star_m} (1-\mathcal{W}_{\v{r}^\star})\nonumber\\
    &\qquad\qquad\left( g_m(\v{r}^\star,T^\star)+g_{-m}(\v{r}^\star,T^\star)\right) (\v{e}_m^\star)_k,
\end{align}
where $k\in\{x,y,z\}$. To that end, for each component $k$ we prepare the following quantum state encoding the boundary of interest:
\begin{align}
    \ket{\mathcal{W}_k} &= \frac{1}{\sqrt{\N_{\mathcal{W}_k}}}\sum_{\v{r}^\star} \sum_{m} \mathcal{W}_{\v{r}^\star+\v{e}^\star_m}(1- \mathcal{W}_{\v{r}^\star})(\v{e}_m^\star)_k \nonumber\\
    &\qquad\qquad\qquad\qquad\ket{\v{r}^\star}\ot (\ket{m}+\ket{-m}),
\end{align}
with $\N_{\mathcal{W}_k}$ denoting normalization constants that can be efficiently precomputed, and the velocity registers written in a concise notation. This quantum state can be efficiently prepared  when certain simplicity criteria on the boundary function are met, e.g., via rejection sampling techniques \cite{lemieux2024quantum}. Broadly speaking, when $\mathcal{W}$ is an analytically defined and relatively simple function, this state preparation will have a subleading complexity cost compared to the rest of the algorithm. 

Now, recalling that
\begin{equation}
    \ket{\psi_F}\approx \ket{\v{g}(T^\star)} = \frac{1}{\|\v{g}(T^\star)\|}\sum_{\v{r}^\star} \sum_{m} g_{m}(\v{r}^\star,T^\star) \ket{\v{r}^\star}\ot \ket{m},
\end{equation}
we see that the overlaps between $\ket{\psi_F}$ and $\ket{\mathcal{W}_k}$ are given by
\begin{equation}
    \braket{\mathcal{W}_k|\psi_F} \approx \frac{\F_k}{\|\v{g}(T^\star)\|\sqrt{\N_{\mathcal{W}_k}}}.
\end{equation}
Note that one can efficiently compute $\N_{\mathcal{W}_k}$ and estimate $\|\v{g}(T^\star)\|$ from the QLS algorithm. More precisely, one can estimate $\|\v{Y}_F\|$ and then use the knowledge of probabilities $p(\v{y}(T^\star))$ and $p(\psi_F)$ to get $\|\v{g}(T^\star)\|$. 

Thus, the only thing left to estimate $\F_k$ is to find the overlaps $\braket{\mathcal{W}_k|\psi_F}$, which can be achieved using quantum amplitude estimation. However, one has to note that the boundary states $\ket{\mathcal{W}_k}$ only have support on $O(N_x^{D-1})$ states, in contrast to $\ket{\psi_F}$ having support on $O(N_x^D)$ states. Hence, one expects $\braket{\mathcal{W}_k|\psi_F}$ to be of the order $O(1/\sqrt{N_x})$, which translates into a multiplicative factor overhead scaling as $O(\sqrt{N_x})$. Recalling the parameter choices in Eqs.~\eqref{eq:parameter_choice1}-\eqref{eq:parameter_choice2}, we conclude that, if one wants to extract the drag force as described above, one needs to multiply the complexity of obtaining the state $\ket{\psi_F}$ by 
\begin{equation}
    q_M:= O(\re^{\beta/2}).
\end{equation}
Moreover, one needs to include the complexity of preparing states $\ket{\mathcal{W}_k}$, which is expected to be negligible for simple boundaries like flat walls, but can become more problematic for complex boundaries.


\section{Quantum algorithm performance}
\label{sec:performance}


\subsection{Carleman error convergence}
\label{sec:error}


\subsubsection{Methodology}
\label{sec:error_methodology}

To numerically investigate Carleman error convergence, we first solve the shifted incompressible LBE directly, by simulating $T^\star$ streaming and collision steps from Eqs.~\eqref{eq:LBE_col_shift}-\eqref{eq:LBE_str_shift} with $\v{g}^{\mathrm{eq}}$ given by Eq.~\eqref{eq:eq_shift} (i.e., assuming incompressibility). This way, we obtain the shifted state vector $\v{g}(t^\star)$ for each discrete time $t^\star\in\{0,\dots, T^\star\}$, with $\v{g}(0)$ encoding our selected initial condition (see below). Next, we numerically solve the shifted incompressible LBE employing the discrete Carleman embedding, by simulating Eq.~\eqref{eq:LBE_recurrence} for $T^\star$ time steps. This way, we obtain the Carleman state vector $\v{y}(t^\star)$ for each discrete time $t^\star\in\{0,\dots, T^\star\}$, with $\v{y}(0)=\v{y}_{\mathrm{ini}}$ encoding the initial condition according to Eq.~\eqref{eq:carleman_ini}. From that, we recover the Carleman approximation $\v{y}_1(t^\star)$ of~$\v{g}(t^\star)$.

In previous literature, a measure of Carleman truncation error was proposed that was based on the root mean square error (RMSE) between the actual LBE vector and its Carleman approximation~\cite{sanavio2024lattice,turro2025practical}. Here, we decide to not use this measure for the following reason. Namely, for different Reynolds numbers $\re$, the values of discretization parameters $\Delta x$ and $\Delta t$ vary. Thus, for the same values of RMSE (which are calculated in lattice units), but for two different Reynolds numbers, the actual errors in the physical velocity field will be different. 

Here, we decided to quantify the Carleman truncation error using the average error in the physical velocity field relative to the characteristic velocity $u$ of the problem, e.g., the flow velocity or, in the absence of driving, maximal velocity in the initial state. More precisely, denote the lattice velocities obtained from $\v{g}$ and its Carleman approximation $\v{y}_1$ by $\v{u}^\star(\v{r}^\star,t^\star)$ and $\tilde{\v{u}}^{\star}(\v{r}^\star,t^\star)$, respectively. Then, for all positions $\v{r}^\star$ and time steps $t^\star$, introduce the velocity error vector:
\begin{equation}
    \Delta \v{u}^\star(\v{r}^\star,t^\star) = \tilde{\v{u}}^{\star}(\v{r}^\star,t^\star) - \v{u}^{\star}(\v{r}^\star,t^\star),
\end{equation}
and its version $\Delta \v{u}(\v{r}^\star,t^\star)$ in physical units. We now define the Carleman truncation error $\epsilon_C$ as follows:
\begin{equation}
    \epsilon_C : = \max_{t^\star\in\{1,\dots,T^\star\}} \frac{1}{N}\sum_{\v{r}^\star} \frac{\|\Delta \v{u}\|}{u},
\end{equation}
i.e., as the relative physical velocity error with respect to the flow velocity $u$, averaged over all $N$ lattice sites, and maximized over all times. Note that we have
\begin{align}
    \frac{\|\Delta \v{u}\|}{u} &= \frac{\|\Delta \v{u}^\star\|}{u}\frac{\Delta x}{\Delta t} =  \frac{\|\Delta \v{u}^\star\|}{u}\frac{3\nu}{\Delta x\left(\taus-\frac{1}{2}\right)} \nonumber\\
    &=  \re^{\frac{\beta D}{2}} \| \Delta \v{u}^\star \|,
\end{align}
where we first converted $\Delta \v{u}$ to the lattice units, then used Chapman-Enskog relation from Eq.~\eqref{eq:chapman}, and finally used the choice of parameter scaling with Reynolds number $\re$ from Eqs.~\eqref{eq:parameter_choice1}-\eqref{eq:parameter_choice2}. Thus, one can easily calculate $\epsilon_C$ using the numerical data on $\Delta \v{u}^\star$ as follows: 
\begin{align}
    \epsilon_C &= \frac{\re^{\frac{\beta D}{2}}}{N} \max_{t^\star\in\{1,\dots,T^\star\}} \sum_{\v{r}^\star} \|\Delta \v{u}^\star\| \nonumber\\
    &= \re^{\frac{-\beta D}{2}} \max_{t^\star\in\{1,\dots,T^\star\}} \sum_{\v{r}^\star} \|\Delta \v{u}^\star\|.
\end{align}

\begin{figure*}[t!]
    \centering
    \includegraphics[width=0.35\linewidth]{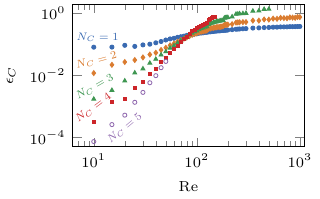}\hspace{2cm}
    \includegraphics[width=0.35\linewidth]{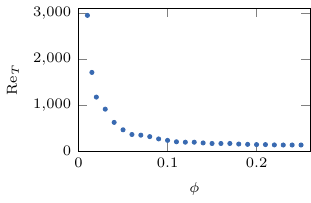}
    \caption{\textbf{Carleman truncation error and convergence for $D=1$.} Left: Truncation error $\epsilon_C$ dependence on the Reynolds number $\re$ and truncation order $N_C$ for D1Q3 model with $\beta=3/4$ (corresponding to spatial discretization resolving the Kolmogorov microscale), $u_0^\star=1$ (enough to keep the norm of the solution below 1 at all times), and the initial sinusoidal state. Right: The threshold value $\re_T$ above which the Carleman truncation error $\epsilon_C$ is larger for truncation order $N_C=2$ than for $N_C=1$. Data for D1Q3 model with $\beta=3/4$, $u_0^\star=1$, and initial states given by colliding states with fraction parameter $\phi$. }
    \label{fig:error_D1}
\end{figure*}

Due to immense memory requirements for classically simulating Carleman evolution of three-dimensional LBE systems,\footnote{For a system with $\re=100$, assuming spatial resolution parameter $\beta=3/4$ to resolve the Kolmogorov microscale, the Carleman vector for the smallest nontrivial truncation order $N_C=2$ has dimension approximately equal to $7.8\times 10^{11}$. Assuming each entry is stored as a double precision real number taking 8 bytes, it means that the Carleman vector takes up roughly 6.2 terabytes.} we restrict our investigations to one- and two-dimensional settings. We consider the spatial discretization parameter $\beta=3/4$, which is the value allowing one to resolve the Kolmogorov microscale. The simulation parameters $N_x$, $T^\star$, $\taus$, and $u_{\mathrm{ini}}^\star$ are chosen according to Eqs.~\eqref{eq:parameter_choice1}-\eqref{eq:parameter_choice2}, so that
\begin{aligns}
    \label{eq:parameter_choice_beta_34_1}
    N_x = \left\lceil \re^{3/4}\right\rceil ,\qquad T^\star = \left\lceil\re^{\frac{3}{4}(D/2+1)}\right\rceil,\\
    \bar{\tau}^\star = \frac{1}{2} + \frac{3}{\re^{\frac{3}{4}(D/2-1)+1}},\quad u_{\mathrm{ini}}^\star = \frac{1}{\re^{3 D/8}}.
    \label{eq:parameter_choice_beta_34_2}
\end{aligns}
We perform simulations for a range of Reynolds numbers $\re\in[10,1000]$ and, for the convenience of the reader, we list the corresponding numerical values of simulation parameters for selected values of $\re$ in Appendix~\ref{app:parameters}.

Note that in the original Eqs.~\eqref{eq:parameter_choice1}-\eqref{eq:parameter_choice2} it was $u_\mx^\star$, and not $u_{\mathrm{ini}}^\star$ scaling as above. Let us remind that the parameter choice leading to such a scaling was motivated by keeping $\|\v{g}(t^\star)\|\leq 1$ at all times. In all of our numerical simulations, we confirmed that choosing $u_{\mathrm{ini}}^\star$ (that we have control over) to scale as above instead of $u_\mx^\star$ (that we have no control over) guarantees that $\|\v{g}(t^\star)\|$ is always below~1. To be more precise, it is the norm $\|\v{g}(t^\star)\|$ that stays below 1, not $\|\v{y}_1(t^\star)\|$. For most simulations $\|\v{y}_1(t^\star)\|$ also stays below 1, but this stops being the case when $\epsilon_C$ becomes large, as is the case e.g., for $D=1$ system with $N_C>1$ and Reynolds number $\re$ above the threshold value $\re_T$, see the left panel of Fig.~\ref{fig:error_D1}. However, in such cases the Carleman output anyway is very erroneous, so even if one could overcome the data extraction bottleneck created by $\|\v{y}_1(t^\star)\|>1$, the extracted data would not approximate the real fluid system being simulated. In other words, as long as the Carleman embedding approximates the evolution well enough, the norm condition is satisfied, and the quantum algorithm can work properly.


\subsubsection{Convergence for \texorpdfstring{$D=1$}{D=1}}
\label{sec:error_d1}

In the one-dimensional case, the compressible Navier–Stokes equations admit non-trivial solutions in the form of pressure waves, since density variations are allowed. By contrast, the incompressible Navier–Stokes equations admit only a trivial, spatially uniform solution, since the velocity field must be divergence-free. Acknowledging that incompressible LBE, to leading order, can only recover a spatially constant state in one-dimension, we nevertheless assess Carleman error convergence for this case, as larger $N_C$ are numerically accessible allowing us to present the main ideas.

We consider two families of initial states, the \emph{sinusoidal states} defined by
\begin{equation}
    \v{g}(\v{r}^\star,0) = \v{g}^{\mathrm{eq}}\left(u^\star = u^\star_{\mathrm{ini}} \sin\frac{2\pi r_x^\star}{N_x} \right),
\end{equation}
and the \emph{colliding states}, parametrized by the fraction parameter $\phi$ as
\begin{equation}
    \v{g}(\v{r}^\star,0) =  \left\{
    \begin{array}{llr}
         \v{g}^{\mathrm{eq}}(u^\star = u^\star_{\mathrm{ini}}) \qquad & \mathrm{for} \qquad & r_x^\star\in\mathcal{R}_L,\\
         \v{g}^{\mathrm{eq}}(u^\star = 0) \qquad & \mathrm{for} \qquad & r_x^\star\in\mathcal{R}_C,\\
         \v{g}^{\mathrm{eq}}(u^\star = -u^\star_{\mathrm{ini}}) \qquad &  \mathrm{for}\qquad &  r_x^\star\in\mathcal{R}_R,
    \end{array}
    \right.
\end{equation}
where
\begin{aligns}
    \mathcal{R}_L&=\{1,\dots,\lceil \phi N_x \rceil\},\\
    \mathcal{R}_C&=\{\lceil \phi N_x \rceil+1,\dots,N_x - \lceil \phi N_x \rceil\},\\
    \mathcal{R}_R&=\{N_x - \lceil \phi N_x \rceil +1,\dots,N_x\}.
\end{aligns}
In the above, $\v{g}^{\mathrm{eq}}(u^\star)$ is the equilibrium distribution corresponding to zero density fluctuations, $\delta\rho=0$, and lattice velocity $u^\star$ (see Eq.~\eqref{eq:eq_shift}), while $u^\star_{\mathrm{ini}}$ depends on the Reynolds number $\re$ according to the parameter choice from Eqs.~\eqref{eq:parameter_choice_beta_34_1}-\eqref{eq:parameter_choice_beta_34_2}.

Our numerical results on the behavior of the Carleman truncation error $\epsilon_C$ as a function of the Reynolds number $\re$ and the Carleman truncation order $N_C$ are presented in
Fig.~\ref{fig:error_D1}. In the left panel, for the sinusoidal states, one can clearly observe the existence of a threshold Reynolds number $\re_T$, such that $\epsilon_C$ converges (i.e., decreases as $N_C$ increases) for $\re\leq \re_T$, and diverges otherwise. The existence of the threshold can be understood, as $\re$ quantifies the relative strength of nonlinear (inertial) effects compared to linear (viscous) effects in a fluid flow, and it is known that the Carleman procedure converges only when nonlinearities are weak enough. One can then expect that, as $\re$ grows, the strength of the nonlinearity grows as well, eventually bringing the system outside of the Carleman convergence radius. For the colliding state, the same threshold behavior is observed, just with the position of the threshold depending on the fraction parameter~$\phi$. As we illustrate in the right panel of Fig.~\ref{fig:error_D1}, the threshold becomes higher for smaller $\phi$, i.e., when there are fewer nodes with non-zero velocity, which directly translates to a smaller average magnitude of velocity. The behavior in $\phi$ is not surprising, as lower $\phi$ means lower average velocity, and so weaker nonlinearities for the same~$\re$.

\begin{figure*}[t!]
    \centering
    \includegraphics[width=0.35\linewidth]{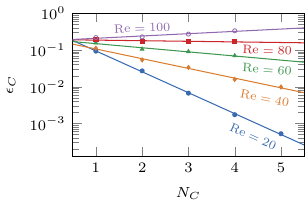}\hspace{2cm}
    \includegraphics[width=0.35\linewidth]{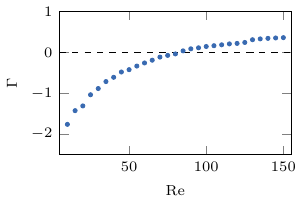}
    \caption{\textbf{Carleman error convergence and divergence for $D=1$.} Left: Dependence of the Carleman error $\epsilon_C$ (points) on the truncation order $N_C$ for selected values of the Reynolds number $\re$ below and above the threshold $\re_T$, together with the best exponential fits (solid lines). Data obtained for the $D=1$ system with the initial sinusoidal state. Right: Values of the corresponding exponential damping parameters $\Gamma$ (see Eq.~\eqref{eq:error_model}) as a function of the Reynolds number $\re$.}
    \label{fig:damping}
\end{figure*}

Note that the threshold behavior described above seems to be contrary to what was suggested in Ref.~\cite{li2025potential}, where the authors claimed that solving the lattice Boltzmann form of the Navier-Stokes equations, one reduces the nonlinearity from $\re$-determined to $\ma$-determined, where $\ma$ is the Mach number characterizing the fluid flow. Our numerical investigations suggest that the picture is more intricate, with both $\re$ and $\ma$ playing their role in convergence, especially when the lattice velocity is chosen to scale with $\re$, as we do in our simulations. Moreover, one does not need to rely only on numerics to understand why increasing $\re$, even with small $\ma$, must break Carleman convergence at some point. One of the explanations may be that increasing $\re$ leads to the appearance of more fixed points and the emergence of more complex dynamics. Since it is known that Carleman embedding method is just an expansion around one fixed equilibrium point (in our case: equilibrium state with zero velocity), as soon as the dynamics becomes more complex (e.g., with more fixed points or even chaotic), one expects Carleman embedding approximation to break\footnote{Note, however, that it is not a priori clear what implications this would have for the attainment of a statistical steady state of fluid turbulence~\cite{cappelli2026lowest}.}~\cite{novikau2025globalizing}. In that respect, the observed loss of convergence at higher Reynolds numbers is not specific to the (shifted) LBE formulation, but rather reflects a well-known limitation of Carleman linearization when applied to increasingly nonlinear dynamics. Another reason for the loss of convergence may be that with increasing $\re$, one needs to increase the number of lattice points $N$ to resolve the dynamics accurately. Increasing $N$, in turn, leads to increasing the parameter $R$ from the so-called $R$-number conditions requiring $R<1$ to guarantee convergence~\cite{forets2017explicit,liu2021efficient,jennings2025quantum}. The reason for that is twofold. First, $R$ depends on the norm of the solution, which grows with growing $N$. Importantly, the $R$-number is independent of rescaling (here: choice of lattice units), so the trick with rescaling the norm that helps us with data extraction, does not reduce $R$ and does not improve convergence. Second, the $R$-number for conservative systems grows inversely with the real spectral gap~\cite{jennings2025quantum} of the linear part of the continuous dynamics, and this gap shrinks with growing $N$. We conclude that the applicability of the Carleman approach to the simulation of highly turbulent flows with high Reynolds number may be limited due to the Carleman error not converging for these cases.

A more detailed analysis of the behavior of Carleman truncation error $\epsilon_C$ with $N_C$ and $\re$ reveals that its convergence (when $\re\leq \re_T$) and divergence (when $\re>\re_T)$ with $N_C$ is, to a good approximation, exponential. As we present in the left panel of Fig.~\ref{fig:damping} for the sinusoidal states, we have $\epsilon_C\propto \exp(\Gamma N_C)$, where $\Gamma$ is negative for $\re\leq \re_T$ and positive otherwise. In the right panel of Fig.~\ref{fig:damping}, we show the dependence on the convergence parameter $\Gamma$ on the Reynolds number $\re$, where we can again clearly see the position of the threshold $\re_T$ when $\Gamma$ changes sign. We can then model the dependence of the Carleman truncation error $\epsilon_C$ by
\begin{equation}
    \label{eq:error_model}
    \epsilon_C = E \exp(\Gamma N_C),    
\end{equation}
where $E$ can be assumed approximately constant, and $\Gamma$ becomes negative for low enough $\re$. As we will see in Sec.~\ref{sec:query}, this will allow us to estimate the dependence of the query complexity of our quantum algorithm on the desired level of error. 


\subsubsection{Convergence for \texorpdfstring{$D=2$}{D=2}}
\label{sec:error_d2}

\begin{figure*}[t!]
    \centering
    \includegraphics[width=0.45\linewidth]{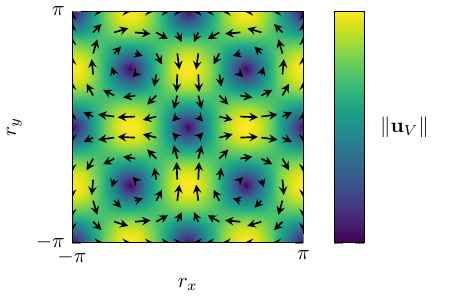}
    \hspace{1cm}
    \includegraphics[width=0.45\linewidth]{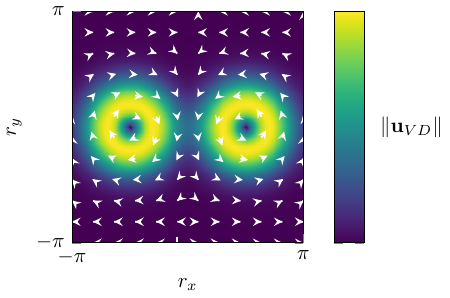}
    \caption{\textbf{Velocity field for the investigated initial states.} Velocity field $\v{u}^\star(\v{r}^\star,0)$ (arrows) and its magnitude $\|\v{u}^\star(\v{r}^\star,0)\|$ (color map) for the Taylor-Green vortex (left) and Gaussian vortex dipole (right) initial states for $D=2$.}
    \label{fig:initial}
\end{figure*}

In contrast to the one-dimensional case, the incompressible Navier-Stokes equations in two dimensions admit spatially non-trivial solutions. To assess Carleman convergence error, we consider the time evolution of the \emph{Taylor-Green vortex} and the \emph{Gaussian vortex dipole} (see Fig.~\ref{fig:initial}):
\begin{aligns}
    \v{u}_{V}(\v{r}) &= [\sin r_x \cos r_y, -\cos r_x \sin r_y],\\
    \v{u}_{VD}(\v{r}) &= \left[\frac{\partial\psi}{\partial r_y},-\frac{\partial\psi}{\partial r_x}\right],
\end{aligns}
with
\begin{equation}
    \psi = e^{-\frac{(r_x-s)^2+r_y^2}{2\sigma^2}}-e^{-\frac{(r_x+s)^2+r_y^2}{2\sigma^2}},
\end{equation}
where $r_x,r_y\in[-\pi,\pi)$ are the spatial variables, $2s$ is the dipole separation, and $\sigma$ is the width of the vortex. As there is no forcing to balance viscous dissipation, both solutions decay exponentially in time. The time evolution of the Taylor-Green vortex has a known analytical solution and is therefore used to validate numerical methods~\cite{kruger2016lattice}. While an analytical solution for the time evolution of the Gaussian vortex dipole does not exist, the vortices' counter-rotation causes them to self-advect~\cite{delbende2009dynamics}. Therefore, although the vortex dipole weakens in intensity and ultimately decays, its time evolution is more nonlinear as compared with the Taylor-Green vortex.

We investigate two families of initial states based on the velocity fields defined above. More precisely, we consider
\begin{equation}
    \v{g}(\v{r}^\star,0) = \v{g}^{\mathrm{eq}}\left(\v{u}^\star = u_{\mathrm{ini}}^\star\v{u}_X (\v{r})\right),
\end{equation}
with
\begin{equation}
    \v{r}=[-\pi+2\pi r_x^\star/N_x, -\pi + 2\pi r_y^\star/N_y],
\end{equation}
where $\v{u}_X$ is either $\v{u}_{V}$ or $\v{u}_{VD}$, $\v{g}^{\mathrm{eq}}(\v{u}^\star)$ is the equilibrium distribution corresponding to zero density fluctuations, $\delta\rho=0$, and lattice velocity $\v{u}^\star$ (see Eq.~\eqref{eq:eq_shift}), while $u^\star_{\mathrm{ini}}$ depends on the Reynolds number $\re$ according to Eqs.~\eqref{eq:parameter_choice_beta_34_1}-\eqref{eq:parameter_choice_beta_34_2}.

\begin{figure*}[t!]
    \centering
    \includegraphics[width=0.35\linewidth]{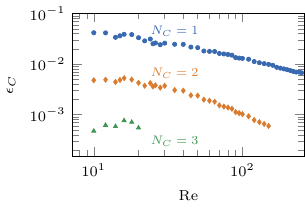}\hspace{2cm}
    \includegraphics[width=0.35\linewidth]{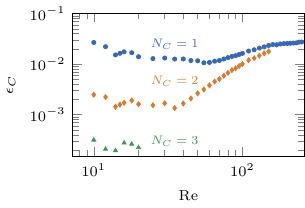}
    \caption{\textbf{Carleman truncation error $\epsilon_C$ for $D=2$.} Data for D2Q9 model with $\beta=3/4$ (corresponding to spatial discretization resolving the Kolmogorov microscale) and $u_0^\star=1$ (enough to keep the norm of the solution below 1 at all times). Initial states given by the Taylor-Green vortex (left) and Gaussian vortex dipole with $s=\pi/2$ and $\sigma=\pi/5$ (right). }
    \label{fig:error_D2}
\end{figure*}

Our numerical results on the dependence of the Carleman truncation error $\epsilon_C$ on the Reynolds number $\re$ and the Carleman truncation order $N_C$ are presented in Fig.~\ref{fig:error_D2}. In contrast to the one-dimensional case, we do not observe the threshold behavior within the numerically accessible regime of Reynolds numbers. For the Taylor-Green vortex, the larger $\re$ gets, the better the improvement in the truncation error $\epsilon_C$ when going from $N_C=1$ to $N_C=2$, hence not only do we not observe the threshold, but we cannot even extrapolate the data to estimate its position. On the other hand, for the Gaussian vortex dipole, the gap between $\epsilon_C$ for $N_C=1$ and $N_C=2$ visibly shrinks as $\re$ gets higher, suggesting that $\re_T$ is around a few hundreds. To confirm that, we extended the range of accessible Reynolds number to $\re=500$ by reducing the resolution parameter to $\beta=1/2$. In such a setting, we observe the threshold behavior for the Gaussian vortex dipole at around $\re_T\approx 400$, while the error behavior for the Taylor-Green vortex stays unchanged (i.e., the error gap widens with growing $\re$).

Let us now briefly discuss the potential explanation for the observed behavior of $\epsilon_C$. First, the much higher threshold for $D=2$ as compared to $D=1$ may be coming from a stronger velocity renormalization, since $u^\star\sim 1/\re^{\beta D/2}$. Smaller velocity could translate into weaker nonlinearity, so that the same $\re$ (physically describing the same strength of nonlinearity) that lies outside the Carleman convergence radius for $D=1$ could lie inside it for $D=2$. This would suggest that for $D=3$ the threshold could be even higher, as $u^\star$ is even smaller for the same $\re$. Also, the fact that initial states for $D=1$ did not satisfy the divergence-free condition could play a role. This is because in that case numerical instabilities might be driving the threshold to lower values. Finally, the difference in the behavior of $\epsilon_C$ between the Taylor-Green vortex and the Gaussian vortex dipole may be coming from their fundamentally different dynamics, with the first one simply decaying, and the second one exhibiting non-trivial advection as well. This would suggest that the Carleman linearization method performance may strongly depend on the type of investigated fluid dynamics problem, not only on its Reynolds number.


\subsection{Block-encoding prefactor analysis}
\label{sec:BE}

We now proceed to analyzing the block-encoding prefactor $\alpha_A / \|A\|$ appearing in the upper bound for the query complexity in Eq.~\eqref{eq:queryBound}. We have already derived the expression for $\alpha_A$ as a function of the relaxation parameter $\taus$ and the Carleman truncation order $N_C$, see Eq.~\eqref{eq:prefactor_final_main}. To estimate $\| A\|$, we note that the norm of $A_H$ (recall Eq.~\eqref{eq:AH}) can be bounded as follows:
\begin{equation}
    \| A_H\| \leq \|I\| + \| A_H-I\| = 1 + \| A_H-I\|.    
\end{equation}
Then, using the definition of operator norm and unitarity of $\S$, we have 
\begin{equation}
    \| A_H\| \leq 1 + \|\C\|.
\end{equation}
On the other hand, taking a vector $\v{v}=[0,\v{x},0,\dots]^\top$ with $\v{x}$ such that $\|\C\v{x}\|=\|\C\|\|\v{x}\|$, we get
\begin{equation}
    \| A_H\| \geq \frac{\|A_H\v{v}\|}{\|\v{v}\|}= \sqrt{1 +\|\C\|^2}.
\end{equation}
Exactly the same argument applies to the linear system $A_F$ for the final Carleman state, and so we conclude that
\begin{equation}
    \label{eq:A_norm_bound}
    \sqrt{1 +\|\C\|^2}\leq \| A \| \leq 1 + \|\C\|.
\end{equation}
Note that the above bound is pretty tight, in the sense that the difference between the upper and lower bounds is at most 1. 

Now, combining the above with Eq.~\eqref{eq:prefactor_final_main}, we get
\begin{equation}
    \label{eq:BE_upper}
    \frac{\alpha_A}{\| A \|} \leq \frac{1+ \sum\limits_{l=0}^{\lfloor\frac{N_C}{2}\rfloor} \binom{N_C-l}{l} \alpha_{ I+F_1}^{N_C-2l} \alpha_{F_2}^{l}}{\sqrt{1+\|\C\|^2}} =: \mathrm{BE}_A^<,
\end{equation}
where the approximation is justified as the relevant terms above grow quickly with $N_C$ and become much larger than 1 in the interesting regime.

\begin{figure*}[t!]
    \centering
    \includegraphics[width=0.25\linewidth]{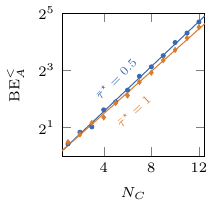}
    \hspace{1cm}
    \includegraphics[width=0.25\linewidth]{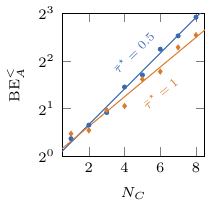}
    \hspace{1cm}
    \includegraphics[width=0.25\linewidth]{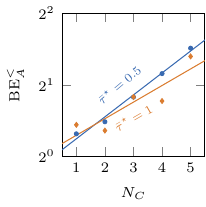}
    \caption{\textbf{Upper bounds for block-encoding prefactors.} The numerical values of $BE^<_A$ from Eq.~\eqref{eq:BE_upper} for different values of the Carleman truncation order $N_C$, relaxation parameter $\taus$, and dimension $D$ (left to right: $D=1$ to $D=3$), together with their best exponential fits as in Eq.~\eqref{eq:BE_fit}. Note that the fits showcase slightly different slopes for the even and odd $N_C$, but we took an overall fit.}
    \label{fig:BE}
\end{figure*}

Now, the only missing component to get the upper bound $\mathrm{BE}_A^<$ for the block-encoding prefactor $\alpha_A / \|A\|$ is the norm of the Carleman collision matrix $\|\C\|$. To estimate it, note that due to the locality of collision matrices $F_k$ (recall Eqs.~\eqref{eq:F1}-\eqref{eq:F2}), the norm of $\C$ for a system with $N$ lattice sites is the same as the norm of a matrix $\C$ for a system with a single lattice site. As such, $\|\C\|$ depends only on the relaxation parameter $\taus$ and the Carleman truncation order $N_C$. Moreover, instead of dealing with the norm of $d_C$-dimensional matrix, one only needs to consider a matrix of dimension
\begin{equation}
    \tilde{d}_C :=\frac{Q(Q^{N_C}-1)}{Q-1},
\end{equation}
where we recall that $Q=3^D$ is the number of discrete velocities in the considered LBE model. This way, we can efficiently calculate $\|\C\|$ for small values of $N_C$ and a range of $\taus\in[1/2,1]$. Using these numerical values of $\|\C\|$, together with the expression for $\alpha_A$, one can calculate upper bounds $\mathrm{BE}^<_A$. We present these in Fig.~\ref{fig:BE} for the extreme cases of $\taus=1/2$ and $\taus=1$, where one can see that $\mathrm{BE}^<_A$ grows exponentially with $N_C$, so that we get the following approximate upper bound:
\begin{equation}
    \frac{\alpha_A}{\| A \|} \lesssim b \exp(a N_C).
\end{equation}
Note that the exponential scaling with growing $N_C$ could also be expected from the analytical form of the Carleman collision matrix $\C$. In particular, the Carleman blocks $C^k_l$ in Eq.~\eqref{eq:Ckl} contain progressively more factors of $I+F_1$ and $F_2$ as $N_C$ increases. Consequently, any non-optimality of block-encoding (i.e., when the block-encoding prefactor is larger than the norm) is amplified exponentially with $N_C$. Note that the complexity cost related to this can be tolerated as long as the Carleman truncation error $\epsilon_C$ gets exponentially suppressed with increasing $N_C$. As discussed in Sec.~\ref{sec:error}, this is the case below the threshold Reynolds number, and this cost then scales inverse polynomially with $\epsilon_C$. Moreover, depending on the parity of $N_C$, the block $C^k_l$ that brings the dominant contribution to the norm of $\C$ has an oscillating ratio of $I+F_1$ and $F_2$ terms. This parity-dependent behavior is indeed visible in Fig.~\ref{fig:BE}, where the general trend is exponential, but distinct for even and odd $N_C$.

To get the numerical values of coefficients, we first note that for all $D$ the value of $\mathrm{BE}^<_A$ is largest for $\taus=1/2$. Thus, for each dimension, we take the worst case scenario and perform best linear fits to the following equation:
\begin{equation}
    \label{eq:BE_fit}
    \ln \mathrm{BE}^<_A = a N_C + \ln b.
\end{equation}
This way, we obtain the following:
\begin{aligns}
    a &= \left\{\begin{array}{cc}
         0.273&\mathrm{~for}\quad D=1,  \\
         0.260&\mathrm{~for}\quad D=2,  \\
         0.213&\mathrm{~for}\quad D=3,
    \end{array}
    \right.\\
    b &= \left\{\begin{array}{cc}
         0.993&\mathrm{~for}\quad D=1,  \\
         0.942&\mathrm{~for}\quad D=2,  \\
         0.957&\mathrm{~for}\quad D=3,
    \end{array}
    \right.
\end{aligns}
so that approximating $a\approx1/4$ and $b\approx1$ for all $D$, we get the following simple approximate bound:
\begin{equation}
    \label{eq:alpha_bound}
    \frac{\alpha_A}{\| A \|} \lesssim \exp\left(\frac{N_C}{4}\right).
\end{equation}


\subsection{Condition number analysis}
\label{sec:condition}


\subsubsection{Lower bound}
\label{sec:condition_bound}

To lower bound the condition number $\kappa_A$ of $A$, we need bounds on both $\| A\|$ and $\| A^{-1}\|$. In the previous section we have already derived a tight bound for $\| A\|$ in Eq.~\eqref{eq:A_norm_bound}, so we just need to bound $\| A^{-1} \|$. We start by noting that, since $A_F^{-1}$ has the same structure as $A_H^{-1}$ (just appended with extra blocks), we have 
\begin{equation}
    \|A^{-1}\| \geq \|A_H^{-1}\|.
\end{equation}
Now, let us introduce a normalized vector
\begin{equation}
    \ket{x_\theta} = \frac{1}{\sqrt{T^\star+1}}\sum_{t^\star=0}^{T^\star}  e^{i\theta t^\star} \ket{t^\star}\ot \ket{\xi_\theta},
\end{equation}
where $\ket{\xi_\theta}$ is the eigenvector of $\S\C$ with eigenvalue $e^{i\theta}$ for some real $\theta$ (we will justify its existence below). 
We then have
\begin{align}
    A^{-1}_H\ket{x_\theta} &= \sum_{r=0}^{T^\star} \sum_{c=0}^{r} \sum_{t^\star=0}^{T^\star}  \frac{e^{i\theta t^\star}}{\sqrt{T^\star+1}}   \ketbra{r}{c}t^\star\rangle \ot (\S\C)^{r-c} \ket{\xi_\theta}\nonumber\\
    & = \frac{1}{\sqrt{T^\star+1}} \sum_{r=0}^{T^\star} \sum_{c=0}^{r} e^{i\theta r}   \ket{r}\ot  \ket{\xi_\theta}\nonumber\\
    & = \frac{1}{\sqrt{T^\star+1}} \sum_{r=0}^{T^\star} e^{i\theta r} (r+1)  \ket{r}\ot  \ket{\xi_\theta},
\end{align}
which can be used to obtain the following lower bound:
\begin{align}
    \|A^{-1}\|&\geq \|A_H^{-1}\| \geq \| A_H^{-1} \ket{x_\theta}\| = \frac{\sqrt{\sum_{r=0}^{T^\star} (r+1)^2}}{\sqrt{T^\star+1}}\nonumber\\
    &= \sqrt{\frac{(2+T^\star)(3+2T^\star)}{6}} \geq \frac{T^\star}{\sqrt{3}}.
\end{align}
Putting the above together with Eq.~\eqref{eq:A_norm_bound}, we end up with
\begin{equation}
    \label{eq:kappa_lower}
    \kappa_A \geq \frac{\|\C\|T^\star}{\sqrt{3}}.
\end{equation}

We now argue for the existence of an eigenvector $\ket{\xi_\theta}$ of $\S\C$ with eigenvalue $e^{i\theta}$. Here, we will just consider the simpler case with no walls, when the streaming matrix $S$ is given by Eq.~\eqref{eq:streaming}, and we will show how to construct $\ket{\xi_0}$ (i.e., the eigenvector with eigenvalue 1). In Appendix~\ref{app:cs_eigenstate}, we consider the harder case, when there are non-trivial boundaries with the streaming operator described by Eq.~\eqref{eq:streaming_walls}, and we construct the eigenstate for $\theta=\pi$, i.e., a state $\ket{\xi_\pi}$ that is a $-1$ eigenstate of $\S\C$.

First, we set all Carleman blocks, except for the first one, to zero:
\begin{equation}
    \ket{\xi_0} = [\ket{\zeta_0},0,\dots,0].
\end{equation}
Thus,
\begin{equation}
    \S\C \ket{\xi_0} = [S(I+F_1)\ket{\zeta_0},0,\dots,0].
\end{equation}
Next, choosing $\ket{\zeta_0}$ to be translationally invariant in space,
\begin{equation}
    \ket{\zeta_0} = \frac{1}{\sqrt{N}}\sum_{\v{r}^\star} \ket{\v{r}^\star} \ot \ket{\psi_0},
\end{equation}
to guarantee that $S\ket{\zeta_0}=\ket{\zeta_0}$. Finally, recalling that $F_1$ acts locally in space (see Eq.~\eqref{eq:F1_loc}), we note that taking $\ket{\psi_0}$ to be the zero eigenstate of $\tilde{F}_1$ (its existence can be checked by direct calculation for $Q\times Q$ matrices), we get $(I+F_1)\ket{\zeta_0}=\ket{\zeta_0}$. We thus conclude that $\ket{\xi_0}$ constructed as described above is indeed an eigenstate of $\S\C$ with eigenvalue 1.


\subsubsection{Numerical evaluation}

While a lower bound on the condition number scaling worse than the classical complexity of solving an LBE problem would rule out the possibility for a quantum advantage, an upper bound scaling better than classical would prove its existence. Thus, deriving such an upper bound is central to any claim of quantum advantage. In Appendix~\ref{app_condition_upper}, we explain how one can derive the following bound for $\kappa_{A_H}$ and $\kappa_{A_F}$:
\begin{aligns}
    \label{eq:kappa_upper_H}
    \!\!\!\kappa_{A_H}  \leq &(\|\C\|+1) \Bigg(\sum_{t^\star=0}^{T^\star} (T^\star-t^\star+1) \|(\S\C)^{t^\star}\|^2\Bigg)^{\frac{1}{2}}\!\!\!,\\
    \!\!\!\kappa_{A_F}  \leq &(\|\C\|+1) \Bigg(\sum_{t^\star=0}^{T^\star} \left[ ( 2^{W}(T^\star+1)-t^\star)\|(\S\C)^{t^\star}\|^2\right] \nonumber\\
    &\qquad\qquad\quad +  2^{2W-1}(T^\star+1)^2\Bigg)^{\frac{1}{2}}\!\!\!.
    \label{eq:kappa_upper_F}
\end{aligns}
The problem with analyzing the above bounds is that one needs to estimate the values of $\|(\S\C)^{t^\star}\|$, which turns out to be a non-trivial problem. One can then resort to numerics and evaluate the necessary norms. It turns out, however, that it is not much more computationally demanding to directly calculate the condition number. 

\begin{figure*}[t!]
    \centering
    \includegraphics[width=0.35\linewidth]{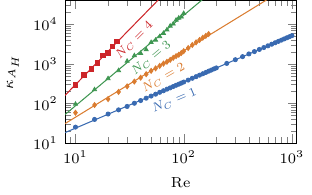}
    \hspace{2cm}
    \includegraphics[width=0.35\linewidth]{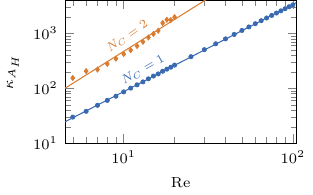}
    \caption{\textbf{Condition number.} Scaling of the condition number $\kappa(A_H)$ with the Reynolds number $\re$ for D1Q3 model (left) and D2Q9 model (right), and different Carleman truncation orders $N_C$. The considered systems have periodic boundary conditions with no walls and $\beta=3/4$.  }
    \label{fig:condition}
\end{figure*}

Thus, instead of bounding it, we evaluate the value of $\kappa_{A_H}$ for $D=1$ and $D=2$ systems with Reynolds number $\re\in[10,1000]$, and extrapolate to get the scaling behavior of $\kappa_{A_H}$. Note that $\kappa_{A_H}$ is a lower bound for $\kappa_{A_F}$, so the query complexity upper bound for the final state system will be at least the one we find for the history state case. At this stage, we will use $\kappa_{A_H}$ as an estimate for $\kappa_{A_F}$, but of course this could be improved upon. More precisely, we first approximate $\|A_H\|$ using the upper bound from Eq.~\eqref{eq:A_norm_bound} (recall that this approximation differs by at most 1 from the real value of $\|A_H\|$). Then, we numerically evaluate $\|A_H^{-1}\|$, by using the Lanczos method to find the largest eigenvalue of $(A_H^{-1})^\dagger A_H^{-1}$, and taking its square root. As before, we consider the spatial discretization parameter $\beta=3/4$ (resolving the Kolmogorov microscale), so that all the simulation parameters entering $A_H^{-1}$ (i.e., $N_x$, $T^\star$, and $\taus$) become simple functions of the Reynolds number $\re$ according to Eqs.~\eqref{eq:parameter_choice_beta_34_1}-\eqref{eq:parameter_choice_beta_34_2}.

We present the dependence of $\kappa_{A_H}$ on the Reynolds number using doubly logarithmic plots in Fig.~\ref{fig:condition}. Using best linear fits to
\begin{equation}
    \ln \kappa_{A_H} = \chi \ln\re + \ln c,
\end{equation}
one can approximate the condition number well by
\begin{equation}
    \label{eq:condition_approx}
    \kappa_{A_H} \approx c\ \re^\chi,
\end{equation}
where parameters $\chi$ and $c$ depend both on problem dimension $D$ and the Carleman truncation order $N_C$:
\begin{aligns}
    \label{eq:chi1}
    \chi = \left\{\begin{array}{cc}
         1.167&\mathrm{~for}\quad D=1,~N_C=1,  \\
         1.691&\mathrm{~for}\quad D=1,~N_C=2,  \\
         2.283&\mathrm{~for}\quad D=1,~N_C=3, \\
         2.792&\mathrm{~for}\quad D=1,~N_C=4, \\ 
         1.588&\mathrm{~for}\quad D=2,~N_C=1,  \\
         1.936&\mathrm{~for}\quad D=2,~N_C=2,  
    \end{array}
    \right.\\
    \label{eq:chi2}
    c = \left\{\begin{array}{cc}
         1.635&\mathrm{~for}\quad D=1,~N_C=1,  \\
         0.905&\mathrm{~for}\quad D=1,~N_C=2,  \\
         0.459&\mathrm{~for}\quad D=1,~N_C=3, \\
         0.486&\mathrm{~for}\quad D=1,~N_C=4, \\ 
         2.254&\mathrm{~for}\quad D=2,~N_C=1,  \\
         5.460&\mathrm{~for}\quad D=2,~N_C=2.
    \end{array}
    \right.
\end{aligns}
Note that, while the increase of $\chi$ with $D$ could be expected and might be handled (after all, classical complexity also increases with $D$), the strong dependence on $N_C$ is troublesome. That is because obtaining low error output may require a high value of $N_C$, which would then directly translate into complexity cost via the dependence of the condition number on $\re^\chi$. Thus, while we postpone a detailed discussion on a possible quantum advantage to Sec.~\ref{sec:conclusions}, we note here that it may be restricted to high error outputs.


\subsection{Estimating query complexity}
\label{sec:query}

We now have all the necessary ingredients to bound the query complexity of running a quantum linear solver for our shifted lattice Boltzmann problem. First, let us briefly comment on the lower bound for the query complexity $q_Q$. It is known that $q_Q$ can at best scale linearly with the condition number $\kappa_A$ of the investigated linear system~$A$~\cite{costa2022optimal}. As we analytically proved a lower bound on $\kappa_A$ that scales linearly with the number of time steps (recall Eq.~\eqref{eq:kappa_lower}), we then have 
\begin{equation}
    q_Q = O(T^\star).
\end{equation}
Recalling the simulation parameter choice from Eqs.~\eqref{eq:parameter_choice1}-\eqref{eq:parameter_choice2}, we conclude that at best (i.e., ignoring all other potential dependence on the Reynolds number not arising from the dependence on $T^\star$) we have:
\begin{equation}
    \label{eq:query_lower}
    q_Q = O\left(\re^{\beta(D/2+1)}\right).
\end{equation}
Thus, even ignoring the block-encoding cost, at best the complexity will scale with the Reynolds number as $\re^{\beta(D/2+1)}$.

Now, let us proceed to the upper bound. We start by simplifying the bound on $q_Q$ from Eq.~\eqref{eq:queryBound}, by slightly weakening it through the observation that for all $\epsilon_Q\geq 10^{-10}$ and all $\kappa\geq 1$, one has
\begin{equation}
    1.05 \ln\left(\frac{\sqrt{1-\epsilon_Q^2}}{\epsilon_Q}\right)+\frac{2.78\ln(\kappa)^3+3.17}{\kappa} \leq 29.0.
\end{equation}
Since the Carleman truncation errors are expected to be much larger than $10^{-10}$, this assumption is justified, and so we get
\begin{equation}
        \label{eq:queryBoundApprox}
        q_Q \leq  \frac{85.0 \ \alpha_A \kappa_A}{\|A\|}.
\end{equation}
Next, we upper bound the block-encoding prefactor using Eq.~\eqref{eq:alpha_bound}, to obtain
\begin{equation}
    q_Q \lesssim  85.0 \exp\left( \frac{N_C}{4}\right) \kappa_A.
\end{equation}
Then, focusing as before on $\beta=3/4$, we approximate the condition number using Eq.~\eqref{eq:condition_approx} to finally arrive at
\begin{equation}
    \label{eq:query_upper}
    q_Q \lesssim  85.0 c \exp\left( \frac{N_C}{4}\right) \re^\chi,
\end{equation}
where, as expected, the main factor comes from $\re^\chi$, with $\chi$ depending on $N_C$.

We also note that by rewriting
\begin{equation}
    \exp\left( \frac{N_C}{4}\right) = [\exp(\Gamma N_C)]^{\frac{1}{4\Gamma}}
\end{equation}
and using the error model from Eq.~\eqref{eq:error_model} under the assumption that we are operating below the threshold Reynolds number ($\re\leq \re_T$), so that $\Gamma<0$, we get
\begin{equation}
    \label{}
    \exp\left(\frac{N_C}{4}\right) = \left(\frac{E}{\epsilon_C}\right)^{\frac{1}{4|\Gamma|}}.
\end{equation}
Plugging this back into the query complexity upper bound, we get
\begin{equation}
    q_Q \lesssim  85.0 c \ \left(\frac{E}{\epsilon_C}\right)^{\frac{1}{4|\Gamma|}} \re^\chi,    
\end{equation}
where we note that $\chi$ depends on $N_C$, which in turn should be chosen as
\begin{equation}
    N_C = \left\lceil \frac{1}{|\Gamma|}\ln \frac{E}{\epsilon_C} \right\rceil.
\end{equation}


\subsection{Estimating gate complexity}
\label{sec:gates}

\begin{figure*}[t!]
    \centering
    \includegraphics[width=0.35\linewidth]{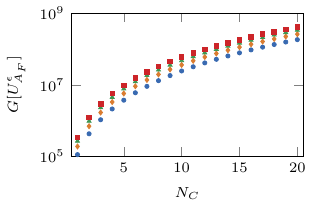}
    \hspace{2cm}
    \includegraphics[width=0.35\linewidth]{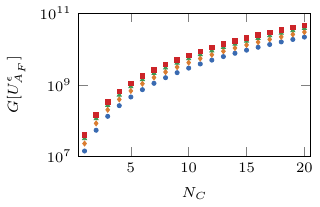}
    \caption{\textbf{$T$-gate cost of $U^\epsilon_{A_F}$.} The estimates for the number $G[U^\epsilon_{A_F}]$ of $T$ gates needed to prepare a circuit $\epsilon$-close to a unitary $U_{A_F}$ encoding the LBE problem with dimension $D=1$ (left) and $D=2$ (right) as function of Carleman truncation order $N_C$. Different symbols, bottom to top, correspond to the allowed error level $\epsilon=10^{-3k}$ for $k\in\{1,2,3,4\}$. Note that the dependence of $G[U^\epsilon_{A_F}]$ on the Reynolds number $\re$ is only logarithmic, making the presented data applicable with a very good approximation to all $\re\geq 1$.}
    \label{fig:gates}
\end{figure*}
In this section, we present our estimates for the cost of a single query to $U_{A_F}$ (which also acts as the upper bound for the cost of $U_{A_H}$) in terms of the number of $T$ gates needed to construct a quantum circuit implementing it. All the details of our calculations, assumptions and approximations made can be found in Appendix~\ref{app:gates}. The $T$-gate cost of approximating $U_{A_F}$ up to error $\epsilon_{\mathrm{{tot}}}$ is estimated to be given by:
\begin{align}
    G[U_{A_F}^{\epsilon_{\mathrm{{tot}}}}] \approx &~ 78 W +28 D N_C +64+102\log\frac{1}{\epsilon} \nonumber\\
    &+ 24(D(1+4 N_C)+2)\log \re\nonumber\\
    &+    6N_C \left(n_B+1+\log\frac{1}{\epsilon}\right) \nonumber\\
    &+ \frac{N_C+2}{2} (14+14n_B+64n_C) \nonumber\\
    & + \frac{N_C^2+4N_C}{4} \left( 28 + 28(n_B+n_C) + 48 \log\frac{1}{\epsilon} \right) \nonumber\\
    & +\left[\left(2+\frac{4}{\sqrt{5}}\right)\left(\frac{1+\sqrt{5}}{2}\right)^{N_C}-4]\right]\nonumber\\
    &\qquad\left( 14+14(n_B+n_C) \right)+2h_{N_C}\nonumber\\
    & + \frac{2N_C(N_C+2  )}{3} (2N_C+5) G[U^{\epsilon_{F_1}}_{I+F_1}] \nonumber\\
    &+ \frac{2N_C(N_C+2  )}{3}(N_C+1)  G[U^{\epsilon_{F_2}}_{{F}_2}] \Big],
    \label{eq:gate_final_lastline}
\end{align}
where 
\begin{align}
    \epsilon_{\mathrm{{tot}}} =& 4\epsilon+ \frac{N_C}{24}\left[ 6(12+N_C)\epsilon +(N_C+2)((2N_C+5)\epsilon_{F_1} \nonumber\right. \\ 
    &\left.\qquad\qquad~+ (N_C+1)\epsilon_{F_2}) \right],
\end{align}
whereas $\epsilon_{F_1}$ and $\epsilon_{F_2}$ depend on the problem dimension in the following way:
\begin{aligns}
    \epsilon_{F_1}  &= \left\{
    \begin{array}{cc} 
         (4\sqrt{2}+2)\epsilon& 
         \mathrm{for~} D=1, \\[6pt]
         (28\sqrt{2}+5)\epsilon& 
         \mathrm{for~} D=2, 
    \end{array}
    \right.\\
    \epsilon_{F_2}  &= \left\{
    \begin{array}{cc}
         (1+2\sqrt{2})\epsilon&  \mathrm{~~for~} D=1, \\[6pt]
         59\sqrt{2}\epsilon&  \mathrm{~~for~} D=2.
    \end{array}
    \right.
\end{aligns}
while $G[U^{\epsilon_{F_1}}_{I+F_1}]$ and $G[U^{\epsilon_{F_2}}_{F_2}]$ are given by
\begin{aligns}
    G[U^{\epsilon_{F_1}}_{I+F_1}] & = \left\{
    \begin{array}{cc} 
         280 + 480 \log\frac{1}{\epsilon}& \mathrm{for~} D=1, \\[6pt]
         6412 + 60291 \log\frac{1}{\epsilon}& \mathrm{for~}D=2, 
    \end{array}
    \right. \\
    G[U^{\epsilon_{F_2}}_{{F}_2}] & = \left\{
    \begin{array}{cc}
         154+240 \log\frac{1}{\epsilon}&  \mathrm{~~for~} D=1, \\[6pt]
         4500+5664 \log \frac{1}{\epsilon}&   \mathrm{~~for~} D=2.
    \end{array}
    \right.
\end{aligns}
Moreover, $\epsilon$ is the equal error on all approximate single-qubit gates used (it should be treated as a free parameter controlling the total error $\epsilon_{\mathrm{{tot}}}$), and we defined
\begin{equation}
    h_{N_C} := \sum_{l=0}^{\lfloor\frac{N_C}{2}\rfloor} \sum_{k=l'}^{N_C-l} \binom{k}{l}\log\binom{k}{l}.
\end{equation}

Although we took care to cost all elements of the circuit, it turns out that the bulk of the cost comes from the two last terms in Eq.~\eqref{eq:gate_final_lastline}. For example, for $D=2$, $N_C=2$, $\epsilon=10^{-6}$, $W=10$ and $\re\in(0,10^{10}]$, it constitutes more than 99.9\% of the total cost, and this fraction only increases for higher Carleman truncation orders $N_C$ and lower approximation error $\epsilon$. Thus, we can use a much simpler estimate for the total $T$-gate cost of our circuit:
\begin{equation}
     \! G[U^{\epsilon}_{A_F}] \!\approx\! \left\{\!\!\! 
     \begin{array}{ll}
         \begin{array}{l}
               N_C(N_C+2  ) [ (476 N_C+1036) \\
               ~~+(800 N_C+1760)\log \frac{K_1}{\epsilon}  ]
         \end{array} &\mathrm{\!for~} D=1,\\[12pt]
          \begin{array}{l}
                 N_C(N_C+2) [\frac{8}{3}(4331 N_C+9140)\\
                ~~+ (84164 N_C+204746)\log\frac{K_2}{\epsilon}]
          \end{array}   &\mathrm{\!for~}D=2,
     \end{array}
     \right.
\end{equation}
with
\begin{aligns}
    K_1  =& \frac{1}{24} \left[ 96 + (94+44\sqrt{2})N_C+(27+42\sqrt{2})N_C^2 \right.\nonumber\\
    &\left.\qquad+ (5+10\sqrt{2})N_C^3 \right],\\
    K_2  =& \frac{1}{24} \left[ 96 + (122+398\sqrt{2})N_C +(51+429\sqrt{2})N_C^2 \right.\nonumber\\
    &\qquad\left.+ (10+115\sqrt{2})N_C^3\right].
\end{aligns}
We present the $T$-gate cost dependence on the truncation order $N_C$ and approximation level $\epsilon$ for both cases $D=1$ and $D=2$ in Fig.~\ref{fig:gates}.

It is thus clear that the $T$-gate cost $G[U^{\epsilon}_{A_F}]$ is directly almost independent of the Reynolds number $\re$ or the parameter $W$ controlling the success probability of getting the Carleman state at the final time $T^\star$. Instead, it has a cubic dependence on the Carleman truncation error $N_C$ and logarithmic dependence on the inverse of the gate approximation error $\epsilon$ (note that this also holds for the case $D=3$). Thus, if we decide to run our algorithm for a fixed $N_C$ and $\epsilon$, the difference between the query cost $q_Q$ and the total $T$-gate cost $q_QG[U^{\epsilon}_{A_F}]$ becomes a constant factor overhead. For example, for low $N_C\leq 5$ and error $\epsilon=10^{-6}$, it is roughly $10^6$ for $D=1$ and $10^8$ for $D=2$, so one can expect it would be around $10^{10}$ for $D=3$. Note, however, that these are just rough estimates that do not include any compilation optimization, and so one can expect that the actual numbers can be a few orders of magnitude better.

Finally, assuming we operate below the threshold Reynolds number ($\re\leq \re_T$) and modeling the Carleman truncation error $\epsilon_C$ via Eq.~\eqref{eq:error_model}, we can replace the $O(N_C^3)$ scaling with $O(\Gamma^{-3}\log^3 \epsilon_C)$, where $\Gamma<0$ is the convergence parameter depending on the Reynolds number. We thus see that the gate cost depends poly-logarithmically on the total solution error, as well as indirectly on the Reynolds number via the cubic dependence on the convergence parameter $\Gamma$. 


\section{Discussion and conclusions}
\label{sec:conclusions}

\begin{figure*}[t!]
    \centering
    \includegraphics[width=0.35\linewidth]{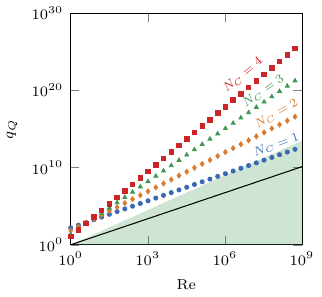}
    \hspace{2cm}
    \includegraphics[width=0.35\linewidth]{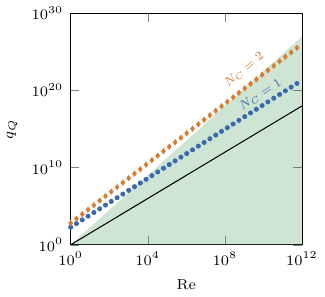}
    \caption{\textbf{Query complexity.} The lower bound from Eq.~\eqref{eq:query_lower} (black line) and upper bounds from Eq.~\eqref{eq:query_upper} for different values of Carleman truncation order $N_C$ (different symbols) of the query complexity $q_Q$ of our quantum LBE algorithm resolving the Kolmogorov microscale and for systems with dimension $D=1$ (left) and $D=2$ (right). The green highlighted region indicates potential for quantum advantage, when $q_Q$ is smaller than classical complexity. In the expressions of the form $O(\re^x)$ for $\mathrm{Complexity}_C$ and the lower bound for $q_Q$, we assumed a constant 1 and plotted $\re^x$.}
    \label{fig:query}
\end{figure*}

Having analyzed the performance of our quantum algorithm for the LBE problem, we are now in a position to compare it with classical LBE algorithms. We begin with a brief comparison between classical and quantum memory requirements. In a classical simulation, it is necessary to store the local LBE state vectors at each lattice site. Given that there are $N=N_x^D$ lattice sites, and parameterizing the number of discretization points per direction as before to be $N_x=\re^\beta$, we get that the number $n_c$ of memory bits needed scales as
\begin{equation}
    n_c = O(\re^{\beta D}).
\end{equation}
On the other hand, the total number of data qubits $n_D$ required by our algorithm, given by Eq.~\eqref{eq:data_qubits}, and ancilla qubits $n_A$, specified by Eq.~\eqref{eq:ancilla_qubits_main}, scales as 
\begin{equation}
    n_D + n_A = O(N_C\log \re).
\end{equation}
Note that the dependence on $N_C$ translates via the error model from Eq.~\eqref{eq:error_model} into a dependence on $\log\frac{1}{\epsilon_C}$. Hence, for a fixed acceptable error level, we get an exponential memory improvement compared to classical algorithms.

Next, we will measure the classical computational complexity using the number of updates of the local LBE state vector that have to be performed during the evolution. Given that there are $\re^{\beta D}$ lattice sites and $T^\star$ time steps, we have
\begin{equation}
    \mathrm{Complexity}_C = O(\re^{\beta D} T^\star).
\end{equation}
Moreover, note that $T^\star$ must also scale at least as $\re^\beta$, as otherwise the lattice velocity would diverge to infinity with growing $\re$. Thus, the classical complexity grows with the Reynolds number $\re$ as
\begin{equation}
    \mathrm{Complexity}_C = O(\re^{\beta(D+1)}),
\end{equation}
which becomes a familiar computational fluid dynamics scaling of $O(\re^3)$ for a three-dimensional system and a choice of $\beta=3/4$ guaranteeing resolution of the Kolmogorov microscale~\cite{kruger2016lattice}.            

Now, the quantum complexity of running our algorithm can be written as a product of three terms:
\begin{equation}
    \mathrm{Complexity}_Q = q_Q \times  g \times q_M,
\end{equation}
where $q_Q$ is the query complexity, $g$ is the number of $T$-gates for a single query, and $q_M$ is the information extraction overhead. However, we will compare the classical complexity with the quantum query complexity $q_Q$ (and then with $q_Q\times q_M$) and not with the total number of $T$-gates. Two main reasons for that are as follows. First, a single classical query is already a complex task, consisting of many elementary logical operations needed to update a local LBE state vector. Thus, it would be unfair to treat on equal footing such a task and the use of an elementary single-qubit $T$ gate. Second, we will be mostly interested in the asymptotic scaling of classical and quantum complexities, and as we have explained in the previous section, the difference between measuring in queries and in $T$-gates is just a constant factor overhead (i.e., if we assume a fixed $N_C)$.

Let us then compare the classical complexity of outputting the final LBE state with the query complexity $q_Q$ for outputting a coherent encoding of the final Carleman state (with a fixed truncation order $N_C$) for the shifted LBE problem. If the fraction $\mathrm{Complexity}_C/q_Q>1$, we can thus speak of a quantum advantage. First, recalling Eq.~\eqref{eq:query_lower}, we see that this fraction will at most scale as
\begin{equation}
    \frac{\mathrm{Complexity}_C}{q_Q} =O\left(\re^{\frac{\beta D}{2}}\right).
\end{equation}
Hence, for the problem of dimension $D=3$ and with the resolution parameter $\beta=3/4$ to resolve the Kolmogorov microscale, the best one can hope for is an improvement by a factor of $\re^{9/8}$, replacing a classical cubic dependence $O(\re^3)$ with the quantum slightly-lower-than-quadratic dependence $O(\re^{1.875})$. Instead, if we use our upper bound on $q_Q$ from Eq.~\eqref{eq:query_upper}, and the numerically extracted data on the power $\chi$ from Eq.~\eqref{eq:chi1} for the case $\beta=3/4$, we get that at worst
\begin{equation}
    \frac{\mathrm{Complexity}_C}{q_Q} = O(\re^\lambda),
\end{equation}
with
\begin{aligns}
    D=1:\quad\lambda = \left\{\begin{array}{cc}
         0.333&\mathrm{~for}\quad N_C=1,  \\
         -0.191&\mathrm{~for}\quad N_C=2,  \\
         -0.783&\mathrm{~for}\quad N_C=3, \\
         -1.292&\mathrm{~for}\quad N_C=4, 
    \end{array} 
    \right.\\
    D=2:\quad\lambda = \left\{\begin{array}{cc}
         0.662&\mathrm{~for}\quad N_C=1,  \\
         0.314&\mathrm{~for}\quad N_C=2,  
    \end{array} 
    \right.
\end{aligns}
where we note that one can achieve quantum advantage for positive $\lambda$. Thus, while it is in principle possible for the complexity of a quantum algorithm to scale more favorably with the Reynolds number $\re$ than for classical algorithms, it is limited to low values of the Carleman truncation order $N_C$. Via the error model from Eq.~\eqref{eq:error_model}, this then limits applications to high-error outputs or low Reynolds number cases (values of $\re$ far from the threshold $\re_T$ require lower $N_C$ for the same error $\epsilon_C$). In Fig.~\ref{fig:query}, we present the dependence of $q_Q$ on the Reynolds number $\re$ and truncation order $N_C$, highlighting the region of potential quantum advantage (under the condition that $\re$ is below a threshold value $\re_T$). 

So far, we have discussed the complexity of outputting the coherently encoded solution. As discussed in Sec.~\ref{sec:info_extraction}, if we want to extract classical information about the drag force, we need to incur additional complexity cost scaling as $q_M=O(\re^{\beta/2})$. Then, the best we can hope for is
\begin{equation}
    \frac{\mathrm{Complexity}_C}{q_Mq_Q} =O\left(\re^{\frac{\beta (D-1)}{2}}\right).
\end{equation}
Hence, there is no chance for a quantum advantage for $D=1$, whereas for $D=3$ one can get an improvement by a factor of $\re^{\beta}$ at best. For the case of $\beta=3/4$ this replaces the classical scaling $O(\re^3)$ with the quantum $O(\re^{2.25})$. Moreover, if we use our numerical estimates on the behavior of the condition number from Eq.~\eqref{eq:chi1}, then all potential for quantum advantage is gone except for the case $D=2$ and $N_C=1$, where one could get an improvement by a factor of $\re^{0.287}$. However, let us emphasize that this extra complexity cost arises from the boundary nature of the drag force observable, and certain bulk observables of interest may avoid such costs, potentially offering greater quantum advantage.

Finally, let us reiterate that for the quantum algorithm to work, the Carleman truncation error $\epsilon_C$ needs to converge. Since our analysis in Sec.~\ref{sec:error} showed that such convergence happens only for Reynolds numbers $\re$ below a threshold value $\re_T$, the actual value of $\re_T$ for industrially-relevant problems is crucial. This is because one may face two opposing constraints: on the one hand, one needs large $\re$ for the quantum algorithm to overcome constant overheads and become more efficient than the classical algorithm; on the other hand, $\re$ may need to be low enough to guarantee error convergence. Also note that although the inclusion of boundaries, driving, and different initial states may potentially increase $\re_T$ (see Ref.~\cite{jennings2026simulating} for more details), we do not expect this effect to be significant. Hence, if the rescaling of the lattice velocity does not push $\re_T$ high enough, the applicability of the Carleman approach to simulating highly turbulent flows may be limited.

\bigskip

\textbf{Authors contributions:} Authors are listed alphabetically within each affiliation. DJ contributed to the block-encoding and compilation theory of the quantum algorithm. KK identified complexity challenges, developed the new model, designed the quantum algorithm, analyzed its complexity, and performed numerical simulations. ML identified complexity challenges, supported the development of the new model, the design of the quantum algorithm, and its complexity analysis. RA performed numerical studies of truncation errors and operator norms as input to complexity bounds. EM performed numerical comparison between LB and Carleman approximation. SR supported the definition of the end user requirements, and the application use cases. KK wrote the paper with the assistance of DJ and ML. 

\bigskip
	
\textbf{Acknowledgments:} We would like to thank Paul Mannix for his invaluable support and input on the physics of the model and CFD methods. We also want to thank Thomas Astoul, Thomas Bendokat, Giovanni Giustini, and Alessandro De Rosis for helpful discussions concerning fluid dynamics and the lattice Boltzmann method. Finally, we are grateful to Alice Barthe, Sam Morley-Short, and Trevor Vincent for their assistance with numerical simulations.

\clearpage
 

\appendix
\onecolumngrid


\section{Notation}
\label{app:notation}

\begin{table}[h]
\centering
\setlength{\tabcolsep}{8pt}
\renewcommand{\arraystretch}{1.0}
\begin{tabular}{cll }
\toprule
& \textbf{Symbol} & \textbf{Description} \\
\midrule
\multirow{11}{*}{\rotatebox{90}{\textbf{Fluid system}}}
&   $D$             & Number of spatial dimensions  \\
&   $\v{r}$         & Position in space\\
&   $t$             & Moment in time \\
&   $L_x,L_y,L_z$   & System size along $x$, $y$, and $z$ \\
&   $T$             & Total evolution time \\
&   $\v{u}$         & Fluid velocity \\
&   $c_s$           & Speed of sound \\
&   $L$             & Characteristic length of the system \\
&   $\nu$           & Kinematic viscosity of the fluid \\
&   $\re$           & Reynolds number \\
&   $\ma$           & Mach number \\
\midrule
\multirow{16}{*}{\rotatebox{90}{\textbf{Lattice Boltzmann model}}}
&   $Q$             & Number of discrete velocities \\
&   $\v{e}_m,\v{e}_m^\star$ & Discrete velocities in physical and lattice units \\
&   $\v{w}$         & Vector of Maxwell weightings for discrete velocities \\
&   $\v{f}$, $\v{g}$ & LBE and shifted LBE state vectors\\
&   $\v{f}^\eq,\v{g}^\eq$ & Equilibrium states for LBE and shifted LBE\\
&   $N_x,N_y,N_z$   & Number of spatial discretization points along $x$, $y$, and $z$  \\
&   $N$             & Total number of spatial discretization points \\
&   $\v{r}^\star$   & Position in lattice units \\
&   $t^\star$       & Time step\\
&   $T^\star$       & Total number of time steps \\
&   $\Delta x$      & Distance between neighboring lattice sites \\
&   $\Delta t$      & Duration of a time step \\
&   $\beta$         & Spatial resolution parameter \\
&   $\v{u}^\star$   & Fluid velocity in lattice units \\
&   $u^\star_\mx$   & Maximal fluid velocity during the evolution (in lattice units)\\
&   $u_0^\star$     & Lattice velocity constant controlling the norm of the solution \\
\midrule
\multirow{13}{*}{\rotatebox{90}{\textbf{Quantum algorithm}}}
&   $d$                 & Dimension of the LBE system  \\
&   $d_C$               & Dimension of the Carleman-embedded LBE system\\
&   $N_C$               & Carleman truncation order \\
&   $\epsilon_C$        & Carleman truncation error\\
&   $\v{y},\v{y}_k$     & Carleman vector and its $k$-th block\\
&   $F_1,F_2$           & Linear and quadratic collision matrices \\
&   $\C,\S$             & Carleman collision and streaming matrices \\
&   $\kappa_A$          & Condition number of a matrix $A$ \\
&   $\alpha_A$          & Block-encoding prefactor of a matrix $A$ \\
&   $A_H, \v{Y}_H$      & Linear system and its solution for the LBE history state \\
&   $A_F, \v{Y}_F$      & Linear system and its solution for the LBE final state \\
&   $W$                 & Number of qubits in the waiting register \\
\bottomrule
\end{tabular}
\caption{Symbols appearing in the text and their meaning.}
\end{table}


\section{Transforming DBE into an ODE}
\label{app:DBE_as_ODE}

With the discretization described by Eq.~\eqref{eq:discretization}, the streaming term in DBE, Eq.~\eqref{eq:DBE}, can be represented using central differencing as follows:
\begin{align}
    -\v{e}_m\cdot \nabla f_m(\v{r},t) \qquad \rightarrow \qquad -\sum_{i\in\{x,y,z\}} (\v{e}_m)_i \frac{f_{m,\v{r}+\v{a}_i}(t)-f_{m,\v{r}-\v{a}_i}(t)}{2 \Delta x},
\end{align}
where we introduced $\v{a}_x=(1,0,0)$, $\v{a}_y=(0,1,0)$ and $\v{a}_z=(0,0,1)$. This can be written concisely as
\begin{align}
    -\v{e}_m\cdot \nabla f_m(\v{r},t) \qquad \rightarrow \qquad G\v{f}(t),
\end{align}
where $G$ is a $d\times d$ matrix describing the streaming generator with matrix elements given by
\begin{equation}
    \label{eq:G}
    (G)_{m,\v{r} ; m_1,\v{r}_1} = - \delta_{m;m_1}\sum_{i\in\{x,y,z\}} \frac{(\v{e}_m)_i (\delta_{\v{r}+\v{a}_i;\v{r}_1}-\delta_{\v{r}-\v{a}_i;\v{r}_1})}{2 \Delta x},
\end{equation}
where $\delta$ denotes a Kronecker delta.

Next, one needs to rewrite the collision operator explicitly as a function of $\v{f}$. Using the fact that for the considered discrete velocity models the lattice speed of sound is given by $c^\star_s=1/\sqrt{3}$, the expression for the equilibrium distribution from Eq.~\eqref{eq:eq} can be rewritten as:
\begin{align}
    \label{eq:eq_lat}
    f_m^\eq(\v{r},t)= \rho(\v{r},t) w_m \left(1 + 3 \v{e}_m^\star\cdot \v{u}^\star(\v{r},t) + \frac{9}{2} (\v{e}_m^\star\cdot \v{u}^\star(\v{r},t))^2 -\frac{3}{2}|\v{u}^\star(\v{r},t)|^2\right).
\end{align}
Then, the authors of Refs.~\cite{li2025potential,penuel2024feasibility} assume a weakly compressible flow, which means that the deviation of $\rho$ from some reference density is small. Normalizing such reference density to 1, the following approximation is then justified:
\begin{align}
    \label{eq:compressible}
    \frac{1}{\rho(\v{r},t)}\approx 2-\rho(\v{r},t).
\end{align}
Using it in Eq.~\eqref{eq:eq_lat}, the equilibrium distribution can be expressed as a third order polynomial in~$\v{f}$ as follows (for clarity we omit the dependence on $\v{r}$ and $t$):
\begin{equation}
\begin{aligned}
    f_m^\eq =&  w_m\left(\sum_{m_1}^Q f_{m_1}+ 3 \v{e}_m^\star\cdot \sum_{m_1=1}^Q \v{e}_{m_1}^\star f_{m_1}\right)\\
    &+2w_m\left(\frac{9}{2}\left(\v{e}_m^\star\cdot\sum_{m_1=1}^Q\v{e}_{m_1}^\star f_{m_1}\right)^2-\frac{3}{2}\left(\sum_{m_1=1}^Q \v{e}_{m_1}^\star f_{m_1}\right)^2\right)\\
    &-w_m\left(\sum_{m_1=1}^Q f_{m_1}\right)\left(\frac{9}{2}\left(\v{e}_m^\star\cdot\sum_{m_1=1}^Q\v{e}_{m_1}^\star f_{m_1}\right)^2-\frac{3}{2}\left(\sum_{m_1=1}^Q \v{e}_{m_1}^\star f_{m_1}\right)^2\right),
\end{aligned}
\end{equation}
where we explicitly separated the expression into linear, quadratic and cubic terms.

Using the above, the BGK collision operator from Eq.~\eqref{eq:BGK} can be replaced by
\begin{equation}
    \Omega(\v{r},t) \qquad \rightarrow \qquad  F_1\v{f}(t) + F_2 \v{f}^{\otimes 2}(t)+F_3 \v{f}^{\otimes 3}(t),
\end{equation}
where $F_k$ are $d\times d^k$ matrices describing linear, quadratic and cubic contributions with matrix elements given by
\begin{aligns}
    \label{eq:F1_DBE}
    &(F_1)_{m,\v{r} ; m_1,\v{r}_1} = \frac{1}{\tau}\delta_{\v{r};\v{r}_1}(-\delta_{m;m_1}+w_m+3w_m E_{m;m_1}),\\[12pt]
    \label{eq:F2_DBE}
    &(F_2)_{m,\v{r} ; m_1,\v{r}_1,m_2,\v{r}_2} = \frac{w_m}{\tau}\delta_{\v{r};\v{r}_1}\delta_{\v{r};\v{r}_2}\left(9E_{m;m_1}E_{m;m_2}-3E_{m_1;m_2}\right),\\[12pt]
    \label{eq:F3_DBE}
    &(F_3)_{m,\v{r} ; m_1,\v{r}_1,m_2,\v{r}_2,m_3,\v{r}_3} =-\frac{w_m}{2\tau}\delta_{\v{r};\v{r}_1}\delta_{\v{r};\v{r}_2}\delta_{\v{r};\v{r}_3}\left(9E_{m;m_1}E_{m;m_2}-3E_{m_1;m_2}\right),
\end{aligns}
and where we have introduced the following $Q\times Q$ matrix:
\begin{align}
    E_{m;m_1}:= \v{e}_m^\star\cdot \v{e}_{m_1}^\star.
\end{align}


\section{Vanishing Carleman error for purely collisional DBE}
\label{app:vanishing_carleman}

Start with Eq.~\eqref{eq:DBEmat} with a single lattice site and no streaming:
\begin{align}
    \label{eq:DBEmat_no_streaming}
    \frac{d\v{f}(t)}{dt} = F_1\v{f}(t) + F_2 \v{f}^{\otimes 2}(t)+F_3 \v{f}^{\otimes 3}(t),
\end{align}
where $F_k$ are $Q\times Q^k$ from Eqs.~\eqref{eq:F1_DBE}-\eqref{eq:F3_DBE}. Given these small sizes, it can be verified by direct calculations that the linear collision matrix $F_1$ has only two degenerate eigenvalues: 0 and $-1/\tau$. Denoting by $V$ the matrix that diagonalizes $F_1$, we thus have
\begin{equation}
    V^{-1} F_1 V = -\frac{1}{\tau} \Pi_-, 
\end{equation}
where $\Pi_-$ is the projector onto the negative eigenspace, and we have $\Pi_0+\Pi_-=I$ with $\Pi_0$ denoting the projector onto the zero eigenspace of  $V^{-1} F_1 V $. Moreover, as can also be verified by direct calculations, the action of matrices $V^{-1}F_k V^{\otimes k}$ for $k$ being 2 or 3 is to send vectors from $\Pi_0^{\otimes k}$ to $\Pi_-$. In other words, their kernel is the complement of $\Pi_0^{\otimes k}$, and their image is $\Pi_-$. This then leads to a set of algebraic relations between matrices $F_k$. Rescaling time so that $\tau=1$ for simplicity, we then have:
\begin{aligns}
    &F_1 F_2= - F_2,\\
    &F_1 F_3= - F_3,\\
    &F_2(F_1\otimes I) = F_2(I\otimes F_1) = 0,\\
    &F_3(F_1\otimes I^{\otimes 2})= F_3(I\otimes F_1\otimes I) = F_3(I^{\otimes 2}\otimes F_1) = 0,\\
    &\exp(F_1x) F_2= e^{-x} F_2,\\
    &\exp(F_1x) F_3= e^{-x} F_3,\\
    &F_2(\exp(F_1x)\otimes I) = F_2(I\otimes \exp(F_1x)) = F_2,\\
    &F_3(\exp(F_1x)\otimes I^{\otimes 2})= F_3(I\otimes \exp(F_1x)\otimes I) = F_3(I^{\otimes 2}\otimes \exp(F_1x)) = F_3,\\
    &F_2(F_2\otimes I)=F_2(I\otimes F_2)=0,\\
    &F_2(F_3\otimes I)=F_2(I\otimes F_3)=0,\\
    &F_3(F_2\otimes I^{\otimes 2})=F_3(I^{\otimes 2}\otimes F_2)=0,\\
    &F_3(F_3\otimes I^{\otimes 2})=F_3(I\otimes F_3\otimes I)=F_3(I^{\otimes 2}\otimes F_3)=0.
\end{aligns}

We now claim that the following is a solution to Eq.~\eqref{eq:DBEmat_no_streaming}:
\begin{equation}
    \label{eq:DBE_postulated}
    \v{f}(t) = e^{F_1 t}\v{f}(0) + (1-e^{-t})\left[F_2\v{f}^{\ot 2}(0) + F_3\v{f}^{\ot 3}(0)\right],
\end{equation}
whose derivative can be straightforwardly calculated to be:
\begin{equation}
    \label{eq:explicit_derivative}
    \frac{d\v{f}(t)}{dt} = F_1 e^{F_1 t}\v{f}(0) + e^{-t}\left[F_2\v{f}^{\ot 2}(0) + F_3\v{f}^{\ot 3}(0)\right].
\end{equation}
Indeed, substituting the postulated solution to Eq.~\eqref{eq:DBEmat_no_streaming}, we get
\begin{align}
    \frac{d\v{f}(t)}{dt} =&~ 
    F_1\left[e^{F_1t}\v{f}(0)+(1-e^{-t}) F_2 \v{f}^{\ot 2}(0) +(1-e^{-t}) F_3 \v{f}^{\ot 3}(0)\right] \\
    & +F_2 \left[e^{F_1t}\v{f}(0)+(1-e^{-t})F_2 \v{f}^{\ot 2}(0) +(1-e^{-t})F_3 \v{f}^{\ot 3}(0) \right]^{\ot 2} \\
    & +F_3 \left[e^{F_1t}\v{f}(0)+(1-e^{-t})F_2 \v{f}^{\ot 2}(0) +(1-e^{-t})F_3 \v{f}^{\ot 3}(0) \right]^{\ot 3},
\end{align}
which can be simplified using the algebraic identities introduced above to
\begin{align}
    \frac{d\v{f}(t)}{dt} =&~ 
    F_1 e^{F_1t}\v{f}(0)-(1-e^{-t}) F_2 \v{f}^{\ot 2}(0) -(1-e^{-t}) F_3 \v{f}^{\ot 3}(0) \\
    & +F_2 \v{f}^{\ot 2}(0)\\
    & +F_3 \v{f}^{\ot 3}(0),
\end{align}
which in turn is exactly the derivative from Eq.~\eqref{eq:explicit_derivative}.

Next, let us introduce the Carleman variables,
\begin{equation}
    \v{y}_k(t) = \v{f}^{\ot k}(t),
\end{equation}
and assume we truncate the Carleman linearization at level $N_C$. Then, for all $k\leq N_C-2$ we have:
\begin{aligns}
    \frac{d\v{y}_k(t)}{dt} &= F_1^{(k)} \v{y}_k(t) + F_2^{(k)} \v{y}_{k+1}(t) + F_3^{(k)} \v{y}_{k+2}(t),\\
    \frac{d\v{y}_{N_C-1}(t)}{dt} &= F_1^{(N_C-1)} \v{y}_{N_C-1}(t) + F_2^{(N_C-1)} \v{y}_{N_C}(t),\\
    \frac{d\v{y}_{N_C}(t)}{dt} &= F_1^{(N_C)} \v{y}_{N_C}(t),
\end{aligns}
where 
\begin{equation}
    F_l^{(k)} = \sum_{j=0}^{k-1} I^{\ot j} \ot F_l \ot I^{\ot k- j -1}.
\end{equation}
We can now write a formal solution of the above:
\begin{aligns}
    {\v{y}}_k(t) &= e^{F_1^{(k)} t} \v{y}_k(0) + \int_0^t ds e^{F_1^{(k)}(t-s)} [F_2^{(k)}\v{y}_{k+1}(s) + F_3^{(k)}\v{y}_{k+2}(s)],\\
    {\v{y}}_{N_C-1}(t) &= e^{F_1^{(N_C-1)} t} \v{y}_{N_C-1}(0) + \int_0^t ds e^{F_1^{(N_C-1)}(t-s)} F_2^{(N_C-1)}\v{y}_{N_C}(s),\\
    {\v{y}}_{N_C}(t) &= e^{F_1^{(N_C)} t} \v{y}_{N_C}(0).
\end{aligns}
Assuming $N_C\geq 3$ and substituting such solutions for $\v{y}_2$ and $\v{y}_3$ into the solution for $\v{y}_1$, while using the algebraic identities listed above, one obtains
\begin{equation}
    \v{y}_1(t) = e^{F_1 t}\v{y}_1(0) + (1-e^{-t})\left[F_2\v{y}_2(0) + F_3\v{y}_3(0)\right]. 
\end{equation}
Finally, since for the truncated system at time $t=0$ we have $\v{y}_k(0)=\v{f}^{\ot k}(0)$, comparing with Eq.~\eqref{eq:DBE_postulated}, we conclude that
\begin{equation}
    \v{y}_1(t) = \v{f}(t). 
\end{equation}
Thus, whenever $N_C\geq 3$, the Carleman linearization procedure reproduces exactly the evolution of one copy of the DBE state vector with a single lattice site.


\section{Details on block-encoding}
\label{app:BE}


\subsection{Preliminaries}
\label{app:BE_preliminaries}

Recall that a block-encoding of a rectangular matrix $A$, mapping $m$ into $m'\leq m$ qubits, is defined as any unitary $U_A$ of the form
\begin{align}
    U_A = \begin{bmatrix}
        A/\alpha_A & * \\
        * & *
    \end{bmatrix}.
\end{align}
In the above, $*$ indicate other non-essential components of the unitary matrix $U_A$, and $\alpha_A$ is implicitly defined by the equation and is the block-encoding rescaling prefactor. More formally, an $(\alpha_A,n_A,\epsilon)$-block-encoding of $A$ is an $(m+n_A)$-qubit unitary $U_A$ such that
\begin{equation}
\label{eq:blockencodingdefinition}
    \!\!\!\!\left\| A\! - \! \alpha_A(I_{m'}\ot \bra{0}^{\ot n_A + \Delta m} )U_A(I_{m}\ot \ket{0}^{\ot n_A})\right\|\leq \epsilon,\!\!\!
\end{equation}
where $\Delta m=m-m'$. The parameter $\alpha_A$ typically enters the quantum algorithmic cost, so we ideally devise block-encodings with $\alpha_A$ as small as possible. Note that $\alpha_A \geq \|A\|$, since $U_A$ is unitary, so an \emph{optimal} block-encoding has $\alpha_A = \|A\|$.

We now note the following, where we assume $\epsilon=0$ for simplicity. Suppose we have a rectangular matrix $A$ and a block-encoding $U_A$ as defined above. This implies that for any $\ket{\psi}$ we have
\begin{equation}
    U_A \ket{\psi}\otimes \ket{0}^{\otimes n_A} = \frac{1}{\alpha_A} A \ket{\psi}\otimes \ket{0}^{\otimes n_A +\Delta m} + \ket{\perp},
\end{equation}
where $\ket{\perp}$ is a vector orthogonal to the first term on the right hand side. Given this, for any rectangular matrix $A$ we can define the \emph{square} matrix $\bar{A}$ in the embedding tensor product space as
\begin{equation}
    \bar{A} \ket{\psi} := A\ket{\psi} \otimes \ket{0}^{\otimes \Delta m}.
\end{equation}
It follows that if $U_A$ is a unitary block-encoding of $A$ with $n_A$ incoming ancillary zero states, $n_A + \Delta m$ outgoing ancillary zero states, and a prefactor $\alpha_A$, then $U_A$ is also a unitary block-encoding of $\bar{A}$ with $n_A$ ancillary qubits beginning and ending in the zero state, and with a rescaling prefactor $\alpha_{\bar{A}}=\alpha_A$. In what follows, we will use this fact in the context of the tensor-product embedding of the Carleman matrix. Note that, since the block-encoding prefactors are the same for the rectangular matrix and its square version, in the main text we always refer to $F_2$ collision matrices to avoid introducing excessive notation, whereas in the appendices below we will  mostly work with $\bar{F}_2$ matrices.

In what follows, we will present constructions of unitary block-encodings $U_{A_H}$ of $A_H$ (from Eq.~\eqref{eq:linearHistory}) and $U_{A_F}$ of $A_F$ (from Eq.~\eqref{eq:linearFinal}), which both consist of a series of nested block-encodings:
\begin{enumerate}
    \item First, we will explain how to construct $U_{A_H}$ and $U_{A_F}$ given the block-encoding $U_{\C}$ of $\C$ (Appendix~\ref{app:blockLin}).
    \item  Next, we show how to construct $U_{\C}$ given the block-encodings of collision matrices $I+F_1$ and $\bar{F}_2$ (Appendix~\ref{app:blockCarleman}).
   \item Finally, we present a way to block-encode collision matrices $I+F_1$ and $\bar{F}_2$ (Appendix~\ref{app:blockCollisions}).
\end{enumerate}  
This way we close the construction and provide the final expressions for the block-encoding prefactors $\alpha_{A_H}$ and $\alpha_{A_F}$, and the corresponding number of ancillary qubits needed, $n_{A_H}$ and $n_{A_F}$. 

Before we proceed, we present two crucial results concerning unitary block-encodings that will be useful in block-encoding the investigated linear systems describing the shifted LBE. Although these are well known results, we provide their proofs, as we will use the construction presented there in the coming sections. 
\begin{lemma}[LCU Block-encodings]
    \label{lem:LCU}
    Given an $N$-qubit matrix
    \begin{equation}
        A=\sum_{i=1}^k A_i,
    \end{equation}
    and $(\alpha_i,n_i,0)$-block-encodings $U_{A_i}$ of $A_i$, one can construct $(\alpha,n,0)$-block-encoding $U_A$ of $A$ with
    \begin{equation}
        \alpha=\sum_{i=1}^k\alpha_i,\qquad
        n=\log k + \max_i n_i. 
    \end{equation}
\end{lemma}
\begin{proof}
    Introduce a state preparation unitary $V$ on $\log k$ qubits, such that
    \begin{equation}
        V\ket{0}=\frac{1}{\sqrt{\alpha}}\sum_{i=1}^{k} \sqrt{\alpha_i}\ket{i},
    \end{equation}
    and denote $n^*=\max_i n_i$.
    Then, extend $(\alpha_i,n_i,0)$-block-encodings of $A_i$ to $(\alpha_i,n^*,0)$-block-encodings $U_{A_i}^*$ of $A_i$ simply by adding extra qubits on which $U_{A_i}^*$ acts trivially. Now, construct the following unitary
    \begin{equation}
        U_A=(I_N\ot I_{n^*}\ot V^\dagger)\left(\sum_{i=1}^k U_{A_i}^*\ot \ketbra{i}{i} \right) (I_N\ot I_{n^*}\ot V),
    \end{equation}
    where $I_m$ is the identity matrix acting on $m$ qubits. We can directly verify that
    \begin{align}
        &(I_N\ot\bra{0}^{\ot n^*+\log k })U_A(I_N\ot\ket{0}^{\ot n^*+\log k})\\
        &\quad=\frac{1}{\alpha}\sum_{j,j'=1}^k \sqrt{\alpha_{j}\alpha_{j'}}(I_N\ot\bra{0}^{\ot n^*}\ot \bra{j'})\left(\sum_{i=1}^k U_{A_i}^*\ot \ketbra{i}{i} \right)(I_N\ot\ket{0}^{\ot n^*}\ot\ket{j})\\
        &\quad=\frac{1}{\alpha}\sum_{i=1}^k \alpha_i(I_N\ot\bra{0}^{\ot n^*})U_{A_i}^*(I_N\ot\ket{0}^{\ot n^*})
        =\frac{1}{\alpha}\sum_{i=1}^k A_i=\frac{A}{\alpha},
    \end{align}
    so that $U_A$ is an $(\alpha,n,0)$-block-encoding of $A$.
\end{proof}

\begin{lemma}[Block-diagonal LCU]
    \label{lem:blockLCU}
    Given a $k\cdot 2^N\times k\cdot 2^N$ matrix
    \begin{equation}
        A=\sum_{i=k_1}^{k_2} \ketbra{i}{i}\ot A_i,
    \end{equation}
    with $1\leq k_1\leq k_2\leq k$, and $(\alpha_i,n_i,0)$-block-encodings $U_{A_i}$ of $N$-qubit matrices $A_i$, one can construct $(\alpha,n,0)$-block-encoding $U_A$ of $A$ with
    \begin{equation}
        \alpha=\max_{i\in\{k_1,\dots,k_2\}}\alpha_i,\qquad
        n=\max_{i\in\{k_1,\dots,k_2\}} n_i +1. 
    \end{equation}
\end{lemma}

\begin{proof}
For $i<k_1$ and $i>k_2$ set $\alpha_i=0$. Then, for all $i\in\{1,\dots,k\}$, introduce a set of single-qubit state preparation unitaries $V_i$, such that
    \begin{equation}
        V_i\ket{0}=\frac{\alpha_i}{\alpha}\ket{0} +\sqrt{1-\frac{\alpha_i^2}{\alpha^2}}\ket{1},
    \end{equation}
    and denote $n^*=\max_i n_i$. Next, for $i\in\{k_1,\dots,k_2\}$, extend $(\alpha_i,n_i,0)$-block-encodings $U_{A_i}$ of $A_i$ to $(\alpha_i,n^*,0)$-block-encodings $U_{A_i}^*$ of $A_i$ simply by adding extra qubits on which $U_{A_i}^*$ acts trivially. Moreover, for $i<k_1$ and $i>k_2$, define $U_{A_i}^*=I_N\ot I_{n^*}$. Now, construct the following unitary
    \begin{align}
        U_A=&(I_{\log k}\ot I_N\ot I_{n^*}\ot I_1)\left(\sum_{i=1}^k  \ketbra{i}{i}\ot U_{A_i}^* \ot I_1 
        \right)\left(\sum_{j=1}^k \ketbra{j}{j}\ot I_N\ot I_{n^*}\ot V_j\right),
    \end{align}
    where $I_m$ is the identity matrix acting on $m$ qubits. We can directly verify that
    \begin{align}
        &(I_{\log k}\ot I_N\ot\bra{0}^{\ot n^*+1 })U_A(I_{\log k}\ot I_N\ot\ket{0}^{\ot n^*+1})\\
        &\quad=(I_{\log k}\ot I_N\ot\bra{0}^{\ot n^*})\left(\sum_{i=1}^k \ketbra{i}{i}\ot U_{A_i}^* \bra{0} V_i \ket{0} \right)(I_{\log k}\ot I_N\ot\ket{0}^{\ot n^*}) \\
        &\quad=(I_{\log k}\ot I_N\ot\bra{0}^{\ot n^*})\left(\sum_{i=k_1}^{k_2} \ketbra{i}{i}\ot U_{A_i}^* \frac{\alpha_i}{\alpha} \right)(I_{\log k}\ot I_N\ot\ket{0}^{\ot n^*}) \\
        &\quad= \sum_{i=k_1}^{k_2} \ketbra{i}{i}\ot \frac{A_i}{\alpha_i}\frac{\alpha_i}{\alpha} = \frac{1}{\alpha}\sum_{i=1}^k \ketbra{i}{i}\ot A_i=\frac{A}{\alpha},
    \end{align}
    so that $U_A$ is an $(\alpha,n,0)$-block-encoding of $A$.
\end{proof}


\subsection{Block-encoding the linear system}
\label{app:blockLin}

\paragraph*{History state.}

First of all, note that we can decompose $A_H$ into a sum over the diagonal and off-diagonal blocks:
\begin{align}
    A_H &=\sum_{t^\star=0}^{T^\star} \ketbra{t^\star}{t^\star}\ot I - \sum_{t^\star=0}^{T^\star-1} \ketbra{t^\star+1}{t^\star}\ot \S\C =I_{t^\star} \ot I - \Delta_{t^\star} \ot \S\C, \label{eq:L_block}
\end{align}
where $\Delta_t^\star$ is just a $(T^\star+1)$-dimensional shift operator on the first (time) register. Recall that in this section we want to show how to block-encode $A_H$ given a block-encoding of $\C$. So, we need to show how to block-encode $\Delta_{t^\star}$ and how to realize $\S$.

It is easy to construct a $(1,1,0)$-block-encoding $U_{\Delta^k}$ of $\Delta^k$ (here we just need $k=1$) in the following way:
\begin{equation}
    \label{circ:Delta}
    \scalebox{1}{
    \begin{quantikz}[row sep={\the\qrow,between origins}, column sep=\the\qcol, baseline=(current bounding box.center), wire types={b}]
        \lstick{$\ket{t^\star}$}& \gate{\Delta^k} & \\
    \end{quantikz}
    =
     \begin{quantikz}[row sep={\the\qrow,between origins}, column sep=\the\qcol, baseline=(current bounding box.center), wire types={q,b}]
        \lstick{$\ket{0}_{a}$} & \gate[wires=2]{+k} 
        \gategroup[2,steps=1,style={dashed,rounded corners,fill=blue!10, inner xsep=2pt},background,label style={label position=below,anchor=north,yshift=-0.2cm}]{$U_{\Delta^k}$} 
        & \rstick{$\bra{0}_{a}$} \\
        \lstick{$\ket{t^\star}$}& & 
    \end{quantikz}
    },
\end{equation}
where the $+k$ gate is a quantum adder adding $k$ to a register consisting of $\log(T^\star+1)+1$ qubits. Note that the extra qubit, initialized in $\ket{0}_a$, plays the role of the most significant bit, and the notation $\bra{0}_a$ in Eq.~\eqref{circ:Delta} indicates that $\Delta^k$ is encoded in the $\ket{0}_a$ subspace, in the sense of Eq.~\eqref{eq:blockencodingdefinition} with $n_A=1$ and $m= m'$. The block-encoding of $(\Delta^\dagger)^k$ is analogous, just using a $-k$ quantum adder.

Next, since $\S$ is unitary, it has a trivial block-encoding and can be implemented as a quantum circuit in the following way. Note that the action of $S$ on a single state vector register $\v{y}^{[k]}$ can be constructed using quantum adders on the corresponding spatial registers controlled by the velocity registers:
\begin{equation}
    \label{circ:S}
    \begin{quantikz}[row sep={\the\qrow,between origins}, column sep=\the\qcol, wire types={b}]
        \lstick{$\ket{\v{y}^{[k]}}$}& \gate{S} & \qw
    \end{quantikz}
    \quad = \quad 
    \begin{quantikz}[row sep={\the\qrow,between origins}, column sep=\the\qcol, baseline=(current bounding box.center), wire types={b,b,b,q,q,q,q,q,q}]
        \lstick{$\ket{r_x^\star}$} & \gate{-1} & \gate{+1} & \qw & \qw & \qw & \qw & \qw\\
        \lstick{$\ket{r_y^\star}$} & \qw       & \qw       & \gate{-1} & \gate{+1} & \qw & \qw & \qw\\
        \lstick{$\ket{r_z^\star}$} & \qw       & \qw       & \qw & \qw & \gate{-1} & \gate{+1} & \qw\\
        \lstick[wires=2]{$\ket{v_x}$} & \ctrl{-3} & \qw & \qw & \qw & \qw & \qw& \qw \\[-0.5cm]
                                & \qw & \ctrl{-4} & \qw & \qw & \qw & \qw & \qw\\
        \lstick[wires=2]{$\ket{v_y}$} & \qw & \qw & \ctrl{-4} & \qw & \qw & \qw& \qw \\[-0.5cm]
                                & \qw & \qw & \qw & \ctrl{-5} & \qw & \qw& \qw \\
        \lstick[wires=2]{$\ket{v_z}$} & \qw & \qw & \qw & \qw & \ctrl{-5} & \qw & \qw\\[-0.5cm]
                                & \qw & \qw & \qw & \qw & \qw & \ctrl{-6}& \qw
    \end{quantikz}    
\end{equation}
Then, the direct way to construct $\S$ is to consecutively apply $S$ to the $k$-th single state vector register, controlled on the Carleman block register $n_C$ being smaller or equal to $k$. However, given that we encode the data in a way that the single state vector registers $\v{y}^{[k]}$ are all set to zero states when $k>n_C$, and that $S$ acts on these simply as the identity matrix,  we can as well remove the controls and encode $\S$ by just applying $S$ to each single state vector register,
\begin{equation}
    \label{circ:SS}
    \begin{quantikz}[row sep={\the\qrow,between origins}, column sep=\the\qcol, wire types={b}]
        \lstick{$\ket{\v{y}}$}& \gate{\S} & \qw
    \end{quantikz}
    \quad = \quad 
    \begin{quantikz}[row sep={\the\qrow,between origins}, column sep=\the\qcol, baseline=(current bounding box.center), wire types={b,b,n,b}]
        \lstick{$\ket{n_C}$} &  & \\
        \lstick{$\ket{\v{y}^{[1]}}$} & \gate{S} & \\
        \lstick{$\vdots$} & & \\
        \lstick{$\ket{\v{y}^{[N_C]}}$} & \gate{S} &
    \end{quantikz}    
\end{equation}
We shall follow this cheaper option.

Finally, let us assume that we have access to $(\alpha_\C,n_\C,0)$-block-encoding $U_\C$ of~$\C$, so that we can construct an $(\alpha_\C,n_\C+1,0)$-block-encoding of the second term in Eq.~\eqref{eq:L_block} using the constructions presented above. Then, using Lemma~\ref{lem:LCU}, we can combine this with the first identity term to construct an $(\alpha_{A_H},n_{A_H},0)$-block-encoding $U_{A_H}$ of $A_H$ with
\begin{equation}
    \label{eq:alphaLH}
    \alpha_{A_H}=\alpha_\C+1,\quad n_{A_H}=n_\C+2.
\end{equation}
The explicit circuit for $U_{A_H}$ is given by:
\begin{equation}
    \label{circ:A_H}
    \begin{quantikz}[row sep={\the\qrow,between origins}, column sep=\the\qcol, baseline=(current bounding box.center), wire types={b,b}]
        \lstick[wires=2]{$\ket{\v{Y}_H}$}& \wireoverride{n}&\lstick{$\ket{t^\star}$} \wireoverride{n} & \gate[wires=2]{A_H} &  \\
         & \wireoverride{n}&     \lstick{$\ket{\v{y}}$} \wireoverride{n} &  &\\
    \end{quantikz}
    \propto 
    \begin{quantikz}[row sep={\the\qrow,between origins}, column sep=\the\qcol, baseline=(current bounding box.center), wire types={q,b,b}]
        \lstick{$\ket{0}_{a}$} & & & \gate{V_{\alpha_{\C}}} 
        \gategroup[3,steps=5,style={dashed,rounded corners,fill=blue!10, inner xsep=2pt},background,label style={label position=below,anchor=north,yshift=-0.2cm}]{$U_{A_H}$} 
        & \ctrl{1} &  \ctrl{2} & \ctrl{2} & \gate {V_{\alpha_{\C}}^\dagger} & \rstick{$\bra{0}_a$}\\
        \lstick[wires=2]{$\ket{\v{Y}_H}$}& \wireoverride{n}&\lstick{$\ket{t^\star}$} \wireoverride{n} & \qw & \gate{U_{\Delta}} & \qw & \qw & \qw &\\
          & \wireoverride{n}&     \lstick{$\ket{\v{y}}$} \wireoverride{n} & \qw & \qw & \gate{U_{\C}} & \gate{\S} & \qw &
    \end{quantikz}
\end{equation}
where $V_\alpha$ is a single qubit gate defined via:
\begin{equation}
    \label{eq:V_alpha}
    V_\alpha \ket{0} \propto \ket{0} + \sqrt{\alpha}\ket{1},
\end{equation}
and, with a slight abuse of notation that we will use throughout this section, the block-encoded operators $U_\Delta$ and $U_\C$ also act on ancillary qubits that are not shown in the circuit.

\paragraph*{Final state.} Similarly to the case of $A_H$, we can decompose $A_F$ as follows:
\begin{align}
    A_F =& I_w\ot I_{t^\star} \ot I - \left[(I_w-\ketbra{0}{0}_w) \ot I_{t^\star} \ot I +\ketbra{0}{0}_w \ot I_{t^\star} \ot \S\C\right][ \Delta_{w,t^\star}\ot I],\label{eq:AF_block}
\end{align}
where $\Delta_{w,t^\star}$ is the shift operator on the waiting and time registers acting as:
\begin{align}
    \ket{0}_w\ot\ket{0}_{t^\star} &\rightarrow \ket{0}_w\ot\ket{1}_{t^\star} \rightarrow \dots \rightarrow \ket{0}_w\ot\ket{T^\star}_{t^\star}\rightarrow \ket{1}_w\ot\ket{0}_{t^\star} \rightarrow \dots \rightarrow \ket{2^W-1}_w\ot\ket{T^\star}_{t^\star}.
\end{align}
Assuming we have access to $(\alpha_\C,n_\C,0)$-block-encoding of $\C$ and to $(1,1,0)$-block-encoding of $\Delta$, we can use Lemma~\ref{lem:blockLCU} to construct an $(\alpha_\C,n_\C+2,0)$-block-encoding of the second term in Eq.~\eqref{eq:AF_block}. Thus, using Lemma~\ref{lem:LCU} as before, we end up with an $(\alpha_{A_F},n_{A_F},0)$-block-encoding of $A_F$ with
\begin{equation}
    \alpha_{A_F}=\alpha_\C+1,\quad n_{A_F}=n_\C+3,
\end{equation}
i.e., with the same block-encoding prefactor as in the case of $A_H$, but with one more ancillary qubit needed.  The explicit circuit for $U_{A_F}$ is given by:

\begin{equation}
    \label{circ:A_F}
    \!\!\!\!\!\!
    \begin{quantikz}[row sep={\the\qrow,between origins}, column sep=\the\qcol, baseline=(current bounding box.center), wire types={b,b,b}]
        \lstick[wires=3]{$\ket{\v{Y}_F}$}& \wireoverride{n}&\lstick{$\ket{w}$} \wireoverride{n} & \gate[wires=3]{A_F} &  \\
        & \wireoverride{n}&\lstick{$\ket{t^\star}$} \wireoverride{n} &  &  \\
         & \wireoverride{n}&     \lstick{$\ket{\v{y}}$} \wireoverride{n} &  &\\
    \end{quantikz}
    \propto
    \begin{quantikz}[row sep={\the\qrow,between origins}, column sep=\the\qcol, baseline=(current bounding box.center), wire types={q,q,b,b,b}]
        \lstick{$\ket{0}_{a_1}$} & & & \gate{V_{\alpha_{\C}}} 
        \gategroup[5,steps=6,style={dashed,rounded corners,fill=blue!10, inner xsep=2pt},background,label style={label position=below,anchor=north,yshift=-0.2cm}]{$U_{A_F}$} 
        & \ctrl{2} & \ctrl{1} &  \ctrl{4} & \ctrl{4} & \gate{V_{\alpha_{\C}}^\dagger} & \rstick{$\bra{0}_{a_1}$}\\
        \lstick{$\ket{0}_{a_2}$} & & &  &  & \gate{V_{\alpha_\C^2-1}} & & &  & \rstick{$\bra{0}_{a_2}$}\\[0.5cm]
        \lstick[wires=3]{$\ket{\v{Y}_F}$}& \wireoverride{n}&\lstick{$\ket{w}$} \wireoverride{n} & \qw & \gate[wires=2]{U_{\Delta}}& \ctrl[wire style={"{\ensuremath{I-\ketbra{0}{0}}}"}]{-1} & \octrl{1} & \octrl{1} & \qw &\\
        & \wireoverride{n}&\lstick{$\ket{t^\star}$} \wireoverride{n} & \qw & &  & \qw & \qw & \qw &\\
          & \wireoverride{n}&     \lstick{$\ket{\v{y}}$} \wireoverride{n} & \qw & \qw & & \gate{U_{\C}} & \gate{\S} & \qw &
    \end{quantikz}
\end{equation}
where both single qubit $V_\alpha$ unitaries are defined by Eq.~\eqref{eq:V_alpha}, whereas $I-\ketbra{0}{0}$ denotes a control that works unless all qubits are in a zero state. Note that this controlled gate can be implemented via
\begin{equation}
    \label{circ:w_control}
    \scalebox{1}{
    \begin{quantikz}[row sep={\the\qrow,between origins}, column sep=\the\qcol, baseline=(current bounding box.center), wire types={q,q,b}]
        \lstick{$\ket{0}_{a_1}$} &\ctrl{1}&   \\
        \lstick{$\ket{0}_{a_2}$} & \gate{V_{\alpha_\C^2-1}}  & \\[0.5cm]
        \lstick{$\ket{w}$} & \ctrl[wire style={"{\ensuremath{I-\ketbra{0}{0}}}"}]{-1} & 
    \end{quantikz}  
    }
     = 
    \scalebox{1}{
    \begin{quantikz}[row sep={\the\qrow,between origins}, column sep=\the\qcol, baseline=(current bounding box.center), wire types={q,q,q,q,n,q}]
        \lstick{$\ket{0}_{a_1}$} &\ctrl{1}& \ctrl{1} & \\
        \lstick{$\ket{0}_{a_2}$} & \gate{V_{\alpha_\C^2-1}} & \gate{V^\dagger_{\alpha_\C^2-1}} & \\
        \lstick{$\ket{w_1}$} & & \octrl{-1} & \\
        \lstick{$\ket{w_2}$} & & \octrl{-2} & \\
        \lstick{$\vdots$} & && &\\
        \lstick{$\ket{w_W}$} & &\octrl[wire style={dotted}]{-4} &
    \end{quantikz}  
    }
\end{equation}


\subsection{Block-encoding the Carleman collision matrix}
\label{app:blockCarleman}

First of all, note that we can decompose $\C$ from Eq.~\eqref{eq:collision_encoded} into a sum over the diagonal and off-diagonal blocks:
\begin{align}
    \C =&  \sum_{k=1}^{N_C} \sum_{l=k}^{\min\{2k,N_C\}} \ketbra{k}{l} \ot \bar{C}^k_l  \ot I^{\ot N_C-l}\\
    =& \sum_{k=1}^{N_C} \ketbra{k}{k} \ot \bar{C}^k_{k} \ot I^{\ot N_C-k} + \sum_{k=1}^{N_C} \sum_{l=1}^{\min\{k,N_C-k\}} \ketbra{k}{k+l} \ot \bar{C}^k_{k+l} \ot I^{\ot N_C-k-l} \\
    = &\sum_{k=1}^{N_C} \ketbra{k}{k} \ot \bar{C}^k_{k} \ot I^{\ot N_C-k} + \sum_{l=1}^{\lfloor \frac{N_C}{2}\rfloor} \sum_{k=l}^{N_C - l} \ketbra{k}{k}(\Delta^\dagger)^l \ot \bar{C}^k_{k+l} \ot I^{\ot N_C-k-l}.
\end{align}
Introducing $l':=\max\{1,l\}$, we rewrite $\C$ as
\begin{align}
    \C &=\sum_{l=0}^{\lfloor\frac{N_C}{2}\rfloor} \sum_{k=l'}^{N_C-l} \ketbra{k}{k}(\Delta^\dagger)^l \ot \bar{C}^k_{k+l} \ot I^{\ot N_C-k-l}=: \sum_{l=0}^{\lfloor\frac{N_C}{2}\rfloor}B_l.
\end{align}
Thus, assuming we found $(\alpha_{B_l},n_{B_l},0)$-block-encodings $U_{B_l}$ for each $B_l$, we can use Lemma~\ref{lem:LCU} to construct $(\alpha_{\C},n_{\C},0)$-block-encoding $U_\C$ of $\C$ with
\begin{equation}
    \alpha_\C =\sum_{l=0}^{\lfloor\frac{N_C}{2}\rfloor} \alpha_{B_l},\qquad
    n_\C=\log(\lfloor N_C/2\rfloor+1) + \max_{l\in\{0,\dots,\lfloor N_C/2\rfloor\}}n_{B_l}.
\end{equation}    
One can then explicitly construct $U_\C$ via a standard LCU circuit:
\begin{equation}
    \label{circ:C}
     \begin{quantikz}[row sep={\the\qrow,between origins}, column sep=\the\qcol, baseline=(current bounding box.center), wire types={b}]
        \lstick{$\ket{\v{y}}$}& \gate{\C} & \\
    \end{quantikz}
    \propto
     \begin{quantikz}[row sep={1.4cm,between origins}, column sep=\the\qcol, baseline=(current bounding box.center), wire types={b,b}]
        \lstick{$\ket{0}_{a}$} & \gate{V_B}  
        \gategroup[2,steps=5,style={dashed,rounded corners,fill=blue!10, inner xsep=2pt},background,label style={label position=below,anchor=north,yshift=-0.2cm}]{$U_{\C}$} 
        & \ctrl[wire style={"{\ensuremath{\ket{0}_a}}"}]{1}& ~\dots ~& \ctrl[wire style={"{\ensuremath{\ket{\lceil N_C/2 \rceil}_a}}"}]{1}  & \gate{V_B^\dagger} & \rstick{$\bra{0}_{a}$} \\
        \lstick{$\ket{\v{y}}$}& &\gate{U_{B_0}} & ~\dots ~ & \gate{U_{B_{\lfloor N_C/2\rfloor}}} &&
    \end{quantikz},
\end{equation}
where each $U_{B_l}$ is controlled on the ancillary register, consisting of $\log(\lceil N_C/2 \rceil +1)$ qubits, being in the state $\ket{l}$, and $V_B$ is defined via
\begin{equation}
    V_B\ket{0} \propto \sum_{l=0}^{\lfloor\frac{N_C}{2}\rfloor} \sqrt{\alpha_{B_l}} \ket{l}.
\end{equation}

In order to block-encode each $B_l$, 
\begin{equation}
    {B_l}=\left(\sum_{k=l'}^{N_C-l} \ketbra{k}{k} \ot \bar{C}^k_{k+l}\ot I^{\ot N_C-k-l}\right) [(\Delta^\dagger)^{l}\ot I^{\ot N_C}],
\end{equation}
we can block-encode the two terms separately. We have already shown in Appendix~\ref{app:blockLin} how to construct a $(1,1,0)$-block-encodings of the second term. From Lemma~\ref{lem:blockLCU}, we know that we can construct a $(\max_k \alpha_{\bar{C}_{k+l}^k},\max_k n_{\bar{C}_{k+l}^k}+1,0)$-block-encoding of the first term, given a $(\alpha_{\bar{C}_{k+l}^k},n_{C_{k+l}^k},0)$-block-encodings of $\bar{C}^k_{k+l}$. We thus conclude that we can construct $(\alpha_{B_l},n_{B_l},0)$-block-encodings of each $B_l$ with
\begin{equation}   
    \alpha_{B_l}=\max_{k\in\{l',\dots, N_C-l\}} \alpha_{\bar{C}^k_{k+l}},\qquad
    n_{B_l}= \max_{k\in\{l',\dots, N_C-l\}}  n_{\bar{C}^k_{k+l}}+1.
\end{equation}
The explicit circuit for each $U_{B_l}$ is given by
\begin{equation}
    \label{circ:B_l}
    \!\!\!  \!\! \begin{quantikz}[row sep={\the\qrow,between origins}, column sep=\the\qcol, baseline=(current bounding box.center), wire types={b,b}]
        \lstick[wires=2]{$\ket{\v{y}}$}& \wireoverride{n}&\wireoverride{n}&\lstick{$\ket{n_C}$} \wireoverride{n} & \gate[wires=2]{B_l} &  \\
         & \wireoverride{n}&\wireoverride{n}&     \lstick{$\ket{\tilde{\v{y}}}$} \wireoverride{n} &  &\\
    \end{quantikz}
    \propto \!
    \begin{quantikz}[row sep={1.2cm,between origins}, column sep=\the\qcol, baseline=(current bounding box.center), wire types={q,b,b}]
        \lstick{$\ket{0}_a$} & & & & \gategroup[3,steps=7,style={dashed,rounded corners,fill=blue!10, inner xsep=2pt},background,label style={label position=below,anchor=north,yshift=-0.2cm}]{$U_{B_l}$}& \gate{V_{C^{l'}_{l'+l}}}  & ~ \dots ~ & \gate{V_{C^{N_C-l}_{N_C}}} && &&\rstick{$\bra{0}_a$} \\
        \lstick[wires=2]{$\ket{\v{y}}$}& \wireoverride{n}&\wireoverride{n}&\lstick{$\ket{n_C}$} \wireoverride{n} & \gate{U_{(\Delta^\dagger)^l}} &\ctrl[wire style={"{\ensuremath{\ket{n_C=l'}}}"}]{1}\wire[u]{q} &~ \dots ~ & \ctrl[wire style={"{\ensuremath{\ket{n_C=N_C-l}}}"}]{1}\wire[u]{q} &&&&\\
          & \wireoverride{n}& \wireoverride{n}&    \lstick{$\ket{\tilde{\v{y}}}$} \wireoverride{n} & &\gate{U_{\bar{C}^{l'}_{l'+l}}} & ~ \dots ~  & \gate{U_{\bar{C}^{N_C-l}_{N_C}}} &&&&
    \end{quantikz}
\end{equation}
where $\tilde{\v{y}}$ denotes all the registers of $\v{y}$ except for the Carleman block register (i.e., single state vector registers $\v{y}^{[1]}$ through $\v{y}^{[N_C]}$), and $V_{C^k_{k+l}}$ are single qubit unitaries defined by
\begin{equation}
    V_{C^k_{k+l}}\ket{0} \propto  \ket{0} + \sqrt{\frac{\alpha_{B_l}^2}{\alpha_{\bar{C}^k_{k+l}}^2}-1}\ket{1}.
\end{equation}

Next, we introduce
\begin{equation}
    F(k,l):=( I+ F_1)^{\ot k-l} \ot \bar{F}_2^{\ot l},
\end{equation}
where $\bar{F}_2$ acts on two single state vector registers as $F_2\ot \ket{0}$, analogously to how $\bar{C}^k_l$ depends on $C^k_l$. We can then rewrite $\bar{C}^k_{k+l}$ as
\begin{align}
    \bar{C}^k_{k+l}&= \left(\sum_\pi \Pi_\pi F(k,l)\Pi^\dagger_\pi\right),
\end{align}
where the summation is over all permutations $\pi$ of a total of $k$ systems with $k-l$ systems of type~1 (each consisting of 1 subsystem) and $l$ systems of type 2 (each consisting of 2 subsystems). Moreover, $\Pi_\pi$ is a unitary permutation over all $k+l$ subsystems that:
\begin{enumerate}
    \item Acts on the $k-l$ systems of type 1, and on all the first subsystems of the $l$ systems of type $2$, as $\pi$;
    \item Then, each second subsystem of the $m$-th system of type $2$ is shifted forward by to a position $k+m$.
\end{enumerate} 

To block-encode $\bar{C}^{k}_{k+l}$ given block-encodings $U_{F(k,l)}$ of $F(k,l)$ we will use the LCU construction from Lemma~\ref{lem:LCU} again. Importantly, note that given a $(\alpha_{F(k,l)},n_{F(k,l)},0)$-block-encoding of $F(k,l)$, all its permuted versions appearing in the expression for $\bar{C}^{k}_{k+l}$ have block-encodings with the same prefactors and numbers of ancillary qubits. Thus, assuming we can construct $(\alpha_{F(k,l)},n_{F(k,l)},0)$-block-encodings $U_{F(k,l)}$ of each $F(k,l)$, we can use Lemma~\ref{lem:LCU} to construct $(\alpha_{\bar{C}^k_{k+l}},n_{\bar{C}^k_{k+l}},0)$-block-encodings $U_{\bar{C}^k_{k+l}}$ of $\bar{C}^k_{k+l}$ with 
\begin{equation}
    \!\!\! \alpha_{\bar{C}^k_{k+l}}=\binom{k}{l} \alpha_{F(k,l)},\quad
    n_{\bar{C}^k_{k+l}}=\log\binom{k}{l}+ n_{F(k,l)}.
\end{equation}
The straightforward circuit decomposition of $U_{\bar{C}^k_{k+l}}$ is given by
\begin{equation}
    \label{circ:Ckl}
    \!\!\!
    \begin{quantikz}[row sep={\the\qrow,between origins}, column sep=0.75em, baseline=(current bounding box.center), wire types={b}]
        \lstick{$\ket{\tilde{\v{y}}}$}& \gate{\bar{C}^k_{k+l}} & \\
    \end{quantikz}
    \propto\!\!\!\!
     \begin{quantikz}[row sep={1.4cm,between origins}, column sep=0.75em, baseline=(current bounding box.center), wire types={b,b}]
        \lstick{$\ket{0}_{a}$} & \gate{V_\Pi}  
        \gategroup[2,steps=9,style={dashed,rounded corners,fill=blue!10, inner xsep=2pt},background,label style={label position=below,anchor=north,yshift=-0.2cm}]{$U_{\bar{C}^k_{k+l}}$} 
        & \ctrl[wire style={"{\ensuremath{\ket{0}_a}}"}]{1}& ~\dots ~& \ctrl[wire style={"{\ensuremath{\ket{M}_a}}"}]{1}  & & \ctrl[wire style={"{\ensuremath{\ket{M}_a}}"}]{1} & ~\dots ~ & \ctrl[wire style={"{\ensuremath{\ket{0}_a}}"}]{1} &\gate{V_\Pi^\dagger} & \rstick{$\bra{0}_{a}$} \\
        \lstick{$\ket{\tilde{\v{y}}}$}& &\gate{\Pi_1} & ~\dots ~ & \gate{\Pi_M} & \gate{U_{F(k,l)}}& \gate{\Pi_M^\dagger} & ~\dots ~ & \gate{\Pi_1^\dagger} & &
    \end{quantikz}
\end{equation}
where we introduced a shorthand notation $M=\binom{k}{l}$, $\Pi_k$ means $\Pi_{\pi_k}$ with $\pi_k$ being the $k$-th permutation using some standard ordering, and $V_\Pi$ prepares a uniform superposition over $M$ states. Note, however, that in this construction the number of $(\log M)$-controlled gates is $M$, which can grow almost exponentially with $N_C$, rendering it an inefficient encoding for large $N_C$. However, in the sub-threshold regime where the Carleman method converges, we expect that $N_C$ will grow only logarithmically with the error, and so will result in an efficient encoding (as we shall see, the cost of these gates is subleading). However, we point out that one could encode the action of all permutations in the circuit above efficiently with growing $N_C$ using the quantum
Fisher-Yates shuffle~\cite{berry2018improved}. 

Finally, assume we have a $(\alpha_{ I+F_1},n_{ I+F_1},0)$-block-encoding of $( I+F_1)$ and  a $(\alpha_{\bar{F}_2},n_{\bar{F}_2},0)$-block-encoding of $\bar{F}_2$. Then, by simply applying these to subsequent single state vector registers, we can construct $(\alpha_{F(k,l)},n_{F(k,l)},0)$-block-encoding of $F(k,l)$ with
\begin{equation}
    \alpha_{F(k,l)} = \alpha_{ I+F_1}^{k-l} \alpha_{\bar{F}_2}^{l},\qquad
    n_{F(k,l)}=(k-l) n_{ I+F_1}+l n_{\bar{F}_2}.
\end{equation}
Putting it all together, we get $(\alpha_\C,n_\C,0)$-block-encoding of $\C$ with
\begin{aligns}
    \label{eq:alphaC}
    \alpha_\C &=  \sum_{l=0}^{\lfloor\frac{N_C}{2}\rfloor}\max_{k\in\{l',\dots,N_C-l\}} \binom{k}{l} \alpha_{ I+F_1}^{k-l} \alpha_{\bar{F}_2}^{l} = \sum_{l=0}^{\lfloor\frac{N_C}{2}\rfloor} \binom{N_C-l}{l} \alpha_{ I+F_1}^{N_C-2l} \alpha_{\bar{F}_2}^{l},\\
    n_\C & = 1 + \log\left(\left\lfloor \frac{N_C}{2}\right\rfloor+1\right) + \max_{l\in\{0,\dots, \lfloor N_C/2\rfloor\}}  \max_{k\in\{l',\dots, N_C-l\}}  \left(\log\binom{k}{l} + (k-l) n_{ I+F_1} + l  n_{\bar{F}_2}\right)\nonumber\\
    & = 1 + \log\left(\left\lfloor \frac{N_C}{2}\right\rfloor+1\right) + \max_{l\in\{0,\dots, \lfloor N_C/2\rfloor\}}  \left(\log\binom{N_C-l}{l} + (N_C-2l) n_{ I+F_1} + l n_{\bar{F}_2}\right).
\end{aligns}


\subsection{Block-encoding collision matrices}
\label{app:blockCollisions}

We now proceed to the final step of block-encoding the collision matrices $I+F_1$ and $\bar{F}_2$. Starting from $I+F_1$, we recall Eq.~\eqref{eq:F1_loc} and note that this matrix is local in space. Therefore, it acts non-trivially only on the velocity register,
\begin{align}
    I + F_1 & = I_{\v{r}^\star} \ot (I+\tilde{F}_1), 
\end{align}
where $\tilde{F}_1$ is of size $2^{2D} \times 2^{2D}$. In fact, since out of 4 basis states in each velocity register we only use 3, for $I+\tilde{F}_1$ it is enough to block-encode a $Q\times Q$ matrix (with $Q=3^D$ denoting the number of discrete velocities in the considered LBE model), as the rest of the unitary can be set to identity. Our problem then reduces to block-encoding a $Q \times Q$ matrix $I+\tilde{F}_1$.

To achieve this, we will use block-encoding based on singular value decomposition (SVD). More precisely, we can always write
\begin{equation}
    (I+\tilde{F}_1)= L_1 \Sigma_1 R_1^\dagger,\qquad
\end{equation}
where $L_1$ and $R_1$ are $Q\times Q$ unitary matrices, while $\Sigma_1$ is a $Q\times Q$ diagonal matrix (in fact, all these matrices have size $2^{2D}\times 2^{2D}$, but only the $Q\times Q$ block is nontrivial). Then, assuming access to the SVD unitaries, which act only on $2D$ qubits and whose explicit construction we present in Appendix~\ref{app:gates_F}, we only need to block-encode the diagonal matrix $\Sigma_1$. Thus, using Lemma~\ref{lem:blockLCU}, we get $(\alpha_{ I+F_1},1,0)$-block-encoding of $(I+F_1)$ with the block-encoding prefactor given by largest element $\sigma_*$ of $\Sigma_1$. This can be found using the explicit formula for $\tilde{F}_1$ and, in the relevant regime of $\taus > 1/2$, we get
\begin{equation}
    \label{eq:alpha_F1}
    \alpha_{I+F_1} = \sigma_* = \frac{
  \sqrt{
    3^{D}
    + 2^{D+1}\left(\bar{\tau}^{\star 2} - \taus\right)
    + \sqrt{
        9^{D}
        + 4 \cdot 6^{D}\left(\bar{\tau}^{\star 2} - \taus\right)
      }
  }
}{\sqrt{2^{D+1}}\,\taus},
\end{equation}
which is a non-increasing function of $\taus$ with values:
\begin{equation}
    \label{eq:alpha_F1_bounds}
    \renewcommand{\arraystretch}{2}
    1<\alpha_{I+F_1} < \left\{
    \begin{array}{ll}
         \sqrt{2 + \sqrt{3}}& \quad \mathrm{for}\quad D=1, \\
         \sqrt{\frac{1}{2}(7+3\sqrt{5})}& \quad \mathrm{for}\quad D=2,\\
         \frac{1}{2}\sqrt{23 + 3\sqrt{57}}& \quad \mathrm{for}\quad D=3.
    \end{array}
    \right.
\end{equation}
The explicit circuit block-encoding $I+F_1$ is given by:
\begin{equation}
    \label{circ:F1}
    \!\!\!  \!\! \begin{quantikz}[row sep={\the\qrow,between origins}, column sep=\the\qcol, baseline=(current bounding box.center), wire types={b,b}]
        \lstick[wires=2]{$\ket{\v{y}^{[k]}}$}& \wireoverride{n}&\wireoverride{n}&\lstick{$\ket{\v{r}^\star}$} \wireoverride{n} & \gate[wires=2]{I+{F}_1} &  \\
         & \wireoverride{n}&\wireoverride{n}&     \lstick{$\ket{\v{v}}$} \wireoverride{n} &  &\\
    \end{quantikz}
    \propto \!
    \begin{quantikz}[row sep={1.2cm,between origins}, column sep=0.85em, baseline=(current bounding box.center), wire types={b,b,q}]
        \lstick[wires=2]{$\ket{\v{y}^{[k]}}$}& \wireoverride{n}&\wireoverride{n}&\lstick{$\ket{\v{r}^\star}$} \wireoverride{n} & \gategroup[3,steps=7,style={dashed,rounded corners,fill=blue!10, inner xsep=2pt},background,label style={label position=below,anchor=north,yshift=-0.2cm}]{$U_{I+F_1}$} & &~ \dots ~ & &&&&\\
          & \wireoverride{n}& \wireoverride{n}&    \lstick{$\ket{\v{v}}$} \wireoverride{n} & \gate{R_1^\dagger} & \ctrl[wire style={"{\ensuremath{\ket{m=1}}}"}]{1} & ~ \dots ~  & \ctrl[wire style={"{\ensuremath{\ket{m=Q}}}"}]{1} &\gate{L_1}&&& \\
          \lstick{$\ket{0}_a$} & & & & & \gate{V_{\sigma_1}}  & ~ \dots ~ & \gate{V_{\sigma_Q}} && &&\rstick{$\bra{0}_a$} 
    \end{quantikz}
\end{equation}
where $\v{v}$ is the velocity register consisting of $2D$ qubits, controls depend on states encoding consecutive discrete velocities $m$, and $V_{\sigma_m}$ are single qubit unitaries defined via
\begin{equation}
    V_{\sigma_m} \ket{0} \propto \ket{0} + \sqrt{\frac{\sigma_*^2}{\sigma_m^2}-1}\ket{1},
\end{equation}
with $\sigma_m$ denoting the $m$-th singular value of $I+\tilde{F}_1$, i.e., the $m$-th diagonal element of $\Sigma_1$.

We now proceed to block-encoding $\bar{F}_2$. From Eq.~\eqref{eq:F2_loc}, we have
\begin{align}
    \bar{F}_2 & = I_{\v{r}^\star,\v{r}^\star} \ot \ket{0}_{\v{r}^{\star}} \ot \tilde{F_2} \ot \ket{0},
\end{align}
where $I_{\v{r}^\star,\v{r}^\star}$ is of size $N\times N^2$ and $\tilde{F_2}$ is of size $2^{2D}\times 2^{4D}$. As before, for the latter matrix we can in fact consider a restriction to a $Q\times Q^2$ matrix. We then need to construct block-encodings of two square matrices: $N^2\times N^2$ matrix $\bar{I}_{\v{r}^\star,\v{r}^\star}:=I_{\v{r}^\star,\v{r}^\star} \ot \ket{0}_{\v{r}^{\star}}$, and $Q^2\times Q^2$ matrix $\tilde{F}_2\otimes\ket{0}$. The first of these can be easily block-encoded with a unit prefactor and a single qubit ancilla using the following permutation matrix:
\begin{equation}
    U_{\bar{I}_{\v{r}^\star,\v{r}^\star}} = \bar{I}_{\v{r}^\star,\v{r}^\star} \ot I_a +(\bar{I}_{\v{r}^\star,\v{r}^\star}-I)\ot X_a,
\end{equation}
where $X_a$ is a Pauli $X$ matrix acting on the ancilla qubit. Concerning the second matrix, as we explained in Sec.~\ref{app:BE_preliminaries}, its block-encoding coincides with a block-encoding of a rectangular matrix $\tilde{F}_2$, and the prefactors are the same, $\alpha_{\tilde{F}_2\ot\ket{0}} = \alpha_{\tilde{F}_2}$. Therefore, assuming access to $(\alpha_{\tilde{F}_2},n_{\tilde{F}_2},0)$-block-encoding of $\tilde{F}_2$, we can construct $(\alpha_{\tilde{F}_2},n_{\tilde{F}_2}+1,0)$-block-encoding of $\bar{F}_2$.

As for the block-encoding of $\tilde{F}_2$, we will use a method based on higher-order singular value decomposition (HOSVD). More precisely, we aim at decomposing the $Q \times Q^2$ matrix $\tilde{F}_2$ as
\begin{equation}
       \tilde{F}_2 = L_2 \Sigma_2 (R_2\ot R_2)^\dagger,
\end{equation}
where $L_2$ and $R_2$ are $Q\times Q$ unitary matrices, while $\Sigma_2$ is a $Q\times Q^2$ with as small rank as possible. Then, again, we assume access to $2D$-qubit unitaries $L_2$ and $R_2$ (that we construct in Appendix~\ref{app:gates_F}), and we only need to block-encode a sparse matrix $\Sigma_2$. The exact block-encoding cost of $\Sigma_2$ will depend on the particular HOSVD used. However, as we show in Appendix~\ref{app:gates_F}, we find constructions for $D\in\{1,2\}$ where this coincides with the block-encoding prefactors for the diagonal matrix with singular values of $\tilde{F_2}$, while the number of ancillary qubits needed is 1 for $D=1$ and 8 for $D=2$. Since we expect that also for $D=3$ the block-encoding will not differ significantly from this, we will use the largest singular value of $\tilde{F}_2$ for the block-encoding prefactor $\alpha_{\tilde{F}_2}$, and $4D$ for the number of ancillary qubits needed. As the largest singular value can be explictly calculated, we end up with $(\alpha_{\bar{F}_2},n_{\bar{F}_2},0)$-block-encoding of $\bar{F}_2$ with
\begin{aligns}
    \label{eq:alpha_F2}
    \alpha_{\bar{F}_2} &= \frac{2}{3\taus}\left(\frac{3}{\sqrt{2}}\right)^D\sqrt{D+2},\\
    n_{\bar{F}_2} &= 4D + 1.
\end{aligns}
The circuit realizing the block-encoding of $\bar{F}_2$ is given by:
\begin{equation}
    \label{circ:F2}
     \!\!\!  \!\! \begin{quantikz}[row sep={\the\qrow,between origins}, column sep=\the\qcol, baseline=(current bounding box.center), wire types={b,b,b,b}]
        \lstick[wires=2]{$\ket{\v{y}^{[k]}}$}& \wireoverride{n}&\wireoverride{n}&\lstick{$\ket{\v{r}^\star}$} \wireoverride{n} & \gate[wires=4]{\bar{F}_2} &  \\
         & \wireoverride{n}&\wireoverride{n}&     \lstick{$\ket{\v{v}}$} \wireoverride{n} &  &\\
         \lstick[wires=2]{$\ket{\v{y}^{[k+1]}}$}& \wireoverride{n}&\wireoverride{n}&\lstick{$\ket{\v{v}}$} \wireoverride{n} & &  \\
         & \wireoverride{n}&\wireoverride{n}&     \lstick{$\ket{\v{r}^\star}$} \wireoverride{n} &  &
    \end{quantikz}
    \propto \!
    \begin{quantikz}[row sep={1.2cm,between origins}, column sep=\the\qcol, baseline=(current bounding box.center), wire types={b,b,b,b}]
          & \wireoverride{n}& \wireoverride{n}&    \lstick{$\ket{\v{v}}~~$} \wireoverride{n} \gategroup[4,steps=4,style={dashed,rounded corners,fill=blue!10, inner xsep=2pt},background,label style={label position=below,anchor=north,yshift=-0.2cm}]{$U_{\bar{F}_2}$} & \gate{R_2^\dagger} & \gate[wires=2]{U_{\Sigma_2\ot\ket{0}}} &  \gate{L_2}& \\
          & \wireoverride{n}&\wireoverride{n}&\lstick{$\ket{\v{v}}~~$} \wireoverride{n} & \gate{R_2^\dagger} & &  &   \\
          & \wireoverride{n}& \wireoverride{n}&    \lstick{$\ket{\v{r}^\star}~~$} \wireoverride{n} &  & \gate[wires=2]{U_{\bar{I}_{\v{r}^\star,\v{r}^\star}}} & &  \\
          & \wireoverride{n}&\wireoverride{n}&\lstick{$\ket{\v{r}^\star}~~$} \wireoverride{n} & & & &  
    \end{quantikz}
\end{equation}
where $U_{\Sigma_2 \otimes \ket{0}}$ is the block-encoding of the square embedding of $\Sigma_2$ in the full space. We also recall that in the circuit diagram we do not show the auxiliary qubits that begin and end in the zero state, for successful implementation of the block-encoding. 

Finally, putting everything we derived in this section together, we end up with the final expression for the block-encoding prefactor $\alpha_A$ and the number of ancillary qubits $n_A$ needed:
\begin{aligns}
    \label{eq:prefactor_final}
    \alpha_{A} =& 1+ \sum_{l=0}^{\lfloor\frac{N_C}{2}\rfloor} \binom{N_C-l}{l} \alpha_{ I+F_1}^{N_C-2l} \alpha_{\bar{F}_2}^{l},\\
    n_A =& \max_{l\in\{0,\dots, \lfloor N_C/2\rfloor\}}  \left[\log\binom{N_C-l}{l} + l(4 D-1)\right] + N_C + 1 + \log\left(\left\lfloor \frac{N_C}{2}\right\rfloor+1\right),
    \label{eq:ancilla_qubits}
\end{aligns}
where $\alpha_{I+F_1}$ and $\alpha_{\bar{F}_2}$ are given by Eqs.~\eqref{eq:alpha_F1} and \eqref{eq:alpha_F2}, respectively. 


\section{Details on QLS}
\label{app:QLS}

In Ref.~\cite{dalzell2024shortcut}, the author considers the linear problem of the form
\begin{equation}
    A\v{Y} = \v{b},
\end{equation}
where $\v{Y}$ is the desired solution vector, $\v{b}$ is a normalized vector, and it is assumed that one has access a unitary block-encoding $U_A$ of $A$ with a block-encoding prefactor $\alpha_A=1$, and to a unitary state preparation $U_{\v{b}}$. Under such assumptions, the singular values of $A$ satisfy
\begin{equation}
\label{eq:singularvaluesstandard}
    \sigma(A)\in[1/\kappa_A,1],
\end{equation}
where $\kappa_A$ denotes the condition number of $A$. In such a setting, the work Ref.~\cite{dalzell2024shortcut} provides a quantum algorithm for outputting a quantum state that is at most $\epsilon_Q$ away from the normalized solution $\ket{\v{Y}}$, with the following query upper bound 
    \begin{equation}
        q_Q \leq 56.0 \kappa_A +1.05 \kappa_A \ln\left(\!\frac{\sqrt{1\!-\!\epsilon_Q^2}}{\epsilon_Q}\!\right)+2.78\ln(\kappa_A)^3+3.17.
    \end{equation}
More precisely, $q_Q$ is the total number of queries to $U_{A}$, $U_{A}^\dagger$ and their controlled versions, and $2q_Q$ is the total number of queries to $U_{\v{b}}$, $U_{\v{b}}^\dagger$ and their controlled versions.

The difference then is that our block-encoding prefactor $\alpha_A$ is not 1 and that $\v{b}$ is not normalized. To bring our problem to the standardized form described above, note that the linear problems we consider, $A\v{Y}=\v{b}$, are equivalent to
\begin{equation}
    \tilde{A} \tilde{\v{Y}} = \v{b},
\end{equation}
where
\begin{equation}
    \tilde{A} = \frac{A}{\alpha_{A}},\qquad  \tilde{\v{Y}} = \frac{\alpha_{A}}{\|\v{b}\|}\v{Y}.
\end{equation}
Note that having access to $(\alpha_A,n_A,0)$-block-encoding of $A$ is equivalent to having access to $(1,n_A,0)$-block-encoding of $\tilde{A}$. Also, considering $\tilde{\v{Y}}$ instead of $\v{Y}$ makes no difference, as the QLS prepares a normalized quantum state, so it is insensitive to simple scaling. The only important difference then is that the range of singular values of $\tilde{A}$ differs from that of $A$:
\begin{equation}
    \sigma(\tilde{A}) \in \left[\frac{1}{\|\tilde{A}^{-1}\|},\|\tilde{A}\|\right] = \left[\frac{1}{\frac{\alpha_A\kappa_A}{\|A\|}}, \frac{\|A\|}{\alpha_A} 
    \right] \subseteq \left[\frac{1}{\frac{\alpha_A\kappa_A}{\|A\|}},1 
    \right],
\end{equation}
where we used that $\alpha_A \geq \|A\|$. Comparing with Eq.~\eqref{eq:singularvaluesstandard}, we then end up with the following query complexity upper bound for our problem:
\begin{align}    
        q_Q \leq &\frac{\alpha_A}{\|A\|} \Bigg(56.0 \kappa_A +1.05 \kappa_A \ln\Bigg(\frac{(1-\epsilon_Q^2)^{\frac{1}{2}}}{\epsilon_Q}\Bigg)\Bigg)+2.78\ln\left(\frac{\alpha_A}{\|A\|}\kappa_A\right)^3+3.17.
\end{align}


\section{Selected simulation parameters}
\label{app:parameters}

In Tables~\ref{tab:params_D1}~and~\ref{tab:params_D2}, we collect the values of simulation parameters $N_x$, $T^\star$, $\taus$, and $u^\star$ for $D=1$ and $D=2$ systems, and a representative subset of Reynolds numbers chosen in our simulations. We also present dimensions of Carleman collision matrices $\C$ (so also of the simulated Carleman systems used to evaluate Carleman truncation errors $\epsilon_C$) and the dimension of matrices $A_H$ describing Carleman linear systems (whose condition numbers we evaluate to estimate the query complexity of our quantum algorithm). Absence of $\mathrm{dim}~A_H$ in some cases means that we have not estimated the condition number for that choice of parameters (typically because sizes of the corresponding matrices were beyond the reach of our numerical methods in reasonable time). 

\renewcommand{\arraystretch}{1}

\begin{table}[t]
\centering
\begin{tabular}{|c|c|c|c|c|c|c|c|c|}
\hline
\boldmath{$N_C$} & \boldmath{$\re$} & \boldmath{$N_x$} & \boldmath{$T^\star$} & \boldmath{$\taus$} & \boldmath{$u^\star$} & \boldmath{$\mathrm{dim}~\C$} & \boldmath{$\mathrm{dim}~A_H$} \\
\hline\hline

\multirow{1}{*}{$1$}  & 1000 & 178 & 2372 & 0.540 & 0.0750 & 534 & 1267182 \\
\hline

\multirow{2}{*}{$2$}  & 200 & 54 & 388 & 0.6094 & 0.1361 & 26406 & 10271934 \\
                         & 1000 & 178 & 2372 & 0.540 & 0.0750 & 285690 & - \\
\hline

\multirow{2}{*}{$3$} & 100 & 32 & 178 & 0.6687 & 0.1768 & 894048 & 160034592 \\
                         & 500 & 106 & 1088 & 0.5617 & 0.0971 & 32258874 & - \\
                         
\hline

\multirow{2}{*}{$4$} & 30 & 13 & 46 & 0.8580 & 0.2774 & 2374320 & 111593040 \\
                         & 150 & 43 & 281 & 0.6310 & 0.1525 & 279086340 & - \\
\hline

\multirow{1}{*}{$5$} & 50 & 19 & 82 & 0.7602 & 0.2294 & 612436557 & - \\
\hline

\end{tabular}
\caption{\textbf{Simulation parameters for $D=1$ systems with $\beta=3/4$. }\label{tab:params_D1}}
\end{table}

\begin{table}[h!]
\centering
\begin{tabular}{|c|c|c|c|c|c|c|c|c|}
\hline
\boldmath{$N_C$} & \boldmath{$\re$} & \boldmath{$N_x$} & \boldmath{$T^\star$} & \boldmath{$\taus$} & \boldmath{$u^\star$} & \boldmath{$\mathrm{dim}~\C$} & \boldmath{$\mathrm{dim}~A_H$} \\
\hline\hline

\multirow{2}{*}{$1$} & 100 & 32 & 1000 & 0.5300 & 0.0313 & 9216 & 9225216 \\
& 250 & 63 & 3953 & 0.5120 & 0.0159 & 35721 & - \\
\hline

\multirow{2}{*}{$2$} & 20 & 10 & 90 & 0.6500 & 0.1000 & 810900 & 73791900 \\
                     & 150 & 43 & 1838 & 0.5200 & 0.0233 & 276939522 & - \\
\hline

\multirow{1}{*}{$3$}  & 20 & 10 & 90 & 0.6500 & 0.1000 & 729810900 & -\\
                         
\hline

\end{tabular}
\caption{\textbf{Simulation parameters for $D=2$ systems with $\beta=3/4$. }\label{tab:params_D2}}
\end{table}


\section{Eigenstate of collision-streaming operator}
\label{app:cs_eigenstate}

In this appendix, we present the construction of the $-1$ eigenstate $\ket{\xi_\pi}$ of the $\S\C$ operator. As in Sec.~\ref{sec:condition_bound}, we set all Carleman blocks, except for the first one, to zero:
\begin{equation}
    \ket{\xi_\pi} = [\ket{\zeta_\pi},0,\dots,0].
\end{equation}
Thus,
\begin{equation}
    \S\C \ket{\xi_\pi} = [S(I+F_1)\ket{\zeta_\pi},0,\dots,0],
\end{equation}
and so we reduce the problem to finding the $-1$ eigenstate of $S(I+F_1)$.

We start with $D=1$ case. It can be easily verified that the following state is an eigenstate of $I+F_1$ corresponding to eigenvalue 1:
\begin{equation}
    \ket{\psi} = \begin{tikzpicture}[baseline=-0.5ex]
		\coordinate (O) at (0,0);
		\fill (O) circle (2pt);
		\node[anchor=south, yshift=2pt] at (O) {$0$};
		\draw[-{Stealth[length=3mm,width=2mm]}] (-0.3,0) -- (-0.9,0) node[midway, anchor=south, xshift=-10pt,yshift=1pt] {$-1$};
		\draw[-{Stealth[length=3mm,width=2mm]}] (0.3,0) -- (0.9,0) node[midway, anchor=south, xshift=10pt,yshift=2pt] {$1$};
	\end{tikzpicture},
\end{equation}
where we used a graphical notation to relate coefficients to discrete velocity components. We then claim that the following state is a $-1$ eigenstate of $S(I+F_1)$
\begin{equation}
    \ket{\zeta_\pi} = ~\cdots~
		 \begin{tikzpicture}[baseline=-0.5ex]
			\fill (0,0) circle (2pt);
			\fill (1.5,0) circle (2pt);
			\fill (3,0) circle (2pt);
			\fill (-1.5,0) circle (2pt);
			\fill (-3,0) circle (2pt);
			\fill (0,0) circle (2pt);
			\node[anchor=south, yshift=2pt] at (0,0) {$\ket{\psi}$};
			\node[anchor=south, yshift=2pt] at (1.5,0) {$-\ket{\psi}$};
			\node[anchor=south, yshift=2pt] at (3,0) {$\ket{\psi}$};
			\node[anchor=south, yshift=2pt] at (-1.5,0) {$-\ket{\psi}$};
			\node[anchor=south, yshift=2pt] at (-3,0) {$\ket{\psi}$};
		\end{tikzpicture}
		~\cdots,
\end{equation}
where different dots correspond to different lattice sites. One can indeed verify that
the application of $I+F_1$ at every site leaves all local states invariant, which is then followed by the streaming operator mapping every $\ket{\psi}$ into $-\ket{\psi}$ and vice versa. Thus, the overall effect is to map $\ket{\zeta_\pi}$ to $-\ket{\zeta_\pi}$. Moreover, note that due to the antisymmetry of the state $\ket{\psi}$ with respect to the reflection around the node point, $\ket{\zeta_\pi}$ is the $-1$ eigenstate of $S(I+F_1)$ even in the presence of walls.

Next, consider $D=2$. It can be verified by direct calculation that the following two states are eigenstates of $I+F_1$ corresponding to the eigenvalue 1:
\begin{equation}
    \ket{\psi_1} = \begin{tikzpicture}[baseline=-0.5ex]
		\coordinate (O) at (0,0);
		\fill (O) circle (2pt);
		\node[anchor=south, yshift=0pt] at (O) {$0$};		
		\draw[-{Stealth[length=3mm,width=2mm]}] (0,0.5) -- ( 0, 1.2) node[midway, anchor=south,yshift=10pt] {$-2$};   
		\draw[-{Stealth[length=3mm,width=2mm]}] (0.5,0.5) -- ( 0.85, 0.85) node[midway, above left,xshift=15pt,yshift=5pt] {$-1$}; 
		\draw[-{Stealth[length=3mm,width=2mm]}] (0.5,0) -- ( 1.2, 0) node[midway,xshift=20pt,yshift=0pt] {$-2$};   
		\draw[-{Stealth[length=3mm,width=2mm]}] (0.5,-0.5) -- ( 0.85,-0.85) node[midway, below left,xshift=15pt,yshift=-5pt] {$0$}; 
		\draw[-{Stealth[length=3mm,width=2mm]}] (0,-0.5) -- ( 0,-1.2) node[midway,xshift=0pt,yshift=-20pt] {$2$};   
		\draw[-{Stealth[length=3mm,width=2mm]}] (-0.5,-0.5) -- (-0.85,-0.85) node[midway, below right,xshift=-15pt,yshift=-5pt] {$1$}; 
		\draw[-{Stealth[length=3mm,width=2mm]}] (-0.5,0) -- (-1.2,0) node[midway,xshift=-20pt,yshift=0pt] {$2$};    
		\draw[-{Stealth[length=3mm,width=2mm]}] (-0.5,0.5) -- (-0.85,0.85) node[midway, above right,xshift=-15pt,yshift=5pt] {$0$}; 
	\end{tikzpicture},\qquad
    \ket{\psi_2} = \begin{tikzpicture}[baseline=-0.5ex]
		\coordinate (O) at (0,0);
		\fill (O) circle (2pt);
		\node[anchor=south, yshift=0pt] at (O) {$0$};
		
		\draw[-{Stealth[length=3mm,width=2mm]}] (0,0.5) -- ( 0, 1.2) node[midway, anchor=south,yshift=10pt] {$2$};   
		\draw[-{Stealth[length=3mm,width=2mm]}] (0.5,0.5) -- ( 0.85, 0.85) node[midway, above left,xshift=15pt,yshift=5pt] {$0$}; 
		\draw[-{Stealth[length=3mm,width=2mm]}] (0.5,0) -- ( 1.2, 0) node[midway,xshift=20pt,yshift=0pt] {$-2$};   
		\draw[-{Stealth[length=3mm,width=2mm]}] (0.5,-0.5) -- ( 0.85,-0.85) node[midway, below left,xshift=15pt,yshift=-5pt] {$-1$}; 
		\draw[-{Stealth[length=3mm,width=2mm]}] (0,-0.5) -- ( 0,-1.2) node[midway,xshift=0pt,yshift=-20pt] {$-2$};   
		\draw[-{Stealth[length=3mm,width=2mm]}] (-0.5,-0.5) -- (-0.85,-0.85) node[midway, below right,xshift=-15pt,yshift=-5pt] {$0$}; 
		\draw[-{Stealth[length=3mm,width=2mm]}] (-0.5,0) -- (-1.2,0) node[midway,xshift=-20pt,yshift=0pt] {$2$};    
		\draw[-{Stealth[length=3mm,width=2mm]}] (-0.5,0.5) -- (-0.85,0.85) node[midway, above right,xshift=-15pt,yshift=5pt] {$1$}; 
	\end{tikzpicture}.
\end{equation}
We then claim that the following state is a $-1$ eigenstate of $S(I+F_1)$:
	\begin{equation}
		\ket{\zeta_\pi} = \cdots~
		\begin{tikzpicture}[baseline=-0.5ex]
			\fill (0,0) circle (2pt);
			\fill (1.5,0) circle (2pt);
			\fill (3,0) circle (2pt);
			\fill (-1.5,0) circle (2pt);
			\fill (-3,0) circle (2pt);
			\fill (0,1.5) circle (2pt);
			\fill (1.5,1.5) circle (2pt);
			\fill (3,1.5) circle (2pt);
			\fill (-1.5,1.5) circle (2pt);
			\fill (-3,1.5) circle (2pt);
			\fill (0,3) circle (2pt);
			\fill (1.5,3) circle (2pt);
			\fill (3,3) circle (2pt);
			\fill (-1.5,3) circle (2pt);
			\fill (-3,3) circle (2pt);
			\fill (-1.5,-1.5) circle (2pt);
			\fill (-3,-1.5) circle (2pt);
			\fill (0,-1.5) circle (2pt);
			\fill (1.5,-1.5) circle (2pt);
			\fill (3,-1.5) circle (2pt);
			\fill (3,-3) circle (2pt);
			\fill (-1.5,-3) circle (2pt);
			\fill (-3,-3) circle (2pt);
			\fill (0,-3) circle (2pt);
			\fill (1.5,-3) circle (2pt);
			\node[anchor=south, yshift=2pt] at (0,0) {$-\ket{\psi_1}$};
			\node[anchor=south, yshift=2pt] at (1.5,0) {$\ket{\psi_2}$};
			\node[anchor=south, yshift=2pt] at (3,0) {$-\ket{\psi_1}$};
			\node[anchor=south, yshift=2pt] at (-1.5,0) {$\ket{\psi_2}$};
			\node[anchor=south, yshift=2pt] at (-3,0) {$-\ket{\psi_1}$};
			\node[anchor=south, yshift=2pt] at (0,3) {$-\ket{\psi_1}$};
			\node[anchor=south, yshift=2pt] at (1.5,3) {$\ket{\psi_2}$};
			\node[anchor=south, yshift=2pt] at (3,3) {$-\ket{\psi_1}$};
			\node[anchor=south, yshift=2pt] at (-1.5,3) {$\ket{\psi_2}$};
			\node[anchor=south, yshift=2pt] at (-3,3) {$-\ket{\psi_1}$};
			\node[anchor=south, yshift=2pt] at (0,-3) {$-\ket{\psi_1}$};
			\node[anchor=south, yshift=2pt] at (1.5,-3) {$\ket{\psi_2}$};
			\node[anchor=south, yshift=2pt] at (3,-3) {$-\ket{\psi_1}$};
			\node[anchor=south, yshift=2pt] at (-1.5,-3) {$\ket{\psi_2}$};
			\node[anchor=south, yshift=2pt] at (-3,-3) {$-\ket{\psi_1}$};
			\node[anchor=south, yshift=2pt] at (0,1.5) {$-\ket{\psi_2}$};
			\node[anchor=south, yshift=2pt] at (1.5,1.5) {$\ket{\psi_1}$};
			\node[anchor=south, yshift=2pt] at (3,1.5) {$-\ket{\psi_2}$};
			\node[anchor=south, yshift=2pt] at (-1.5,1.5) {$\ket{\psi_1}$};
			\node[anchor=south, yshift=2pt] at (-3,1.5) {$-\ket{\psi_2}$};
			\node[anchor=south, yshift=2pt] at (0,-1.5) {$-\ket{\psi_2}$};
			\node[anchor=south, yshift=2pt] at (1.5,-1.5) {$\ket{\psi_1}$};
			\node[anchor=south, yshift=2pt] at (3,-1.5) {$-\ket{\psi_2}$};
			\node[anchor=south, yshift=2pt] at (-1.5,-1.5) {$\ket{\psi_1}$};
			\node[anchor=south, yshift=2pt] at (-3,-1.5) {$-\ket{\psi_2}$};
		\end{tikzpicture}
		~\cdots
	\end{equation}
This can be confirmed by verifying that the streaming operator brings a negative phase factor, and that this still holds in the presence of nontrivial walls due to the antisymmetric nature of $\ket{\psi_1}$ and $\ket{\psi_2}$. 

Finally, in $D=3$, there exist three eigenstates of $I+F_1$ corresponding to eigenvalue 1 and being antisymmetric under the reflection with respect to the node point. One can then again find a spatial distribution of these states with phase factors to construct a $-1$ eigenstate of $S(I+F_1)$.


\section{Upper bound on the condition number}
\label{app_condition_upper}

Consider a block decomposition of a matrix $A$,
\begin{equation}
        A=\sum_{i=1}^d \sum_{j=1}^d \ketbra{i}{j}\ot A_{ij},
\end{equation}
and note that
\begin{align}
    \max_{\v{x}:~\|\v{x}\|=1} \|A\v{x}\| &= \max_{\substack{\{\v{x}_k\}\\ \sum_k \|\v{x}_k\|^2=1} }\left\|\sum_{i=1}^d\sum_{j=1}^d\sum_{k=1}^d \left(\ketbra{i}{j}\ot A_{ij}\right)(\ket{k}\ot \v{x}_k)\right\|\\
    &= \max_{\substack{\{\v{x}_k\}\\ \sum_k \|\v{x}_k\|^2=1} }\left\|\sum_{i=1}^d\ket{i}\ot\sum_{k=1}^d A_{ik}\v{x}_k\right\|= \max_{\substack{\{\v{x}_k\}\\ \sum_k \|\v{x}_j\|^2=1} }\sqrt{\sum_{i=1}^d\left\|\sum_{k=1}^d A_{ik}\v{x}_k\right\|^2}\\
    &\leq \max_{\substack{\{\v{x}_k\}\\ \sum_k \|\v{x}_k\|^2=1} }\sqrt{\sum_{i=1}^d\left(\sum_{k=1}^d \|A_{ik}\|\|\v{x}_k\|\right)^2} \leq \max_{\substack{\{\v{x}_k\}\\ \sum_k \|\v{x}_k\|^2=1} }\sqrt{\sum_{i=1}^d\left(\sum_{j=1}^d \|A_{ij}\|^2\sum_{k=1}^d \|\v{x}_k\|^2\right)}\\
    &= \sqrt{\sum_{i=1}^d\sum_{j=1}^d \|A_{ij}\|^2}.
\end{align}
Thus,
\begin{equation}
    \label{eq:norm_block_bound}
    \|A\| \leq \sqrt{\sum_{i=1}^d \sum_{j=1}^d \|A_{ij}\|^2}.
\end{equation}

Applying the above to bound the norm of $A_H^{-1}$ from Eq.~\eqref{eq:AHInv}, noting that one can sum blocks over diagonals, leads to
\begin{equation}
    \|A_H^{-1}\|\leq \sqrt{\sum_{t^\star=0}^{T^\star}  (T^\star-t^\star+1)\|(\S\C)^{t^\star}\|^2}.
\end{equation}
Combining this with an upper bound for $\|A_H\|$ from Eq.~\eqref{eq:A_norm_bound} leads to the claimed bound on the condition number from Eq.~\eqref{eq:kappa_upper_H}. Next, recall that $A_F^{-1}$ has the same blocks as $A_H^{-1}$ followed by rows of blocks, where each row $k\in\{1,\dots, (2^{W}-1)(T^\star+1)\}$ consists of the blocks from the last row of $A_H^{-1}$ followed by $k$ identity blocks. Thus, we get the following upper bound
\begin{align}
    \|A_F^{-1}\|^2 & \leq  \sum_{t^\star=0}^{T^\star}  (T^\star-t^\star+1)\|(\S\C)^{t^\star}\|^2 + (2^{W}-1)(T^\star+1)\sum_{t^\star=0}^{T^\star} \|(\S\C)^{t^\star}\|^2 + \sum_{k=1}^{(2^{W}-1)(T^\star+1)} k\\
    & = \sum_{t^\star=0}^{T^\star}  ( 2^{W}(T^\star+1)-t^\star)\|(\S\C)^{t^\star}\|^2 +  \frac{1}{2}(2^W-1)(T^\star+1)(2^W+(2^W-1)T^\star)\\
    & \leq \sum_{t^\star=0}^{T^\star}  ( 2^{W}(T^\star+1)-t^\star)\|(\S\C)^{t^\star}\|^2 +  2^{2W-1}(T^\star+1)^2.
\end{align}
Combining it with the bound for $\|A_F\|$, which is the same as for $\|A_H\|$, we end up with the claimed bound for $\kappa_{A_F}$ from Eq.~\eqref{eq:kappa_upper_F}.


\section{Details on estimating gate complexity}
\label{app:gates}


\subsection{Preliminaries}
\label{app:gates_pre}

Our aim is to estimate the cost of a single query to $U_{A_F}$ or $U_{A_H}$ in terms of the number of $T$ gates needed. We denote the $T$-cost of implementing exactly a general unitary circuit $U$ by $G[U]$, and the cost of implementing it up to error $\epsilon$ by $G[U^{\epsilon}]$. Moreover, a gate $U$ controlled on $k$ qubits will be denoted by $c^kU$. In what follows, we focus on estimating $G[U_{A_F}]$, which also acts as an upper bound for $G[U_{A_H}]$. 

In the process, we will use the following estimates of $T$-costs based on known results on circuit compilation. First, we focus on estimating the cost of a general unitary $U$ controlled on $k$ qubits, as a function of the cost of implementing just $U$. For $k\geq 2$, one can realize $c^kU$ with $2k-2$ Toffoli gates, $k-1$ auxiliary qubits and a single $cU$ gate (for completeness, we explain this in Appendix~\ref{app:gates_extra},). Moreover, since each Toffoli gate can be replaced with 7 $T$ gates~\cite{selinger2013quantum,amy2013meet}, we we can realize an arbitrary $c^kU$ with a $T$-count of 
\begin{equation}
    G[c^kU] = 14k-14 + G[cU].
\end{equation}
Next, given a decomposition of $U$ into Clifford and $T$ gates, one can estimate the $T$-cost of $cU$ in the following way. First, all the controlled versions of the single-qubit Clifford gates can be realized with CNOTs and Cliffords, so they do not contribute to $G[cU]$. Then, every CNOT appearing in the decomposition of $U$ becomes a Toffoli in $cU$, and so contributes 7 $T$ gates. Finally, every $T$ gate in $U$ becomes a controlled $T$ in $cU$, and each such gate can be constructed from $9$ $T$ gates~\cite{selinger2013quantum,amy2013meet}. Here, we will make a crude approximation by assuming that the number of CNOTs in the decomposition of $U$ is roughly equal to the number of $T$ gates. We thus approximate the $T$-cost of $cU$ by
\begin{equation}
    G[cU] \approx 16 G[U].
\end{equation}
Consequently, for $k\geq 2$, we get
\begin{equation}
     G[c^kU] \approx 14k- 14 + 16G[U].
\end{equation}
For the special case when $U=X$ is a single qubit NOT gate, we use a slightly better estimate based on Ref.~\cite{khattar2025rise}:
\begin{equation}
    G[c^kX]= 8k -12,
\end{equation}
which requires a single extra ancilla. Moreover, we take the $T$-cost of a controlled Hadamard gate to be given by
\begin{equation}
    G[cH] = 2,
\end{equation}
due to a decomposition $cH=(I\ot R_y(\pi/4))cZ (I\ot R_y^\dagger(\pi/4))$ and $R_y(\pi/4) = S H T H S^\dagger$.

Now, we proceed to the estimates for the cost of implementing an $n$-qubit unitary $U$. For $n=1$, the cost of implementing a single qubit rotation $R_{\hat{n}}(\theta)$ around axis $\hat{n}\in\{x,y,z\}$ up to error $\epsilon$ is given by~\cite{ross2014optimal}:
\begin{equation}
    G[R_{\hat{n}}^\epsilon(\theta)] \approx 3\log\frac{1}{\epsilon}.
\end{equation}
We choose equal error $\epsilon$ on all approximate single-qubit gates used, and assume that the total error arising from gates implemented up to $\epsilon_1$ and $\epsilon_2$ is $\epsilon_1+\epsilon_2$. For $n\geq 2$, one can estimate the $T$-cost for synthesizing $U$ in terms of the $T$-costs $G[\mathrm{PREP}^{\epsilon_j}(u_j)]$, where $\mathrm{PREP}^{\epsilon_j}(u_j)$ is a state preparation of the $j$-th column of $U$ to error $\epsilon_j$ (see Appendix~\ref{app:gates_extra} for details). For an $n$-qubit unitary $U$ with $d$ columns $u_j$ specified, this estimate is given by
\begin{equation}
    \label{eq:unitarysynthesiscost}
       G [U^\epsilon] = 32 \sum_{j=1}^{d}G[\mathrm{PREP}^{\epsilon_j} (u_j)] + d(8n -12),
\end{equation}
where the resulting error $\epsilon$ depends on $\epsilon_j$ as
\begin{equation}
    \epsilon = \sqrt{2}\sum_{j=1}^{d} \epsilon_j.
\end{equation}
Note that such unitary state preparations can be constructed from single-qubit rotations, and that a uniform superposition over $2^k$ states can be prepared using only Clifford gates. Moreover, a state preparation unitary $\mathrm{PREP(\psi)}$ of a general $n$-qubit state $\ket{\psi}$ can be realized using roughly $2^{n+1}$ single qubit rotations~\cite{mottonen2004transformation}, and so the cost of implementing it up to error $2^{n+1}\epsilon$ can be approximated by
\begin{equation}
 \label{eq:n_qubit_state_prep}
   G[\mathrm{PREP}^{2^{n+1}\epsilon}(\psi)] \approx 3\cdot 2^{n+1} \log\frac{1}{\epsilon} .
\end{equation}
Finally, we assume that the $T$-cost of a quantum adder acting on $n$ qubits is given by~\cite{gidney2018halving}
\begin{equation}
    G[\add(n)] \approx  4n.
\end{equation}


\subsection{Gate cost of the linear system circuit}
\label{app:gates_A}

We start by using the gate decomposition of $U_{A_F}$ presented in Eq.~\eqref{circ:A_F}, and employing Eq.~\eqref{circ:w_control}, to get
\begin{align}
    G[U_{A_F}] =&  G[cU_\Delta] + G[c^2\S] + G[c^2U_\C] + 2G[V_{\alpha_\C}]+G[cV_{\alpha^2_\C-1}] + G[c^{W+1}V_{\alpha^2_\C-1}].
\end{align}
All the state preparation unitaries $V$ are single-qubit $y$ rotations, and so implementing them up to error $\epsilon$ has the following cost:
\begin{aligns}
    G[V_{\alpha_\C}^\epsilon] & \approx 3 \log \frac{1}{\epsilon}, \\
    G[cV_{\alpha^2_\C-1}^\epsilon] & \approx 16 G[V^{\epsilon}_{\alpha^2_\C-1}] \approx 48 \log \frac{1}{\epsilon}, \\
    G[c^{W+1}V_{\alpha^2_\C-1}^\epsilon] & \approx 14W + G[cV_{\alpha^2_\C-1}^\epsilon] \approx 14W + 48 \log \frac{1}{\epsilon}.
\end{aligns}
Then, from Eq.~\eqref{circ:Delta}, we have
\begin{align}
    G[cU_\Delta] &\approx 16 G[\add(\log(T^\star+1) + W+ 1)]\approx 64(\log T^\star + W +1).
\end{align}
Next, from Eqs.~\eqref{circ:S} and \eqref{circ:SS}, we get
\begin{align}
    G[c^2\S] & = N_C G[c^2S] = 2D N_C G[c^2\add(\log N_x)]
    \approx 2D N_C (14 + 16G[\add(\log N_x)] ) 
     \approx 2D N_C (14+ 64\log N_x ).
\end{align}
Putting it together, we get
\begin{align}
    G[U^{4\epsilon+\epsilon_\C}_{A_F}] \approx & G[c^2U^{\epsilon_\C}_\C]+78W+64+102\log\frac{1}{\epsilon} + 64\log T^\star +2DN_C( 14+64\log N_x).
\end{align}
Moreover, recalling the choice of simulation parameters, $N_x$ and $T^\star$, from Eqs.~\eqref{eq:parameter_choice1}-\eqref{eq:parameter_choice2} and restricting to the case $\beta=3/4$ (allowing one to resolve the Kolmogorov microscale), the above can be expressed as a function of the Reynolds number:
\begin{align}
    G[U_{A_F}^{4\epsilon+\epsilon_\C}] \approx & G[c^2U^{\epsilon_\C}_\C]+78 W + 28 D N_C +64+102\log\frac{1}{\epsilon} + 24(D(1+4 N_C)+2)\log \re .
    \label{eq:UAF_cost}
\end{align}


\subsection{Gate cost of the Carleman collision circuit}
\label{app:gates_C}

We are now facing the central problem of estimating the $T$-cost of $c^2U_\C$. First, let us introduce some shorthand notation:
\begin{equation}
    n_B=\lceil\log \lfloor N_C/2\rfloor \rceil,~~ n_C = \lceil\log N_C \rceil,~~ n_\Pi = \left\lceil \log\binom{k}{l}\right\rceil.
\end{equation}
From Eq.~\eqref{circ:C} we have
\begin{equation}
    G[c^2U_\C] = 2G[V_B] + \sum_{l=0}^{\lfloor\frac{N_C}{2}\rfloor} G[c^{2+n_B}U_{B_l}].
\end{equation}
Since $V_B$ is an $n_B$-qubit unitary, its approximation has the following $T$-cost according to Eq.~\eqref{eq:n_qubit_state_prep}:
\begin{equation}
    G[V_B^{N_C \epsilon}] \approx 3 N_C \log\frac{1}{\epsilon}.
\end{equation}
Then, from Eq.~\eqref{circ:B_l}, we see that
\begin{equation}
    \label{eq:B_l_cost}
    G[c^{2+n_B}U_{B_l}] = G[c^{2+n_B}U_{(\Delta^\dagger)^l}] + \sum_{k=l'}^{N_C-l} \left( G\left[c^{2+n_B+n_C}V_{C^k_{k+l}}\right] + G\left[c^{2+n_B+n_C}U_{C^k_{k+l}}\right]\right).
\end{equation}
The first term can be easily estimated to be
\begin{equation}
    G[c^{2+n_B}U_{(\Delta^\dagger)^l}] \approx 14+14n_B+G[c(\Delta^\dagger)^l]  \approx 14+14n_B+16G[\add(n_C)] \approx 14+14n_B + 64 n_C.
\end{equation}
As $V_{C^k_{k+l}}$ is a single qubit unitary, we get the second term via
\begin{equation}
    G\left[c^{2+n_B+n_C}V_{C^k_{k+l}}^{\epsilon}\right] \approx 14+14(n_B+n_C) + G\left[cV_{C^k_{k+l}}^{\epsilon}\right] \approx 14+14(n_B+n_C) + 48\log\frac{1}{\epsilon}.
\end{equation}

Turning to the last term of Eq.~\eqref{eq:B_l_cost}, we use Eq.~\eqref{circ:Ckl} to get
\begin{equation}
    G[c^{2+n_B+n_C} U_{C^k_{k+l}}] = 2G[V_\Pi] + 2 \sum_{i=1}^{\binom{k}{l}} G[c^{2+n_B+n_C+n_\Pi} \Pi_i] + G[c^{2+n_B+n_C}U_{F(k,l)}].
\end{equation}
Since $V_\Pi$ is a uniform state preparation, we ignore its cost, as it will be negligible. We estimate the cost of each multiply-controlled $\Pi_i$ as
\begin{equation}
    G[c^{2+n_B+n_C+n_\Pi} \Pi_i] \approx 14 +14(n_B+n_C+n_\Pi) + G[c\Pi_i] \approx 14 +14(n_B+n_C+n_\Pi),
\end{equation}
where we made a crude assumption that controlled permutations have zero $T$-cost. The cost of $c^{2+n_B+n_C}F(k,l)$ can be approximated by
\begin{align}
    G[c^{2+n_B+n_C}U_{F(k,l)}] & \approx 14+14(n_B+n_C) + G[cU_{F(k,l)}] \approx 14+14(n_B+n_C) + 16 G[U_{F(k,l)}] \\
    & \approx 14+14(n_B+n_C) + 16(k-l) G[U_{I+F_1}] + 16 l G[U_{\bar{F}_2}]. 
\end{align}

Putting things together, we get:
\begin{align}
    G[c^2U_\C^{\epsilon_\C}] = &~ 6 N_C \log\frac{1}{\epsilon} + \sum_{l=0}^{\lfloor \frac{N_C}{2}\rfloor} (14+14n_B+64n_C) \\
    & +\sum_{l=0}^{\lfloor \frac{N_C}{2}\rfloor} \sum_{k=l'}^{ N_C-l} \left( 14 + 14(n_B+n_C) + 48 \log\frac{1}{\epsilon} \right) \\
    & +2\sum_{l=0}^{\lfloor\frac{N_C}{2}\rfloor} \sum_{k=l'}^{N_C-l} \sum_{i=1}^{\binom{k}{l}} \left( 14+14(n_B+n_C+n_\Pi)\right)\\
    & +\sum_{l=0}^{\lfloor\frac{N_C}{2}\rfloor} \sum_{k=l'}^{N_C-l} \left( 14+14(n_B+n_C) + 16(k-l) G[U_{I+F_1}^{\epsilon_{F_1}}] + 16 l  G[U_{\bar{F}_2}^{\epsilon_{F_2}}]\right),
\end{align}
where $\epsilon_\C$ is just the sum over errors of all $\epsilon$-approximate single-qubit gates used so far, together with the errors introduced while compiling $U_{I+{F}_1}$ and $U_{\bar{F}_2}$ (we will discuss these in the next section):
\begin{align}
    \epsilon_\C &= 2N_C\epsilon + \sum_{l=0}^{\lfloor \frac{N_C}{2}\rfloor} \sum_{k=l'}^{ N_C-l} [\epsilon + (k-l)\epsilon_{F_1}+ l \epsilon_{F_2} ] \\
    &= \frac{N_C\left[ 6(12+N_C)\epsilon +(N_C+2)((2N_C+5)\epsilon_{F_1} + (N_C+1)\epsilon_{F_2}) \right]}{24}.
\end{align}
Most of the sums can be explicitly calculated to yield 
\begin{align}
    \label{eq:UC_cost}
    G[c^2U_\C^{\epsilon_\C}] \approx &~ 6N_C \left(n_B+1+\log\frac{1}{\epsilon}\right) +  \frac{N_C+2}{2} (14+14n_B+64n_C) \\
    & + \frac{N_C^2+4N_C}{4} \left( 14 + 14(n_B+n_C) + 48 \log\frac{1}{\epsilon} \right) \\
    & +\left[\left(2+\frac{4}{\sqrt{5}}\right)\left(\frac{1+\sqrt{5}}{2}\right)^{N_C}-4]\right]\left( 14+14(n_B+n_C) \right)+2h_{N_C}\\
    &  +\frac{N_C^2+4N_C}{4} ( 14+14(n_B+n_C))\\
    & + \frac{2N_C(N_C+2  )}{3} \left[(2N_C+5) G[U_{I+F_1}^{\epsilon_{F_1}}] + (N_C+1)  G[U_{\bar{F}_2}^{\epsilon_{F_2}}] \right],
\end{align}
where we used
\begin{equation}
    \sum_{l=0}^{\lfloor\frac{N_C}{2}\rfloor} \sum_{k=l'}^{N_C-l} \binom{k}{l} = \left(1+\frac{2}{\sqrt{5}}\right)\left(\frac{1+\sqrt{5}}{2}\right)^{N_C}-2,
\end{equation}
and introduced
\begin{equation}
    h_{N_C} := \sum_{l=0}^{\lfloor\frac{N_C}{2}\rfloor} \sum_{k=l'}^{N_C-l} \binom{k}{l}\log\binom{k}{l}.
\end{equation}

Combining Eqs.~\eqref{eq:UAF_cost} and \eqref{eq:UC_cost}, we arrive at the final $T$-gate cost of approximating $U_{A_F}$ up to error $\epsilon_{\mathrm{{tot}}}$ given in the main text it Eq.~\eqref{eq:gate_final_lastline}.


\subsection{Gate cost of the collision circuits }
\label{app:gates_F}

To estimate the $T$ gate cost of compiling $U_{I+F_1}$ and $U_{\bar{F}_2}$, we will consider separately the cases of dimension $D=1$ and $D=2$, and leave $D=3$ case for future work.

\paragraph*{Case $D=1$.}

Performing the SVD of $I+\tilde{F}_1=L_1\Sigma_1R_1^\dagger$, one finds that not only $\Sigma_1$ depends on $\taus$, but also the unitaries $L_1$ and $R_1$ do. However, independently of the specific value of $\taus$, the columns of $L_1$ and $R_1$ have a very specific structure. Namely, up to permutations, we have: 
\begin{aligns}
    \mathrm{2~columns~with~pattern:} &\quad (a,a,b) \\
     \mathrm{1~column~with~pattern:} &\quad(1/\sqrt{2},-1/\sqrt{2},0),
\end{aligns}
where $a,b$ are real numbers. Thus, we can estimate the $T$-cost of $L_1$ and $R_1$ using Eq.~\eqref{eq:unitarysynthesiscost} and the costs of state preparation for the above states. Clearly, the last state can be prepared with Clifford gates only. Moreover, the first two states each require just one general single-qubit rotation and one controlled Hadamard gate. Hence:
\begin{equation}
    G[L^{2\sqrt{2}\epsilon}_1] \approx 32\left(2\cdot(G[R_y^{\epsilon}(\theta)]+2)]\right) +12 \approx 140 + 192\log\frac{1}{\epsilon},
\end{equation}
and the same cost for $R_1$. 

Next, $\Sigma_1$ can be found to be
\begin{equation}
    \Sigma_1 = \mathrm{diag}\left(1,\sigma_1,\sigma_2\right).
\end{equation}
From Eq.~\eqref{circ:F1}, we see that this then reduces to three single-qubit unitaries, $V_{\sigma_1}$, $V_{\sigma_2}$, and $V_{\sigma_3}$, controlled on two qubits. By noticing that $V_{\sigma_1}=I$ and that encoding never uses the basis state $\ket{11}$, it can be reduced to two single-qubit unitaries controlled on a single-qubit. Thus, we get 
\begin{equation}
    G[U_{\Sigma_1}^{2\epsilon}] = 2G[cR^\epsilon_y(\theta)] \approx 96 \log\frac{1}{\epsilon}.
\end{equation}
Putting it together with the costs for $L_1$ and $R_1$, we get:
\begin{equation}
    G[U_{I+F_1}^{\epsilon_{F_1}}] \approx 280 + 480 \log\frac{1}{\epsilon},
\end{equation}
with
\begin{equation}
    \epsilon_{F_1}=(4\sqrt{2}+2)\epsilon.
\end{equation}

Now, we find the following HOSVD of $\tilde{F}_2=L_2\Sigma_2(R_2\ot R_2)^\dagger$:
\begin{align}
    L_2 = \begin{pmatrix}
        \frac{1}{\sqrt{6}} & \frac{1}{\sqrt{2}} & \frac{1}{\sqrt{3}} \\
        -\frac{\sqrt{2}}{\sqrt{3}} & 0 & \frac{1}{\sqrt{3}} \\
        \frac{1}{\sqrt{6}} & -\frac{1}{\sqrt{2}} & \frac{1}{\sqrt{3}}
    \end{pmatrix},\quad
    R_2 = \begin{pmatrix}
        \frac{1}{\sqrt{2}} & 0 & \frac{1}{\sqrt{2}} \\
        0 & 1 & 0 \\
        -\frac{1}{\sqrt{2}} & 0 & \frac{1}{\sqrt{2}}
    \end{pmatrix},
\end{align}
and $\Sigma_2$ has a single non-zero element $(\Sigma_2)_{11}=\sqrt{6}/\taus$. Thus, it is actually an SVD and $(\Sigma_2)_{11}$ is the largest singular value of $\tilde{F}_2$. The right unitary $R_2$ can be compiled using only Clifford gates, whereas the $T$-cost of the left can be estimated using Eq.~\eqref{eq:unitarysynthesiscost} to be
\begin{equation}
    G[L_2^{\sqrt{2}(\epsilon_1+\epsilon_2+\epsilon_3)}] \approx 32 \sum_{j=1}^{3}G[L^{\epsilon_j}_{2,j}] + 12,
\end{equation}
where $L_{2,j}$ is a unitary state preparation of the $j$-th column of $L_2$. Since the second column can be prepared using Clifford gates, whereas the first and third column require each just one single-qubit gate and one controlled Hadamard gate, we conclude that
\begin{equation}
    G[L^{2\sqrt{2}\epsilon}_2] \approx 64 (G[R_{y}^{\epsilon}(\theta)]+2) + 12 \approx 192\log\frac{1}{\epsilon}+140.
\end{equation}
Moreover, block-encoding of $\Sigma_2$ requires one single-qubit unitary gate controlled on two qubits, meaning that
\begin{equation}
    G[U^\epsilon_{\Sigma_2}] \approx G[c^2R_y^\epsilon(\theta)] \approx 14 + 48\log\frac{1}{\epsilon}.     
\end{equation}
We thus conclude that
\begin{align}
    G[U_{{F}_2}^{\epsilon_{F_2}}] & \approx 154+240 \log\frac{1}{\epsilon}
\end{align}
where
\begin{equation}
    \epsilon_{F_2} = (1+2\sqrt{2})\epsilon.
\end{equation}

\paragraph*{Case $D=2$.}

Performing the SVD of $I+\tilde{F}_1=L_1\Sigma_1R_1^\dagger$, one finds that, up to permutations, the columns of $L_1$ and $R_1$ have the following structure:
\begin{aligns}
    \mathrm{2~columns~with~pattern:} &\quad (a,a,a,a,b,b,b,b,c) \\
    \mathrm{4~columns~with~pattern:} &\quad(a,-a,b,-b,c,-c,d,-d,0), \\
     \mathrm{3~column~with~pattern:} &\quad(a,a,b,b,c,c,d,d,e),
\end{aligns}
with $a,b,c,d,e$ real. As before, we can estimate the $T$-cost of $L_1$ and $R_1$ using Eq.~\eqref{eq:unitarysynthesiscost} and the costs of state preparation for the above states. The first two states can be prepared with a single qubit rotation, followed by a single qubit rotation controlled on a single qubit, followed by two controlled Hadamard gates. The next four states can be prepared with a single qubit rotation, followed by two single qubit rotations controlled on a single qubit, followed by Clifford gates. Finally, the last three states can be obtained by a single qubit rotation, followed by a single qubit rotation controlled on a single qubit, followed by two single qubit rotations controlled on two qubits, followed by a controlled Hadamard. Hence:
\begin{align}
    G[L^{28\sqrt{2}\epsilon}_1] \approx & 32\left(2(G[R_y^{\epsilon}(\theta)]+G[cR_y^{\epsilon}(\theta)]+4)+4(G[R_y^{\epsilon}(\theta)]+2G[cR^{\epsilon}_y(\theta)]) \right.\\
    &\left. + 3 (G[R_y^{\epsilon}(\theta)] + G[cR_y^{\epsilon}(\theta))] + 2G[c^2R_y^{\epsilon}(\theta)]+2) \right) +180 \\
    \approx & 3316 + 30048\log\frac{1}{\epsilon},
\end{align}
and the same cost for $R_1$. 

Next, $\Sigma_1$ can be found to be
\begin{equation}
    \Sigma_1 = \mathrm{diag}\left(\sigma_1,\sigma_1,\sigma_1,\sigma_2,\sigma_2,\sigma_3,\sigma_3,\sigma_4,\sigma_5\right),
\end{equation}
From Eq.~\eqref{circ:F1}, we see that this then can be achieved by nine single-qubit unitaries, $V_{\sigma_1},\dots,V_{\sigma_9}$, controlled on four qubits. By noticing that the first three rotations are the same, and there are also two equal pairs, this can be reduced to one uncontrolled single-qubit unitary rotation, two single-qubit unitaries controlled on a three qubits, and two single-qubit unitaries controlled on four qubits. Thus, we get
\begin{equation}
    G[U_{\Sigma_1}^{5\epsilon}] = G[R^{\epsilon}_y(\theta)] + 2G[c^3R^\epsilon_y(\theta)]+ 2G[c^4R^\epsilon_y(\theta)] \approx 140+195 \log\frac{1}{\epsilon}.
\end{equation}
Putting it together with the costs for $L_1$ and $R_1$, we get:
\begin{equation}
    G[U_{I+F_1}^{\epsilon_{F_1}}] \approx 6412 + 60291 \log\frac{1}{\epsilon},
\end{equation}
with
\begin{equation}
    \epsilon_{F_1}=(28\sqrt{2}+5)\epsilon.
\end{equation}

Next, we find the following HOSVD of $\tilde{F}_2=L_2\Sigma_2(R_2\ot R_2)^\dagger$:
\small
\begin{align}
    L_2 = \left(
\begin{array}{ccccccccc}
 -\frac{1}{3} \left(2 \sqrt{2}\right) & 0 & 0 & 0 & \frac{1}{5} & \frac{2 \sqrt{\frac{2}{17}}}{5} & \frac{4}{\sqrt{595}} & \frac{4}{\sqrt{1855}} & \frac{2 \sqrt{\frac{2}{53}}}{3} \\
 \frac{1}{6 \sqrt{2}} & -\frac{1}{2} & 0 & 0 & 0 & 0 & \sqrt{\frac{17}{35}} & \frac{17}{\sqrt{1855}} & -\frac{19}{6 \sqrt{106}} \\
 \frac{1}{6 \sqrt{2}} & -\frac{1}{2} & 0 & 0 & 0 & 0 & 0 & 0 & \frac{\sqrt{\frac{53}{2}}}{6} \\
 \frac{1}{6 \sqrt{2}} & \frac{1}{2} & 0 & 0 & 0 & 0 & 0 & \sqrt{\frac{35}{53}} & \frac{17}{6 \sqrt{106}} \\
 \frac{1}{6 \sqrt{2}} & \frac{1}{2} & 0 & 0 & 0 & 0 & \sqrt{\frac{17}{35}} & -\frac{18}{\sqrt{1855}} & \frac{17}{6 \sqrt{106}} \\
 \frac{1}{6 \sqrt{2}} & 0 & \frac{1}{2} & -\frac{1}{\sqrt{2}} & \frac{2}{5} & \frac{4 \sqrt{\frac{2}{17}}}{5} & -\frac{1}{2 \sqrt{595}} & -\frac{1}{2 \sqrt{1855}} & -\frac{1}{6 \sqrt{106}} \\
 \frac{1}{6 \sqrt{2}} & 0 & -\frac{1}{2} & 0 & 0 & \frac{5}{\sqrt{34}} & -\frac{1}{2 \sqrt{595}} & -\frac{1}{2 \sqrt{1855}} & -\frac{1}{6 \sqrt{106}} \\
 \frac{1}{6 \sqrt{2}} & 0 & -\frac{1}{2} & 0 & \frac{4}{5} & -\frac{9}{5 \sqrt{34}} & -\frac{1}{2 \sqrt{595}} & -\frac{1}{2 \sqrt{1855}} & -\frac{1}{6 \sqrt{106}} \\
 \frac{1}{6 \sqrt{2}} & 0 & \frac{1}{2} & \frac{1}{\sqrt{2}} & \frac{2}{5} & \frac{4 \sqrt{\frac{2}{17}}}{5} & -\frac{1}{2 \sqrt{595}} & -\frac{1}{2 \sqrt{1855}} & -\frac{1}{6 \sqrt{106}} \\
\end{array}
\right),
\end{align}
\normalsize
\small
\begin{align}
  R_2 = \left(
\begin{array}{ccccccccc}
 0 & -\frac{1}{2 \sqrt{3}} & \frac{1}{2 \sqrt{3}} & -\frac{1}{2 \sqrt{3}} & \frac{1}{2 \sqrt{3}} & -\frac{1}{\sqrt{3}} & 0 & 0 & \frac{1}{\sqrt{3}} \\
 0 & -\frac{1}{2 \sqrt{3}} & \frac{1}{2 \sqrt{3}} & \frac{1}{2 \sqrt{3}} & -\frac{1}{2 \sqrt{3}} & 0 & -\frac{1}{\sqrt{3}} & \frac{1}{\sqrt{3}} & 0 \\
 0 & \frac{1}{\sqrt{3}} & 0 & \frac{1}{\sqrt{3}} & 0 & 0 & 0 & 0 & \frac{1}{\sqrt{3}} \\
 0 & \frac{1}{\sqrt{3}} & 0 & -\frac{1}{\sqrt{3}} & 0 & 0 & 0 & \frac{1}{\sqrt{3}} & 0 \\
 0 & -\frac{1}{\sqrt{15}} & 0 & \frac{1}{\sqrt{15}} & 0 & 0 & \sqrt{\frac{3}{5}} & \frac{2}{\sqrt{15}} & 0 \\
 0 & -\frac{1}{\sqrt{15}} & 0 & -\frac{1}{\sqrt{15}} & 0 & \sqrt{\frac{3}{5}} & 0 & 0 & \frac{2}{\sqrt{15}} \\
 0 & 0 & 0 & \frac{1}{\sqrt{30}} & \sqrt{\frac{5}{6}} & \frac{1}{\sqrt{30}} & -\frac{1}{\sqrt{30}} & \frac{1}{\sqrt{30}} & -\frac{1}{\sqrt{30}} \\
 0 & \frac{1}{\sqrt{30}} & \sqrt{\frac{5}{6}} & 0 & 0 & \frac{1}{\sqrt{30}} & \frac{1}{\sqrt{30}} & -\frac{1}{\sqrt{30}} & -\frac{1}{\sqrt{30}} \\
 1 & 0 & 0 & 0 & 0 & 0 & 0 & 0 & 0 \\
\end{array}
\right),
\end{align}
\normalsize
and $\Sigma_2$ is a rectangular matrix with only nonzero elements being:
\begin{align}
    (\Sigma_2)_{1,1} = (\Sigma_2)_{1,11}=\frac{3 \sqrt{3}}{\taus}, \quad  (\Sigma_2)_{2,2} =  (\Sigma_2)_{1,10} = -\frac{3}{\taus}, \quad (\Sigma_2)_{3,1} = \frac{3}{2\taus} = - (\Sigma_2)_{3,11}. 
\end{align}

We then need to estimate the cost of the state preparations corresponding to each column of $L_2$. We summarize the results in Table~\ref{tab:L2columns}, where $\mathrm{USP}_n$ denotes a unitary state preparation over $n$ states. Using Eq.~\eqref{eq:unitarysynthesiscost}, it follows that the state preparation cost for the $n=4$ qubit unitary $L_2$ is
\begin{align}
  G[L_2^{19\sqrt{2}\epsilon}] =  32 \times 57 \log\frac{1}{\epsilon} + 9(8 \times 4 - 12) = 180+1824 \log\frac{1}{\epsilon}.
\end{align}
Now, let us do the same for $R_2$ (for convenience we analyze $R_2^\dag$). We summarize the results in Table~\ref{tab:R2columns}. The overall cost for $R_2 \otimes R_2$ is then
\begin{align}
    G[R_2^{14\sqrt{2}\epsilon} \otimes R_2^{14\sqrt{2}\epsilon}] = 2 \left( 32 \times 48 \log\frac{1}{\epsilon} + 9(8 \times 4 -12)\right) = 360 +3072\log\frac{1}{\epsilon}. 
\end{align}

\begin{table}[b!]
    \centering
    \renewcommand{\arraystretch}{1.35}
    \begin{tabular}{|c|c|c|}
    \hline 
        {\bf Column} & {\bf Operations} & {\bf Cost} \\
        \hline 
       1  & 1 single-qubit $y$-rotation, 1 $\mathrm{USP}_8$   & $3\log\frac{1}{\epsilon}$ \\  \hline 
        2 & $\mathrm{USP}_4$, $(-1)$ phases & 0 \\    \hline 
        3 & $\mathrm{USP}_4$, $(-1)$ phases & 0 \\    \hline 
        4 & $\mathrm{USP}_2$, $(-1)$ phases & 0 \\    \hline 
        5 & 3 single-qubit $y$-rotations & $9 \log\frac{1}{\epsilon}$\\    \hline 
        6 &  4 single-qubit $y$-rotations & $12 \log\frac{1}{\epsilon}$\\    \hline 
        7 & 3 single-qubit $y$-rotations, $\mathrm{USP}_4$ & $9 \log\frac{1}{\epsilon}$ \\    \hline 
        8 & 4 single-qubit $y$-rotations, $\mathrm{USP}_4$ & $12\log\frac{1}{\epsilon}$ \\    \hline 
        9 & 4 single-qubit $y$-rotations, $\mathrm{USP}_2$  & $12\log\frac{1}{\epsilon}$ \\    \hline \hline
        Total & 19 single-qubit $y$-rotations & $57 \log\frac{1}{\epsilon}$ \\ \hline
    \end{tabular}
    \caption{State preparation costs for the columns of $L_2$. Here $\epsilon$ is the per-gate error.}
    \label{tab:L2columns}
\end{table}

\begin{table}[h]
    \centering
    \renewcommand{\arraystretch}{1.35}    
    \begin{tabular}{|c|c|c|}
    \hline 
        {\bf Column} & {\bf Operations} & {\bf Cost} \\
        \hline 
       1  & 2 single-qubit $y$-rotations, 1 $\mathrm{USP}_4$   & $6\log\frac{1}{\epsilon}$ \\  \hline 
        2 & 2 single-qubit $y$-rotations, 1 $\mathrm{USP}_4$, $(-1)$ phases & $6\log\frac{1}{\epsilon}$ \\    \hline 
        3 & $\mathrm{USP}_3$  & $3\log\frac{1}{\epsilon}$ \\    \hline 
        4 & $\mathrm{USP}_3$, $(-1)$ phases & $3\log\frac{1}{\epsilon}$ \\    \hline 
        5 & 2 single-qubit $y$-rotations, Hadamard & $6 \log\frac{1}{\epsilon}$\\    \hline 
        6 &  2 single-qubit $y$-rotations, Hadamard & $6 \log\frac{1}{\epsilon}$\\    \hline 
        7 & 3 single-qubit $y$-rotations, $\mathrm{USP}_4$ & $9 \log\frac{1}{\epsilon}$ \\    \hline 
        8 & 3 single-qubit $y$-rotations, $\mathrm{USP}_4$ & $9 \log\frac{1}{\epsilon}$ \\    \hline 
        9 & \textrm{permutation} & 0 \\    \hline \hline
        Total & 14 single-qubit $y$-rotations, 2 $\mathrm{USP}_3$ & $48 \log\frac{1}{\epsilon}$ \\ \hline
    \end{tabular}
    \caption{State preparation costs for the columns of $R^\dag_2$. We cost $\textrm{USP}_3$ as $3 \log_3(1/\epsilon)$, as it can be realized via a single qubit $y$ rotation followed by a controlled Hadamard. Here $\epsilon$ is the per-gate error.}
    \label{tab:R2columns}
\end{table}

We still need to block-encode $\Sigma_2$. To achieve this, consider the following $90 \times 90$ unitary:
\begin{align}
    U_{\Sigma_2} = \begin{bmatrix}
        \Sigma_2 & \sqrt{I_9-\Sigma_2 \Sigma_2^\dagger}  \\
        \sqrt{I_{81}-\Sigma_2^\dagger \Sigma_2} & -\Sigma_2^\dagger.
    \end{bmatrix}
\end{align}
By direct inspection, $U_{\Sigma_2}$ is an optimal block-encoding of $\Sigma_2$ (the block-encoding prefactor is $6$, which equals $\|\Sigma_2\|$). For simplicity, let us consider the matrix $U^\dag_{\Sigma_2}$. Now, we need to estimate the cost of the state preparations corresponding to each column. We summarize the results in Table~\ref{tab:Sigma2columns}. The overall cost for $U_{\Sigma_2}$ is then:
\begin{align}
    G[U_{\Sigma_2}^{8\sqrt{2}\epsilon}] = 32\times  24 \log\frac{1}{\epsilon} + 90(8 \times 7 -12) = 3960 + 768 \log \frac{1}{\epsilon}. 
\end{align}

\begin{table}[h]
    \centering
    \renewcommand{\arraystretch}{1.35}
    \begin{tabular}{|c|c|c|}
    \hline 
        {\bf Column} & {\bf Operations} & {\bf Cost} \\
        \hline 
       1  & Hadamard   & 0 \\  \hline 
        2 & 2 Hadamards, $(-1)$ phases & 0 \\    \hline 
        3 & 1 single-qubit $y$-rotation, Hamadard, $(-1)$ phases & $3\log\frac{1}{\epsilon}$ \\    \hline 
        4-9 & $\mathrm{permutation}$ & 0 \\    \hline 
        10 & 2 single-qubit $y$-rotations, Hadamard, $(-1)$ phases & $6\log\frac{1}{\epsilon}$ \\   \hline
        11 &  2 single-qubit $y$-rotations, permutation, $(-1)$ phases & $6\log\frac{1}{\epsilon}$ \\    \hline 
        12-18 & permutation & 0 \\    \hline 
        19 & 2 single-qubit $y$-rotations, permutation, $(-1)$ phases & $6\log\frac{1}{\epsilon}$ \\    \hline 
        20 & 1 single-qubit $y$-rotation, permutation, $(-1)$ phases  & $3\log\frac{1}{\epsilon}$ \\    \hline
          21-90 & permutation  & 0 \\    \hline \hline
        Total & 8 single-qubit $y$-rotations & $24 \log\frac{1}{\epsilon}$ \\ \hline
    \end{tabular}
    \caption{State preparation costs for the columns of $U_{\Sigma_2}$. Here $\epsilon$ is the per-gate error.}
    \label{tab:Sigma2columns}
\end{table}

Combining the results, we thus conclude that $T$-cost of a block-encoding of $F_2$ is given by
\begin{align}
    G[U_{{F}_2}^{\epsilon_{F_2}}] = 4500+5664 \log \frac{1}{\epsilon},
\end{align}
with
\begin{equation}
    \epsilon_{F_2}=59\sqrt{2} \epsilon.
\end{equation}


\subsection{Extra details on gate compilation}
\label{app:gates_extra}

We start with explaining how to realize $c^kU$ using $cU$ and Toffoli gates. Given control qubits $c_1,\dots, c_k$, the auxiliary qubits $a_1,\dots, a_{k-1}$ are prepared in $\ket{0}$ and are then in turn used to iteratively record $\mathrm{AND}(c_1,c_2)$, $\mathrm{AND}(c_3,a_1) \dots \mathrm{AND}(c_k,a_{k-2})$. A $cU$ gate is then applied to the target $t$ controlled on $a_{k-1}$, and then we uncompute the auxiliary qubits. The explicit circuit for $k=5$ is as follows:
\begin{equation}
   c^5U= \begin{quantikz}[row sep={0.6cm,between origins}, baseline=(current bounding box.center),column sep=1.2em]
\lstick{$c_1$} & \ctrl{5} & \qw      & \qw      & \qw      & \qw      & \qw      & \qw    & \qw      & \ctrl{5} & \qw \\ 
\lstick{$c_2$} & \ctrl{4} & \qw      & \qw      & \qw      & \qw      & \qw     &\qw      & \qw      & \ctrl{4} & \qw \\
\lstick{$c_3$} & \qw      & \ctrl{4} & \qw      & \qw      & \qw      & \qw     & \qw     & \ctrl{4} & \qw      & \qw \\
\lstick{$c_4$} & \qw      & \qw      & \ctrl{4} & \qw      & \qw      & \qw     & \ctrl{4}  & \qw      & \qw      & \qw \\
\lstick{$c_5$} & \qw      & \qw      & \qw      & \ctrl{4} & \qw      & \ctrl{4} & \qw    & \qw       & \qw      & \qw \\
\lstick{$\ket{0}$} & \targ{}    & \ctrl{1} & \qw      & \qw      & \qw      & \qw      & \qw    & \ctrl{1} & \targ{}    & \qw \\
\lstick{$\ket{0}$} & \qw      & \targ{}    & \ctrl{1} & \qw      & \qw      & \qw & \ctrl{1}     &\targ{}    & \qw      & \qw \\
\lstick{$\ket{0}$} & \qw      & \qw      & \targ{}    & \ctrl{1} & \qw      & \ctrl{1}  & \targ{}  & \qw      & \qw      & \qw \\
\lstick{$\ket{0}$} & \qw      & \qw      & \qw      & \targ{}    & \ctrl{1} & \targ{}    & \qw & \qw      & \qw      & \qw \\
\lstick{$t$}   & \qw      & \qw      & \qw      & \qw      & \gate{U} & \qw      & \qw   &\qw    & \qw      & \qw
\end{quantikz}
\end{equation}
This uses $2(k-1) = 2k-2$ Toffoli gates for the auxiliary qubits, and a single $cU$. 

We now proceed to deriving a $T$-cost for general unitary synthesis based on state preparations. In Ref.~\cite{kliuchnikov2013synthesis}, it is shown that with a single auxiliary qubit we can embed any unitary $U$ as $\tilde{U}:=\ketbra{0}{1} \otimes U + \ketbra{1}{0}\otimes U^\dagger$, and so we can simulate $U$ by $\tilde{U} \ket{1}\otimes \ket{\psi} = \ket{0} \otimes U \ket{\psi}$. Then, it can be shown that
\begin{equation}
    \tilde{U} = \prod_{j=1}^{2^n} ( I - 2 \ketbra{w_j}{w_j}),
\end{equation}
where 
\begin{equation}
 \ket{w_j} := \frac{1}{\sqrt{2}} \left(\ket{1}\otimes \ket{j} - \ket{0}\otimes \ket{u_j}\right),
\end{equation}
 with $\ket{u_j}$ denoting the $j$-th column of $U$. Each reflection term $W_j:=I- 2 \ketbra{w_j}{w_j}$ can be approximately realized using an $\epsilon_j$--approximate state preparation of $\ket{w_j}$ as
 \begin{equation}
   W_j^{\epsilon_j''} = \mathrm{PREP}^{\epsilon_j'}(w_j) ( I - 2\ketbra{0}{0} ) \mathrm{PREP}^{\epsilon_j'}(w_j)^\dagger,
 \end{equation}
 where we assume $\mathrm{PREP}^{\epsilon_j'}(w_j)$ is a unitary that prepares the state $\ket{w_j}$ from $\ket{0^{n+1}} = \ket{0}^{\otimes n+1}$ with error $\epsilon_j'
 $. The preparation of $\ket{w_j}$ can be realized with a single controlled use of $\mathrm{PREP}^{\epsilon_j}(u_j)$, and so the total cost for generating $W_j^{\epsilon_j''}$ is
 \begin{equation}
     G[W_j^{\epsilon_j''} ] =2G[\mathrm{cPREP}_{\epsilon_j} (u_j)] + G\left[I -\ketbra{0^{n+1}}{0^{n+1}}\right].
 \end{equation}

We now provide error analysis for this step. The unitary $\mathrm{PREP}^{\epsilon_j}(u_j)$ obeys $\|\mathrm{PREP}^{\epsilon_j}(u_j)\ket{0}^{\ot n} - \ket{u_j}\| \le \epsilon_j$. This implies that 
\begin{equation}
    \left\|\frac{1}{\sqrt{2}}(\ket{1}\otimes \ket{j} - \ket{0}\otimes \mathrm{PREP}^{\epsilon_j}(u_j) \ket{0}) - \ket{w_j}\right\| \le \frac{\epsilon_j}{\sqrt{2}} = \epsilon_j'.
\end{equation}
Moreover, it is readily checked that
\begin{equation}
    \left \| \mathrm{PREP}^{\epsilon_j'}(w_j) \ketbra{0^{n+1}}{0^{n+1}} \mathrm{PREP}^{\epsilon_j'}(w_j)^\dagger - \ketbra{w_j}{w_j} \right \| \le \epsilon_j' \sqrt{1-\epsilon_j^{'2}/4} \le \epsilon_j'.
\end{equation}
This implies
\begin{equation}
    \|\mathrm{PREP}^{\epsilon_j'}(w_j) ( I - 2\ketbra{0}{0} ) \mathrm{PREP}^{\epsilon_j'}(w_j)^\dagger - (I -2 \ketbra{w_j}{w_j}) \| \le 2\epsilon_j' = \sqrt{2}\epsilon_j.
\end{equation}
Thus, $\epsilon_j'' = \sqrt{2}\epsilon_j$.
 
To generate $I -2\ketbra{0^{n+1}}{0^{n+1}}$, we first note that this is a unitary $\mathrm{diag}(-1,1,1,\dots 1)$, which can be viewed as a $c^n(-Z)$ unitary, controlled on $n$ $0$'s. Up to Clifford gates, this is equivalent to $c^nX$, and so this provides a $T$-count of $G[c^nX]$. The total $T$-count of unitary synthesis from this method is then
\begin{equation}
       G [U^{\epsilon}] =  \sum_{j=1}^{2^n}(2G[\mathrm{cPREP}_{\epsilon_j} (u_j)] + G[c^nX]) = 32 \sum_{j=1}^{2^n}G[\mathrm{PREP}^{\epsilon_j} (u_j)] + 2^n(8n -12),
\end{equation}
where again we use that $G[cU] \approx 16 G[U]$, the total error introduced is
\begin{equation}
    \epsilon = \sqrt{2}\sum_{j=1}^{2^n} \epsilon_j,
\end{equation}
and a single auxiliary qubit is used in the construction.

\newpage
    \nocite{REVTEX42Control}
\bibliographystyle{apsrev4-2}
\bibliography{references_APS} 

\end{document}